\documentclass[11pt,letterpaper]{article}
\usepackage{amsmath,amssymb,amsfonts,amsthm}
\usepackage{tikz}
\usepackage{tikz-cd}
\usepackage{enumitem}
\usepackage{bbm}
\usepackage[mathscr]{euscript}
\usepackage{thm-restate}
\usepackage{algpseudocode}
\usepackage{xcolor}
\usepackage{algorithm}
\usepackage[margin=1in]{geometry}
\usepackage[colorlinks, citecolor = blue, linkcolor = magenta]{hyperref}
\usepackage{cleveref}
\usepackage{float}
\usepackage{array}   

\usetikzlibrary{decorations.markings}

\DeclareMathOperator*{\E}{\mathbb{E}}
\newcommand{\sketch}{\boldsymbol{s}}
\renewcommand{\d}{\,\mathrm{d}}
\newcommand{\supp}{\operatorname{supp}}

\numberwithin{equation}{section}
\newtheorem{theorem}{Theorem}[section]
\newtheorem{lemma}[theorem]{Lemma}
\theoremstyle{definition}

\newtheorem{proposition}[theorem]{Proposition}
\newtheorem{corollary}[theorem]{Corollary}
\newtheorem{definition}[theorem]{Definition}
\newtheorem{claim}{Claim}

\DeclareMathOperator{\dist}{dist}

\newcommand{\eps}{\varepsilon}

\newcolumntype{L}{>{$}l<{$}} 

\newcommand{\ALG}{\textsf{ALG}}

\newcommand{\calD}{\mathcal{D}}

\newcommand{\calS}{\mathcal{S}}

\newcommand{\calQ}{\mathcal{Q}}

\def\calD{\mathcal{D}} 
\def\calS{\mathcal{S}}
\def\FindRep{\textsc{FindRep}}
\def\frak{\mathfrak}
\def\Sketch{\textsc{Sketch}}
\def\sketch{\calS}

\def\yes{\mathsf{yes}} 
\def\no{\mathsf{no}}

\author{Xi Chen\thanks{Xi Chen is supported in part by NSF grants IIS-1838154, CCF-2106429, and CCF-2107187.} \\ Columbia University\\ \normalsize \texttt{xc2198@columbia.edu}
\and 
Yumou Fei \\ Peking University \\ \normalsize \texttt{feiym2002@stu.pku.edu.cn}
\and
Shyamal Patel\thanks{Shyamal Patel is supported in part by NSF grants IIS-1838154, CCF-2106429, CCF-2107187, CCF-2218677, ONR grant ONR-13533312, and an NSF Graduate Student Fellowship.} \\ Columbia University\\ \normalsize\texttt{shyamalpatelb@gmail.com}}
\date{}

\title{Distribution-Free Testing of Decision Lists\\ with a Sublinear Number of Queries\vspace{0.3cm}}

\begin{document}

\maketitle

\begin{abstract}
We give a distribution-free testing algorithm for 
  decision lists with $\tilde{O}(n^{11/12}/\eps^3)$ queries.
This is the first sublinear algorithm for this problem, which
  shows that, unlike halfspaces, testing is strictly easier than learning for decision lists.
Complementing the algorithm, we show that any distribution-free tester
  for decision lists must make $\tilde{\Omega}(\sqrt{n})$ queries, or draw
  $\tilde{\Omega}(n)$ samples when the algorithm is sample-based.
\end{abstract}

\section{Introduction}

A Boolean function $f:\{0,1\}^n\rightarrow \{0,1\}$ is called
  a \emph{decision list} (or $1$-decision list) if there exists~a list of pairs 
  $(\alpha_1,\beta_1),\ldots,(\alpha_k,\beta_k)$ where each $\alpha_i$
  is a literal and $\beta_i\in \{0,1\}$, such that~$f(x)$~is~set~to be  
  $\beta_j$ of the smallest index $j$ such that $\alpha_j$ is satisfied by $x$ and is set to
  be a default value $\beta_{k+1}\in \{0,1\}$ if no literal is satisfied.
Decision lists were first introduced by Rivest \cite{Rivest87},
  and have been one of the most well studied classes of Boolean functions 
  in computational learning theory.
In particular, 
the fundamental theorem of Statistical Learning \cite{MLbook}
  shows  that  the  VC dimension  of  a class essentially 
  captures the number of  samples needed for its PAC learning \cite{Valiant},
  which gives a tight bound of $\Theta(n)$ samples for learning 
  decision lists. Moreover, this bound is tight even if we give the learner query access to the underlying decision list \cite{turan1993lower}.

In this paper, we study the \emph{distribution-free testing} of decision lists, where
  the goal of  a tester is to determine 
  whether an unknown  function $f:\{0,1\}^n\rightarrow \{0,1\}$ is a decision list  or $\eps$-far from decision lists with respect  to an 
  unknown  distribution $\calD$ over $\{0,1\}^n$ (i.e., $\Pr_{x\sim \calD} [f(x)\ne g(x)]\ge \eps$ for any decision list $g$), given  oracle (query) access to $f$ and sampling access to $\calD$.
Inspired by~the PAC learning model, 
   distribution-free  testing was  first introduced  
  by Goldreich, Goldwasser, and Ron \cite{goldreich1998property}
  and has been studied 
  extensively \cite{AilonChazelle,HalevyKushilevitz,GlasnerServedio,HalevyKushilevitz2,HalevyKushilevitz3,DolevRon,chen2016tight,blais2021vc,ChenPatel}.
While testing is known to be no harder than proper learning \cite{goldreich1998property}, much of the work is motivated by understanding whether concept classes well studied in learning theory 
can be tested more efficiently under the distribution-free testing model.

In \cite{GlasnerServedio}, Glasner and Servedio obtained a lower bound of $\tilde{\Omega}(n^{1/5})$\footnote{For convenience we focus on the case when $\eps$ is a constant in the discussion of related work.} on the query complexity\footnote{For distribution-free testers, the query complexity refers to the number of queries made on $f$ plus the number of samples drawn from $\calD$. In many cases we simply refer to it as the number of queries made by the algorithm.} of distribution-free testing of conjunctions, decision lists and halfspaces.\footnote{Recall that conjunctions are a subclass of decision lists, which in turn are a subclass of halfspaces.}
In \cite{DolevRon},~Dolev and~Ron obtained a distribution-free testing algorithm for conjunctions with $\tilde{O}(\sqrt{n})$ queries.
Later in \cite{chen2016tight}, Chen and Xie gave
  a tight bound of $\tilde{\Theta}(n^{1/3})$ for conjunctions; their $\tilde{\Omega}(n^{1/3})$ lower bound applies~to decision lists and halfspaces as well.
For sample-based distribution-free testing,\footnote{A tester is sample-based if it can only draw samples $x_1,\ldots,x_q\sim \calD$ and receive $f(x_1),\ldots,f(x_q)$.} on the~other~hand,   Blais, Ferreira Pinto Jr. and Harms \cite{blais2021vc} obtained strong lower bounds for a number of concept~classes based on a variant of VC dimension they proposed called the ``lower~VC''~dimension.~In particular, they showed that the distribution-free testing of halfspaces requires $\tilde{\Omega}(n)$ samples. Indeed even for general testers  with queries, Chen and Patel \cite{ChenPatel} recently showed that $\tilde{\Omega}(n)$ queries are necessary, which implies that testing halfspaces is as hard as PAC learning.

To summarize, before this work, there remains wide gaps in our understanding of distribution-free testing of decision lists. In particular, it is not known whether sample-based distribution-free testing requires $\tilde{\Omega}(n)$ samples, and it is not known, when queries are allowed, whether there exists~a distribution-free tester for decision lists with query complexity sublinear in $n$. \medskip

\noindent\textbf{Our Contribution.}
We give the first sublinear distribution-free tester for decision lists:

\begin{theorem}\label{theo:mainalg}
There is a two-sided, adaptive, distribution-free testing algorithm for decision lists that makes $\tilde{O}(n^{11/12}/\eps^3)$ queries and has the same running time.\footnote{For the running time we assume that standard bitwise operations such as bitwise AND, OR and XOR over $n$-bit strings each cost one step.}
\end{theorem}

Theorem \ref{theo:mainalg} is obtained by first giving an $\tilde{O}(n^{11/12}/\eps^2)$-query algorithm for \emph{monotone} decision lists in Section \ref{sec:mdl} (where a decision list is said to be monotone \cite{GLR01} if all literals $\alpha_i$ in the list are positive), and combining it with a reduction to testing general decision lists in Section \ref{sec:gdl}.

On the lower bound side, we show that any distribution-free testing algorithm for decision lists must make $\tilde{\Omega}(\sqrt{n})$ queries, and must draw $\tilde{\Omega}(n)$ samples when the algorithm is sample-based.

\begin{theorem}\label{thm:lower_bound_wrapup}
Any two-sided, adaptive distribution-free testing algorithm for decision lists must make $\tilde{\Omega}(\sqrt{n})$ queries when $\eps$ is a sufficiently small constant. The same lower bound also applies to testing monotone decision lists.
\end{theorem}

\begin{theorem}\label{thm:sample_based_wrapup}
    Any two-sided, sample-based distribution-free testing algorithm for decision lists must draw $\tilde{\Omega}(n)$ samples when $\eps$ is a sufficiently small constant. The same lower bound also applies to testing monotone conjunctions, conjunctions, and monotone decision lists. 
\end{theorem}

As a warm-up for our main algorithm behind Theorem \ref{theo:mainalg}, we give an optimal distribution-free testing algorithm for \emph{total orderings}, which highlights, in a simplified setting, some of the most crucial ideas behind the main algorithm for monotone decision lists. To our knowledge, this is also the first tester for total orderings in the distribution-free setting. 
The input consists of 1) oracle access to a comparison function $<_\sigma$ over $[n]$
  (i.e., one can pick $i\ne j\in [n]$ to reveal whether $i<_\sigma j$ or $j<_\sigma i$);
  and 2) sampling access to a distribution $\calD$ over the 
  set of $\smash{n\choose 2}$ many $2$-subsets of $[n]$.~The goal is to determine whether $<_\sigma$ is a total ordering
  or $\eps$-far from total orderings with respect to $\calD$.
  {Equivalently, $<_\sigma$ can be considered as a tournament graph $G_\sigma$
 over $[n]$ and the algorithm is given oracle access to it (i.e., one can pick $u\ne v \in [n]$ and query whether $(u,v)$ or $(v,u)$ is in $G_\sigma$.
The goal~is~to decide whether $G_\sigma$ is acyclic or $\eps$-far from acyclic with respect to $\calD$ (i.e., any feedback edge set of $G_\sigma$ has probability mass at least $\eps$ in $\calD$, where a feedback edge set is a set of edges such that the graph $G_\sigma$ becomes acyclic after its removal).}

\begin{theorem}\label{thm:warm-up}
There is a two-sided, adaptive distribution-free testing algorithm for total orderings that makes $\tilde{O}(\sqrt{n}/\eps)$ queries. On the other hand, any such algorithm for total orderings must make $\Omega(\sqrt{n})$ queries when $\eps$ is a sufficiently small constant.
\end{theorem}

The paper is organized as follows. In \Cref{sec:prelim}, we introduce the preliminaries, including two birthday-paradox-type lemmas that will be proved later in \Cref{sec:birthday}. In \Cref{sec:warmup}, we prove the upper bound part of \Cref{thm:warm-up}, which serves as a warm-up for the proof of \Cref{theo:mainalg}. The proof of \Cref{theo:mainalg} is then presented in \Cref{sec:mdl,sec:gdl}. In \Cref{sec:lowerbound}, we prove a lower bound result (\Cref{thm:lower_bound}) that implies \Cref{thm:lower_bound_wrapup} as well as the lower bound part of \Cref{thm:warm-up}. In \Cref{sec:sample_based}, we prove a lower bound result (\Cref{thm:sample_based}) that implies \Cref{thm:sample_based_wrapup}.

\subsection{Technical Overview}

\begin{figure}[t]
\label{fig:d-no-total-order}
\begin{tikzpicture}[scale=.9]

\begin{scope}[very thick,decoration={
    markings,
    mark=at position 0.5 with {\arrow{>}}}
    ] 

\coordinate (A) at (0,1.5);
\coordinate (B) at (-1.5,0.5);
\coordinate (C) at (-1.118,-1.118);
\coordinate (D) at (1.118,-1.118);
\coordinate (E) at (1.5,0.5);

\draw[postaction={decorate}] (A) -> (B);
\draw[postaction={decorate}] (B) -> (C);
\draw[postaction={decorate}] (C) -> (D);
\draw[postaction={decorate}] (D) -> (E);
\draw[postaction={decorate}] (E) -> (A);
\draw[postaction={decorate}, dotted] (A) -> (C);
\draw[postaction={decorate}, dotted] (C) -> (E);
\draw[postaction={decorate}, dotted] (E) -> (B);
\draw[postaction={decorate}, dotted] (B) -> (D);
\draw[postaction={decorate}, dotted] (D) -> (A);

\node at (A) [above] {$\pi(1)$};
\node at (B) [left] {$\pi(2)$};
\node at (C) [left] {$\pi(3)$};
\node at (D) [right] {$\pi(4)$};
\node at (E) [right] {$\pi(5)$};


\coordinate (A2) at (5,1.5);
\coordinate (B2) at (3.5,0.5);
\coordinate (C2) at (3.882,-1.118);
\coordinate (D2) at (6.118,-1.118);
\coordinate (E2) at (6.5,0.5);

\draw[postaction={decorate}] (A2) -> (B2);
\draw[postaction={decorate}] (B2) -> (C2);
\draw[postaction={decorate}] (C2) -> (D2);
\draw[postaction={decorate}] (D2) -> (E2);
\draw[postaction={decorate}] (E2) -> (A2);
\draw[postaction={decorate}, dotted] (A2) -> (C2);
\draw[postaction={decorate}, dotted] (C2) -> (E2);
\draw[postaction={decorate}, dotted] (E2) -> (B2);
\draw[postaction={decorate}, dotted] (B2) -> (D2);
\draw[postaction={decorate}, dotted] (D2) -> (A2);

\node at (A2) [above] {$\pi(6)$};
\node at (B2) [left] {$\pi(7)$};
\node at (C2) [left] {$\pi(8)$};
\node at (D2) [right] {$\pi(9)$};
\node at (E2) [right] {$\pi(10)$};

 \node at (8.25,0) {$\dots$};



\coordinate (A3) at (12,1.5);
\coordinate (B3) at (10.5,0.5);
\coordinate (C3) at (10.882,-1.118);
\coordinate (D3) at (13.118,-1.118);
\coordinate (E3) at (13.5,0.5);

\draw[postaction={decorate}] (A3) -> (B3);
\draw[postaction={decorate}] (B3) -> (C3);
\draw[postaction={decorate}] (C3) -> (D3);
\draw[postaction={decorate}] (D3) -> (E3);
\draw[postaction={decorate}] (E3) -> (A3);
\draw[postaction={decorate}, dotted] (A3) -> (C3);
\draw[postaction={decorate}, dotted] (C3) -> (E3);
\draw[postaction={decorate}, dotted] (E3) -> (B3);
\draw[postaction={decorate}, dotted] (B3) -> (D3);
\draw[postaction={decorate}, dotted] (D3) -> (A3);

\node at (A3) [above] {$\pi(n-4)$};
\node at (B3) [left] {$\pi(n-3)$};
\node at (C3) [left] {$\pi(n-2)$};
\node at (D3) [right] {$\pi(n-1)$};
\node at (E3) [right] {$\pi(n)$};
\end{scope}
\end{tikzpicture}
\caption{One-side Lower Bound Construction for Total Orderings. An edge from $x$ to $y$ indicates that $x <_\sigma y$. The solid edges in the figure denote those in the support of $\calD_{NO}$.}\label{fig:tot-ord-lb}
\end{figure}
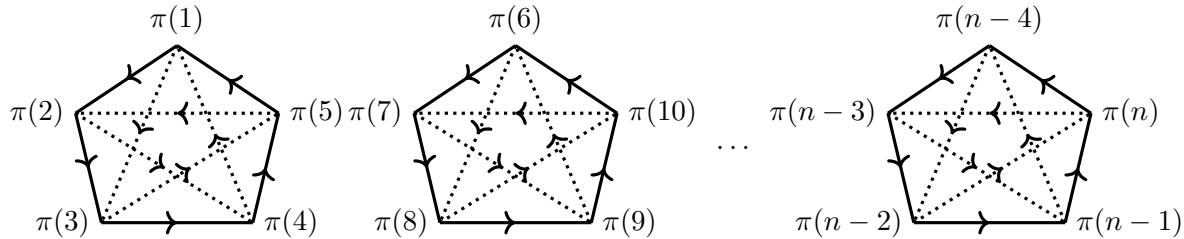

We start by describing an easy $\Omega(\sqrt{n})$ one-sided\footnote{Recall that a testing algorithm  is \emph{one-sided} if it never rejects $(<_\sigma,\calD)$ when $<_\sigma$ is a total ordering.} lower bound for total orderings.
We construct a distribution $\calD_{\text{NO}}$
  over pairs $(<_\sigma,\calD)$ such that   $<_\sigma$ is far from total orderings 
  under $\calD$.
It suffices to show that no deterministic algorithm with $o(\sqrt{n})$ queries
  can  find a violation in $(<_\sigma,\calD)\sim\calD_{\text{NO}}$ (or equivalently,
  a (directed) cycle in the tournament graph $G_\sigma$) with probability at least $2/3$. 

To draw $(<_\sigma,\calD)\sim \calD_{NO}$ \footnote{As it will become clear soon, partitioning
  $[n]$ into triangles would already yield the $\Omega(\sqrt{n})$ one-sided lower
  bound. The more involved construction of $\calD_{\text{NO}}$ here poses 
  extra challenges to motivate discussion on some of the most crucial
    ideas behind our testing algorithm for total orderings.} 
  we first draw a random permutation $\pi$ over~$[n]$ and use it partition~$[n]$ into $n/5$ groups, where the $k$-th group $V_k$ consists
  of vertices $\pi(5k-4),\ldots,\pi(5k)$, for each $k\in [n/5]$.
The comparison function $<_\sigma$ over each group $V_k$ is set according to 
  Figure \ref{fig:d-no-total-order}.
Across two different groups, $<_\sigma$ is made to be consistent with a total ordering, namely, $\pi(x) <_\sigma \pi(y)$ if $\lceil  {x}/{5} \rceil< \lceil {y}/{5} \rceil$. Finally the distribution $\calD$ is uniform over edges
  $\{\pi(5k-4),\pi(5k-3)\},\{\pi(5k-3),\pi(5k-2)\},\{\pi(5k-2),\pi(5k-1)\},
  \{\pi(5k-1),\pi(5k)\},\{\pi(5k),\pi(5k-4)\}$ of each group $k\in [n/5]$.
We write $E_k$ to denote the set of these five edges in the $k$-th group $V_k$.  

Clearly, to make $<_{\sigma}$ into a total ordering, one must change at least one edge in each $E_k$, so $<_\sigma$ is $(1/5)$-far from total orderings. On the other hand, in order for a one-sided algorithm to~reject, it must find a cycle in $V_k$ for some $k$.
Using a birthday paradox argument, with only~$o(\sqrt{n})$~samples, edges 
  sampled from $\calD$ most likely lie in distinct groups. 
When this happens, it is unlikely for~the algorithm to find 
  a cycle using $o(\sqrt{n})$ queries to the black-box oracle. 
Our lower bound for decision lists follows a similar high-level scheme, but with extra care to handle the case where the tester queries a string $x$ with large support
  (ignoring some details, testing total orderings can be thought of as testing decision lists with the restriction that the algorithm can only query the function $f $ on  strings
  $x$ with support size $2$).

We now use instances in $\calD_{\text{NO}}$ to discuss ideas behind our 
  $\tilde{O}(\sqrt{n}/\eps)$-query tester for total orderings. In particular, consider a one-sided tester that aims to find a violation (i.e., a cycle in $G_\sigma$) in $(<_\sigma,\calD)$ from $\calD_{\text{NO}}$. It must draw $\Omega(\sqrt{n})$ samples from $\calD$. After doing so, it is likely to have drawn two edges from the same $E_k$, say
  $\{\pi(5k-4,5k-3)\}$ and $\{\pi(5k-2,5k-1)\}$ for some $k\in [n/5]$.
If the algorithm continues to query the rest of four edges between these 
  four vertices, then a cycle will be found as desired.
That said, the algorithm does not know which pair of edges sampled from $\calD$
  lies in the same group, and working on all pairs would require $\Omega(n)$ queries.\footnote{Note that if the algorithm receives two samples that are 
    consecutive in the same group, then it certainly knows this because they
    share a vertex; the construction, however, makes sure that the triangle
    they form is never a cycle.}

To circumvent this issue, we create a ``sketch'' to attack $(<_\sigma,\calD)$ from 
  $\calD_{\text{NO}}$ as follows:
\begin{flushleft}\begin{enumerate}
\item Sample $\sqrt{n}$ vertices from $[n]$ uniformly at random;  
 	sort them into $\ell_1<_\sigma \ell_2 <_\sigma \cdots <_\sigma \ell_{\sqrt{n}}$ using $O(\sqrt{n}\log n)$ queries on $<_\sigma$;
\item Partition $[n]$ into \emph{blocks} $B_0,\ldots,B_{\sqrt{n}}$ where $B_i$
    consists of all $k\in [n]$ such that running binary search of $k$
    on $\ell_1,\ldots,\ell_{\sqrt{n}}$ sandwiches it in $\ell_i<_\sigma k<_
    \sigma\ell_{i+1}$.
\end{enumerate}\end{flushleft}
We note the following properties of the sketch:
\begin{flushleft}\begin{enumerate}
\item We cannot afford to compute the blocks but given any $ k \in [n]$, it is easy to find  the block $B_i$ that contains $k$
     with $O(\log n)$ queries (by just running binary search);
\item With high probability (over samples used to build the sketch),
   every $B_i$ is of size $\tilde{O}(\sqrt{n})$.
\end{enumerate}\end{flushleft} 
With this sketch in hand, we can use it to find a violation in 
  $(<_\sigma,\calD)$ from $\calD_{\text{NO}}$ efficiently by
  (1) sampling $O(\sqrt{n})$ edges from $\calD$ so that with high probability
    two edges $\{u,v\}$ and $\{u',v'\}$ from the same group are sampled;
  (2) find the block of every vertex in the $O(\sqrt{n})$ edges sampled
    in (1); let $U$ denote this set of $O(\sqrt{n})$ vertices;
  (3) for every block $B_i$ and every two vertices in $U\cap B_i$, query 
    $<_\sigma$ on them and reject if a cycle is found within $U\cap B_i$ for some $i$.
Given that most likely $u,v,u',v'$ lie in the same block,
  the algorithm finds a violation with high probability;
  its query complexity is at most $\tilde{O}(\sqrt{n})$ because
  $|U\cap B_i|$ can be bounded from above by $O(\log n)$ with high probability.
So the savings come from the fact that we only query edges between
  vertices in the same block.

The algorithm for the general case (rather than just dealing with
  instances from $\calD_{\text{NO}}$) follows the same high level idea.
It starts by building a sketch but the vertices $\ell_1,\ldots,\ell_{\sqrt{n}}$
  used to build it are no longer sampled uniformly but from a natural distribution $\calD^*$
  over $[n]$
  defined from $\calD$: to draw $\ell\sim \calD^*$, one first draws an edge
  from $\calD$ and then set $\ell$ to be one of its two vertices uniformly.
Accordingly, the second property of the sketch becomes that every $B_i$ has
  probability mass at most $O(1/\sqrt{n})$ in $\calD^*$.
With such a sketch in hand, 
  we consider cycles in $<_\sigma$. 
Since $<_\sigma$ is $\eps$-far~from~total orderings under $\calD$, 
  we can divide cycles in $G_\sigma$ into two types: those with vertices lying in multiple blocks (called \emph{long} cycles) and those that are completely contained within a single block (called \emph{local} cycles), and consider two cases:
  the distance to total orderings mainly comes from long cycles or local cycles.
To deal with the case where there are many violating long cycles, we show that $\{u,v\}\sim \calD$ satisfies
  $u<_\sigma v$, $u\in B_i$, $v\in B_j$ but $i>j$ with probability $\Omega(\eps)$.
As a result, drawing $O(1/\eps)$ samples from $\calD$ and finding
  buckets of their vertices leads to a violation with high probability.
The case of local cycles, on the other hand, is the case with instances of $\calD_{\text{NO}}$.
To this end we use a birthday paradox lemma in \Cref{sec:paradox} to show that with $\tilde{O}(\sqrt{n})$ samples $U$ from~$\calD^*$,
  some $U\cap B_i$ contains a cycle with high probability, which can be found by brute-force
  search of each $U\cap B_i$.

Unfortunately, several aspects of this approach break when attempting to adapt the algorithm to monotone decision lists. Note that a monotone decision list $f:\{0,1\}^n\rightarrow \{0,1\}$ naturally induces an ordering over strings $x \in \{0,1\}^n$ based the rule in $x$ that fires in $f$. That said, in this setting, one can only compare two strings $x,y$ with $f(x) \not = f(y)$: If $f(x\lor y)=f(x)$, then the rule that fires in $x$ is ranked higher. To accommodate this in the sketch, we bucket elements of $[n]$ into blocks $B_0, \dots, B_{\sqrt{n}}$ that now have alternating values, i.e. all indices $k \in B_i$ have that $f(e_k) = i \mod 2$, where $e_k \in \{0,1\}^n$ denotes the string in which the only $1$-entry is $k$. However, even for a monotone decision list $f$, blocks $B_0,\ldots,B_{\sqrt{n}}$ of a sketch no longer guarantee that all elements in $B_i$ are ranked higher than those in $B_{i+1}$; only a weaker guarantee holds that all elements in $B_i$ are ranked higher than those in $B_j$ for all $j > i + 1$.

The primary challenge when testing monotone decision lists is determining what constitutes a violation. In the case of total orderings, each comparison provides a concrete bit, indicating that one element is larger than the other under $<_\sigma$, and a violation is clearly defined as a cycle. However, in the case of a monotone decision list, querying a string $x$ with, say, $f(x) = 0$, only tells the algorithm that some zero rule fired in $x$ is ranked higher than all the one rules fired~in~$x$.
To address this, we design a procedure that determines the value of the maximum element $k \in \supp(x)$. However, this procedure is effective only for blocks $B_i$ that contain a small number of indices, say $n^{\delta}$ for some small constant $\delta>0$ (the number of queries made by the algorithm is linear in $n^{\delta}$ so is efficient only when $\delta$ is small).  Once we identify the maximum elements, cycles in an associated hypergraph
  naturally leads to violations. As a simple example, let $x$ and $y$ be two strings with $f(x)=0$ and $f(y)=1$. Let $k,\ell$ be maximum elements in $x$ and $y$, respectively. 
If in addition we have $\ell\in \supp(x)$ and $k\in \supp(y)$, then we get a violation
  because being the maximum element in $x$, $k$ should be ranked higher than $\ell$ but
  on the other hand, $y$ tells us that $\ell$ is ranked higher.)

Nevertheless, this procedure is insufficient for testing since many blocks in the sketch may have more than $n^\delta$ indices. For instance, if $f$ is a conjunction, there are only 2 blocks and at least one must be large. To handle such large blocks, we prove that if $f$ is a decision list and $B_i$ is a large block, then most elements of $B_i$ are smaller than those in $B_{i+1}$. If we could check that this property holds for a general $f$, which may not be a decision list, then we are in a similar setting to that of the total ordering case and can easily control violations involving elements from any large block. Verification of this property turns out to be somewhat tricky, but we demonstrate that it can be achieved with an argument similar in spirit to Dolev and Ron's conjunction tester, but crucially modified to use an asymmetric version of the birthday paradox.

Finally to extend our algorithm to test general decision lists, we note that given an arbitrary decision list $f$, if we know the default string $r$, then $f(x \oplus r)$ is now a monotone decision list. While it is not clear how to find $r$ exactly, we show that it suffices to find a string whose firing rule has sufficiently low priority in the decision list. We can then draw many sample strings and try each of them out as the candidate default string $r$. 
\section{Preliminaries}\label{sec:prelim}

\noindent\textbf{Notation.}
 Given a positive integer $n$, we write $[n]$ to denote $\{1,\ldots,n\}$. Given two integers $a\le b$,
we write $[a : b]$ to denote the set of integers $\{a,\ldots, b\}$ between 
  $a$ and $b$.
Given a probability~distribution $\calD$ over a finite set $S$, we write $\calD(p)$ to denote
  the probability mass of $p\in S$ in $\calD$, and~write  
  $\calD(P)$ for a given subset $P\subseteq S$ to denote $\sum_{p\in P} \calD(p)$. We will denote by $\supp(\calD)$ the set $\{p\in S:\calD(p)>0\}$. 
Throughout the paper, drawing a set $T$ of $m$ samples from $\calD$ always means to draw $m$ independent samples from $\calD$ (with replacements) and take $T$ to be the set they form (so in general $|T|$ could be smaller than $m$).

For any string $x\in\{0,1\}^n$, let $\supp(x)$ denote the set $\{i\in[n]:x_{i}=1\}$. Given two strings $x$ and $y\in \{0,1\}^n$, we write $x\vee y$ to denote the bitwise OR of $x$ and $y$, i.e., $x\vee y\in \{0,1\}^n$ with
  $(x\vee y)_i= x_i\lor y_i$ for all $i\in [n]$, and $x\oplus y$ to denote their bitwise XOR, with $(x\oplus y)_i=x_i\oplus y_i$ for all $i\in [n]$.
Given $i\in [n]$ we write $e_i$ to denote the string in $\{0,1\}^n$ such that $(e_i)_i=1$ and $(e_i)_j=0$ for all $j\ne i$.
Given a probability distribution $\calD$ over $\{0,1\}^n$ and $r\in \{0,1\}^n$, we write $\calD\oplus r$ to denote the distribution over $\{0,1\}^n$ with $\calD\oplus r(x)=\calD(x\oplus r)$.

Given $f:\{0,1\}^n\rightarrow \{0,1\}$, $x\in \{0,1\}^n$ is a $1$-string of $f$ if $f(x)=1$ and a $0$-string if $f(x)=0$.
\medskip

\noindent \textbf{Monotone Decision Lists.}
A function $f:\{0,1\}^n\rightarrow \{0,1\}$ is said to be a monotone decision list if it can be represented by a pair $(\pi,\nu)$\footnote{Note though that the representation is not unique in general.}, where $\pi$ is a permutation over $[n]$ and $\nu\in\{0,1\}^{n+1}$, such that
$f(x)=\nu_j$ if $j$ is the smallest integer in $[n]$ such that $x_{\pi(j)}=1$, and $f(x)=\nu_{n+1}$ when~$x=0^n$.
Variable $i\in [n]$ is said to be a $b$-rule variable if $\nu_j=b$ for $j=\pi^{-1}(i)$, where $b\in \{0,1\}$.
We write $\textsc{MonoDL}$ to denote the class of monotone decision lists.

Given $\pi$ and $x\in \{0,1\}^n$, we write $\min_{\pi}(x)$ to denote the smallest $j\in [n]$ such that $x_{\pi(j)}=1$ and it is set to $n+1$ when $x=0^n$.
Let $f$ be an \emph{arbitrary} Boolean function and $x,y$ be two strings with $f(x)\ne f(y)$. We write $x\succ_f y$ (or $y\prec_f x$) if $f(x\lor y)=f(x)$. 
Note that when $f$ is a monotone decision list, we have $x\succ_f y$ if and only if $\min_{\pi}(x)<\min_{\pi}(y)$. As such, for a decision list the ordering simply corresponds to the ordering of the rules in the decision list.
\medskip

\noindent \textbf{Decision Lists.} 
A function $f:\{0,1\}^n\rightarrow \{0,1\} $ is said to be a decision list if $g:=f(x\oplus r)$ is a monotone decision list for some $r\in \{0,1\}^n$.
Equivalently, $f$ is a decision list if it can be
  represented by a triple $(\pi, \mu,\nu)$, where $\pi:[n]\rightarrow [n]$ is a permutation over $[n]$,
  $\mu\in \{0,1\}^n$, and $\nu\in \{0,1\}^{n+1}$, such that 
$f(x)=\nu_j$ if $j$ is the smallest integer in $[n]$ such that $x_{\pi(j)}=\mu_{\pi(j)}$, and 
  $f(x)=\nu_{n+1}$ if no such $j$ exists.
Similarly, we say variable $i\in [n]$ is a $b$-rule variable if $\nu_j=b$ for $j=\pi^{-1}(i)$.
Given $\pi,\mu$ and $x\in \{0,1\}^n$,
  we  let $\min_{\pi,\mu} (x)$  denote the smallest $j$
  with $x_{\pi(j)}=\mu_{\pi(j)}$, and it is set to $n+1$ if no such $j$ exists.
\medskip

\noindent \textbf{Distribution-free Testing.} We review the model of
  distribution-free property testing. 
Let $f,g:\{0,1\}^n\rightarrow \{0,1\}$
denote two Boolean functions  and $\calD$ denote a distribution over $\{0, 1\}^n$.

\def\fC{\frak{C}}

We define the distance between $f$ and $g$ with respect to $\calD$ as
$$
\dist_\calD(f,g)=\Pr_{{x}\in \calD} \big[f({x})\ne g(x)\big].
$$
Given a class $\frak{C}$ of Boolean functions 
  (such as the class of (monotone) decision lists), we define
$$
\dist_\calD(f,\frak{C})=\min_{g\in \frak{C}}\Big(\dist_\calD(f,g)\Big)
$$
as the distance between $f$ and $\fC$ with respect to $\calD$.
We also say that $f$ is $\epsilon$-far from $\fC$ with respect to $\calD$ for some $\epsilon\ge 0$
  if $\dist_\calD(f,\fC)\ge \epsilon$.
Now we define distribution-free testing algorithms. 

Let $\fC$ be a class of Boolean functions over $\{0, 1\}^n$.
A distribution-free testing algorithm~$\ALG$ for $\fC$ has access to a pair $(f,\calD)$, where
  $f$ is an unknown Boolean function $f:\{0,1\}^n\rightarrow \{0,1\}$
  and $\calD$ is an unknown probability distribution   over $\{0,1\}^n$, via
\begin{enumerate}
\item a black-box oracle that returns the value $f(x)$ when $x\in \{0,1\}^n$ is queried; and \vspace{-0.1cm}
\item a sampling oracle that returns a sample $x\sim \calD$
  drawn independently each time.
\end{enumerate}
The algorithm $\ALG$ takes $(f,\calD,\eps)$ as input, where $\eps>0$ is a  distance parameter,
  and satisfies:
\begin{enumerate}
\item If $f\in \fC$, then $\ALG$ accepts with probability at least $2/3$; and\vspace{-0.1cm}
\item If $f$ is $\eps$-far from $\fC$ with respect to $\calD$,
  then $\ALG$ rejects with probability at least $2/3$.
\end{enumerate}

We say an algorithm is sample-based if it can only receive a sequence of samples $z_1,\ldots,z_q\sim \calD$ together with $f(z_1),\ldots,f(z_q)$.

\subsection{Birthday Paradox Lemmas}\label{sec:paradox}

As highlighted earlier in the sketch of our algorithms, birthday paradox arguments play an important role in the analysis.
We include the proof of two birthday paradox lemmas in Section \ref{sec:birthday},~one for bipartite graphs and one for hypergraphs. 
The bipartite graph lemma (\Cref{lem:bipbirthday} below) has been previously incorporated as a crucial component of the analysis in~\cite{DolevRon}
for the distribution-free testing of monomials, though without an explicit statement. In Section  \ref{sec:birthday}, we provide an alternative proof of this lemma, drawing upon the classical birthday paradox from probability theory, and then extend the proof to work for hypergraphs (Lemma \ref{lem:hypbirthday}).  

\begin{restatable}{lemma}{birthdayparadoxone}\label{lem:bipbirthday}
Let $G=(U,V,E)$ be a bipartite graph, with probability distributions $\mu$ on $U\cup \{\#\}$ and $\nu$ on $V\cup\{\#\}$. Assume that any vertex cover $C=C_{1}\sqcup C_{2}$ of $G$, where $C_{1}\subset U$ and $C_{2}\subset V$, has $\mu(C_{1})+\nu(C_{2})\geq \varepsilon$. Let $S$ be a set of $m$ independent samples from $\mu$ and $S'$ be a set of $m'$ independent samples from $\nu$, with $m$ and $m'$ satisfying $m\cdot m'\geq 100 |U|/\eps^2$ and $m,m'\geq 100/\eps$. With probability at least $0.99$, there exist $x\in S$ and $y\in S'$ such that $(x ,y )$ is an edge in $G$.
\end{restatable}
\begin{restatable}{lemma}
{birthdayparadoxtwo}\label{lem:hypbirthday}
Let $G=(V,E)$ be a $k$-uniform hypergraph and let $\mu$ be a probability distribution  over $V\cup\{\#\}$ such that  any vertex cover $C$ of $G$ has $\mu(C)\geq \varepsilon$. Let $S$ be a set of $m$ samples from $\mu$ with  $$m\geq \frac{10k^{2} |V|^{(k-1)/k}}{\eps}.$$ Then $S$ contains an edge in $G$ with probability at least $0.99$.
\end{restatable}

\renewcommand{\algorithmicrequire}{\textbf{Input:}}

\section{Warm-up: Testing Total Orderings}\label{sec:warmup}

In this section, we present a distribution-free testing algorithm for  \emph{total orderings} as a warm-up to demonstrate some of the ideas (such as the use of \emph{sketches} and the classification of cycles into \emph{long cycles} and \emph{local cycles}) that will play  important roles in our algorithm for monotone decision lists.  

In  the problem of testing total orderings, we are given 
query access to a comparison function $<_{\sigma}$ over $[n]$ and sampling access to a distribution $\calD$ over $\binom{[n]}{2}$. For any $u\ne v \in [n]$, the tester can query $<_\sigma$ on $\{u,v\}$ to reveal if $u <_{\sigma} v$ or $v <_{\sigma} u$. 
Given $<_\sigma, \calD$ and $\eps$, the goal of the tester is to
\begin{flushleft}\begin{enumerate}
\item accept with probability at least $2/3$ if the comparison function $<_\sigma$ is a total ordering; and 
\item reject with probability at least $2/3$ if $<_\sigma$ is $\eps$-far from total orderings with respect to $\calD$, i.e.,  
$$
\Pr_{\{u,v\}\sim \calD} \Big[\big[u<_\sigma v\ \text{and}\ u>_\tau v\big]\ \text{or}\ \big[
u>_\sigma v\ \text{and}\ u<_\tau v\big]\Big]\ge \eps,\quad\text{for any total ordering $<_\tau$}.
$$
\end{enumerate}\end{flushleft}

We will prove the following theorem for the distribution-free testing of total orderings:

\begin{theorem}\label{thm:total}
There is a distribution-free tester for total orderings with $\tilde{O}(\sqrt{n}/\eps)$ queries. 
\end{theorem}

We remark that our tester is optimal up to logarithmic factors. Indeed one can easily modify the lower bound proof from \Cref{sec:lowerbound} to show that any tester must make ${\Omega}(\sqrt{n})$ many queries {when $\eps$ is a sufficiently small constant.}
\subsection{Sketches}
The backbone of our tester for total orderings (as well as monotone decision lists in Section \ref{sec:mdl}) are \emph{sketches}, which, roughly speaking, can help us compare elements that are far in the ordering. 


\begin{definition}[Sketch]
A sketch $\calS=(s^{(1)}, \ldots, s^{(k)}) $ is a tuple  of distinct elements from $[n]$ for some $k\ge 1$. We say $\calS$ is \emph{consistent} with a comparison function $<_\sigma$ if $s^{(i)} <_{\sigma} s^{(i+1)}$ for all $i\in [k-1]$. 
\end{definition}

Note that when $<_\sigma$ is a total ordering, one can infer from a consistent sketch $\calS$ that $s^{(i)}<_\sigma s^{(j)}$ for all $i<j$. This, however, does not hold for general comparison functions. 

The procedure $\textsc{Sketch}$ described in Algorithm \ref{alg:total-order-sketching}  
efficiently builds a sketch  by simply sampling and sorting elements drawn from $\calD^*$, where 
  $\calD^*$ is a distribution over $[n]$ defined using 
  $\calD$ as follows$$\calD^*(i):=\frac{1}{2}\cdot \sum_{j\ne i} \calD(\{i,j\}),\quad\text{for each $i\in [n]$.}$$
Note that   sampling access to $\calD^*$ can be simulated using sampling access to $\calD$, query by query, by first sampling from $\calD$ and returning one of the two elements uniformly at random.

\begin{algorithm}[t!]
\caption{$\textsc{Sketch}(<_\sigma,\calD,\eps)$}\label{alg:total-order-sketching}
\begin{algorithmic}[1]
\Require {Oracle access to $<_\sigma$, sampling access to $\calD$ and $\eps>0$}
\State Draw $  O({\sqrt{n}}/{\eps})$ samples $\calD^*$ and let $S$ be the set of elements sampled
\State Sort elements in $S$ into $s^{(1)},\ldots,s^{(k)}$ by running MergeSort with $<_\sigma$, where $k=|S|\ge 1$
\State Query $\smash{\{s^{(i)}, s^{(i+1)}}\}$  and \textbf{reject} if $\smash{s^{(i)} >_{\sigma} s^{(i+1)}}$ for any $\smash{i\in [k-1]}$
\State \textbf{return} $\calS := (s^{(1)}, \ldots,s^{(k)} )$
\end{algorithmic}
\end{algorithm}
 
We summarize performance guarantees of $\textsc{Sketch}$ in the following lemma:

\begin{lemma}
\label{lem:total-ordering-small-mass-partition}
$\Sketch$ makes $\tilde{O}(\sqrt{n}/\eps)$ queries. It rejects or returns
  a sketch  consistent with $<_\sigma$.

Suppose that $<_\sigma$ is a total ordering. 
Then $\Sketch$ always returns a sketch 
  $\calS= (s^{(1)}, \ldots, s^{(k)})$ that is consistent with $<_\sigma$.
Moreover, with probability at least $1-o_n(1)$, $\calS$ satisfies for all $i\in [0:k]$,
\begin{equation}\label{eq:goodcond}
\Pr_{u \sim \calD^*} \left[ s^{(i)} <_\sigma  u <_\sigma s^{(i+1)}  \right] < \frac{100 \eps \log n}{\sqrt{n}}, 
\end{equation}
where the event above is $u<_\sigma s^{(1)}$ when $i=0$ and is
  $s^{(k)}<_\sigma u$ when $i=k$.
\end{lemma}

\begin{proof}
The only nontrivial part of the lemma is to show that the event described at the end occurs with probability at least $1-o_n(1)$.
Without loss of generality, take $<_\sigma$ to be the total ordering with $1<_\sigma 2<_\sigma \cdots <_\sigma n$. For each $i$, let $j_i$ denote the smallest integer with $\calD^*([i,j_i]) \geq  {100 \eps \log n}/{\sqrt{n}}$. Note that
	\[\Pr\Big[ \textsc{Sketch} \text{ does not sample any element in $[i,j_i]$} \Big] \leq \left(1 - \frac{100 \eps \log n }{\sqrt{n}} \right)^{\sqrt{n}/\eps} \leq n^{-100}.\]
So by a union bound, $\textsc{Sketch}$ samples a point from each interval $[i,j_i]$ with high probability. The lemma  follows because, if there exists an $i$ such that
	\[\Pr_{ u \sim \calD^*} \left[ s^{(i)} <_\sigma u <_\sigma s^{(i+1)}  \right] \geq \frac{100 	\eps \log n }{\sqrt{n}} \]
then it must be the case that $\Sketch$ did not sample any point in some interval $[i,j_i]$.
\end{proof}

\def\FindBlock{\textsc{FindBlock}}

Given a sketch $\calS$ that is consistent with a total ordering $<_\sigma$,
  $\FindBlock(<_\sigma,\calS,u)$ (described in Algorithm \ref{alg:total-order-find-block}) returns the unique $i\in [0:k]$ such that
\begin{enumerate}
\item $i=0$ if $u<_\sigma s^{(1)}$; 
\item $i\in [k-1]$ if either $u=s^{(i)}$ or $s^{(i)}<_\sigma u<_\sigma s^{(i+1)}$; and
\item $i=k$ if either $u=s^{(k)}$ or $s^{(k)}<_\sigma u$.
\end{enumerate}
Indeed, $\FindBlock$ returns such an $i$ for $u$ even when $<_\sigma$ is 
  an arbitrary comparison function.
  
We summarize its performance guarantees below:

\begin{lemma}
$\FindBlock(<_\sigma,\calS,u)$ is deterministic and makes $O(\log n)$ queries.
It 
  always returns an $i\in [0:k]$ that satisfies the conditions above for $u$.
\end{lemma}

Given any $<_\sigma$ and a sketch $\calS$ consistent with $<_\sigma$,
  $\FindBlock$  (which is deterministic) uses $\calS$ to induce a partition of $[n]$ into \emph{blocks}.
We say $u\in [n]$ lies in 
  the $\ell$-th block (with respect to $\calS$) for some $\ell\in [0:k]$ if $\ell=\FindBlock(<_\sigma,\calS,u)$.

\begin{algorithm}[t!]
\caption{$\textsc{FindBlock}(<_\sigma,\calS,u)$}\label{alg:total-order-find-block}
\begin{algorithmic}[1]
\Require{Oracle access to $<_\sigma$, a 
  sketch $\calS = (s^{(1)}, \ldots, s^{(k)})$ consistent with $<_{\sigma}$ and $u \in [n]$}
\State \textbf{return $i$} if $u = s^{(i)}$ for some $i\in [k]$ 
\State \textbf{return} $0$ if $u<_\sigma s^{(1)}$; \textbf{return} $k$ if $s^{(k)}<_\sigma u$
\State Set $\smash{\text{upper} \gets k}$ and $\smash{\text{lower} \gets 1}$
\While{$\smash{\text{upper} - \text{lower} > 1}$}
\State Set $\text{mid} \gets \lfloor {(\text{upper} + \text{lower})}/{2} \rfloor$
\State If $s^{(\text{mid})}>_\sigma u$, set $\text{upper} \gets \text{mid}$; otherwise, $\text{lower} \gets \text{mid}$
\EndWhile
\State \textbf{return} $\text{mid}$
\end{algorithmic}
\end{algorithm}


\subsection{The order graph and classification of cycles}

We now move to discuss how we will reject comparison functions that are far from total orderings. Towards this goal, we define the \emph{order graph} and introduce some notation:

\begin{definition}[Order graph]
Given a comparison function $<_\sigma$, the order graph $G_\sigma$ is an orientation of the complete graph $K_n$, where edge $(u,v)$ is oriented towards $v$ if $u <_\sigma v$.
The distribution $\calD$  naturally induces a distribution over 
  edges of $G_{\sigma}$: the probability mass of an edge $(u,v)$ in $G_{\sigma}$
  is given by $\calD(\{u,v\})$.
For convenience we will still use $\calD$ to denote the distribution over edges of $G_\sigma$
  and write $\calD(R)$ to denote the total probability of a set of edges $R$ in $G_\sigma$.
\end{definition}

It's easy to see that if the order graph is acyclic if and only if $<_\sigma$ is a total ordering. Moreover, we can connect distance between $<_\sigma$ and total orderings
  with feedback edge sets of $G_\sigma$:

\begin{lemma}
\label{lem:total-order-workspace}
If $<_\sigma$ is $\eps$-far from total orderings with respect to $\calD$, then any set $R$
  of edges of $G_{\sigma}$ such that $G_\sigma$  is acyclic
  after removing $R$ (i.e., $R$ is a feedback edge set) must satisfy $\calD(R) \geq \eps$. 
\end{lemma}

\begin{proof}
Let $R$ be such a set. As $G_\sigma$ after removing $R$ is acyclic, there is a total ordering of $[n]$ that is consistent with all edges of $G_\sigma$ except those in $R$. 
As a result, the distance between $<_\sigma$ and total orderings under $\calD$ is
  at most $\calD(R)$, from which we have that $\calD(R)\ge \eps$.
\end{proof}

Consider $(<_\sigma,\calD)$ such that $<_\sigma$ is $\eps$-far from
  total orderings with respect to $\calD$.
We use a sketch $\calS$ to classify cycles of $\calD$ into two types: \emph{long} cycles and \emph{local} cycles.

\begin{definition}[Long and local cycles]
Given a sketch $\calS = (s^{(1)},\ldots, s^{(k)})$, we 
  say a directed edge $(u,v)$ in $G_\sigma$ (which means
  that $u<_\sigma v$) is a \emph{long} edge (with respect to $\calS$)
  if $$\textsc{FindBlock}(<_\sigma,\calS,u)>\textsc{FindBlock}(<_\sigma,\calS,v).$$
A cycle in $G_\sigma$ is said to be a \emph{long} cycle
  if it contains at least one long edge.
A cycle in $G_\sigma$ is said to be a \emph{local} cycle
  if it does not contain any long edges.
\end{definition}

Given that every cycle is either long or 
   local, we have the following corollary of Lemma \ref{lem:total-order-workspace}:

\begin{corollary}
\label{coro1}
\emph{Suppose $<_\sigma$ is $\eps$-far from total orderings with respect to $\calD$, and $\calS$ is a sketch that is consistent with $<_\sigma$.
Then either any feedback edge set $R$ for long cycles of $G_\sigma$
  has $\calD(R)\ge \eps/2$, or any feedback edge set $R$ for local cycles of $G_\sigma$ has $\calD(R)\ge \eps/2$.}
\end{corollary}

$\textsc{TestLongCycles}$ (see Algorithm \ref{alg:test-long-cycles}) is the procedure
  that helps reject $(<_\sigma,\calD)$ when $\calD(R)\ge \eps/2$
  for any feedback edge  set $R$ of long cycles of $G_\sigma.$
It simply draws edges from $\calD$ and rejects when a long edge is found.
Given that a total ordering has no long edges, $\textsc{TestLongCycles}$
  is trivially one-sided.
Its performance guarantees are stated in the following lemma:

\begin{algorithm}[t!]
\caption{$\textsc{TestLongCycles}(<_\sigma,\calD,\eps,\calS)$}\label{alg:test-long-cycles}
\begin{algorithmic}[1]
\Require{Oracle access to $<_\sigma$, sampling access to $\calD$, $\eps>0$ and a sketch $\calS$ consistent with $<_\sigma$}
\State Draw $ {100}/{\eps}$ samples from $\calD$
\State For each $\{u,v\}$ sampled with $u<_\sigma v$, \textbf{reject} if $\textsc{FindBlock}(<_\sigma,\calS,u)>\textsc{FindBlock}(<_\sigma,\calS,v)$
\State \textbf{accept}
\end{algorithmic}
\end{algorithm}

\begin{lemma}\label{lem:testlongcycles}
$\textsc{TestLongCycles}(<_\sigma,\calD,\eps,\calS)$ makes $O(\log n/\eps)$ queries. 

When $<_\sigma$ is a total ordering, $\textsc{TestLongCycles}$  always accepts.

Suppose that any feedback edge set $R$ for long cycles 
  in $G_\sigma$ satisfies $\calD(R)\ge {\eps}/{2}$.
Then \textsc{TestLongCycles} rejects with probability at least $0.99$.
\end{lemma}
\begin{proof}
Note that the set of long edges forms a feedback edge set for  long cycles. It follows that we sample a long edge on line 1 with probability at least
$1-(1 - \eps/2)^{ {100}/{\eps}} \ge 0.99.$
\end{proof}

Next we consider the case when any feedback edge set for
  local cycles of $G_\sigma$ has mass at least $\eps/2$.
It follows from the definition that a cycle $C$ is local if and only if all of its
  vertices lie in the same block, i.e., $\FindBlock(<_\sigma,\calS,u)$
  is the same for all $u\in C$.
The following lemma motivates the procedure $\textsc{TestLocalCycles}$ for this case.
To state the lemma, we let $H$ denote the following undirected bipartite graph:
  the left side of $H$ consists of edges of $G_\sigma$;
  the right side of $H$ consists of vertices $[n]$ of $G_\sigma$;
  $(u,v)$ and $w$ has an edge iff $v<_\sigma w<_\sigma u$ and 
$$
\FindBlock(<_\sigma,\calS,u)=\FindBlock(<_\sigma,\calS,v)=\FindBlock(<_\sigma,\calS,w).
$$ 
Combining with $u<_\sigma v$ as $(u,v)$ is an edge in $G_\sigma$,
  an edge between $(u,v)$ and $w$ in $H$ implies that $u,v,w$ form
  a directed triangle, a violation to $<_\sigma$ being a total ordering.

We are now ready to state the lemma:

\begin{lemma}\label{lem:localcycles}
Suppose that any feedback edge set $R$ for local cycles in $G_\sigma$
  has $\calD(R)\ge \eps/2$.
Then any vertex cover $C=C_1\sqcup C_2$ of $H$ must have 
  $\calD(C_1)+\calD^*(C_2)\ge \eps/2$.
\end{lemma}
\begin{proof}
First we show that for any local cycle $C = c_0\ldots c_{k-1}$ in $G_\sigma$, there must exist $i$ and $j$ such that $c_i c_{(i+1)\ \text{mod}\ k} c_j$ forms a directed triangle. Assume for a contradiction that this is not the case. We start by noting that we must have that $c_0 <_\sigma c_j$ for all $j$. Indeed, assume that $c_0 <_\sigma c_i$. If $c_{i+1} <_\sigma c_0$ then $c_0 c_i c_{(i+1) \mod k}$ forms a directed triangle. Since the base case holds ($c_0 <_\sigma c_1$), we conclude that $c_0 <_\sigma c_j$ for all $j$. But now we have reached a contradiction since $c_{k-1} <_\sigma c_0$. 

Therefore, for any local cycle $C = c_0\ldots c_{k-1}$ of $G_\sigma$, there exist 
  $i$ and $j$ such that there is an edge between $(c_i,c_{i+1\ \text{mod}\ k})$ and $c_j$ in $H$.
We now claim that if $C=C_1\sqcup C_2$ is a vertex cover of $H$,
  then there is a feedback edge set $R$ for local cycles in $G_\sigma$
  with $\calD(R)\le \calD(C_1)+\calD^*(C_2)$, from which the lemma follows.
To see this is the case, we set $R$ to be the following set of edges 
  in $G_\sigma$: all edges in $C_1$ and all edges that are incident 
  to a vertex in $C_2$.
It is easy to verify that $R$ is a feedback edge set for local cycles of $G_\sigma$. This finishes the proof of the lemma.
\end{proof}
  

Based on Lemma \ref{lem:localcycles}, $\textsc{TestLocalCycles}$ (Algorithm \ref{alg:test-local-cycles}) mimics the bipartite birthday paradox  Lemma \ref{lem:bipbirthday} by drawing
  $\sqrt{n}/\eps$ samples $S$ from $\calD$ and $\sqrt{n}/\eps$ samples $T$
  from $\calD^*$.
Then for any edge $(u,v)$ in $S$ and any vertex $w$ in $T$ with all $u,v,w$ lying in the same block,
  we query~$\{u,w\}$ and $\{v,w\}$ to see if they form a directed
  triangle.
Naively, however, this could lead to $\Omega(n)$ queries (e.g., consider the worst case when all elements lie in the same block). However, by Lemma \ref{lem:total-ordering-small-mass-partition}, this is unlikely to occur when $<_\sigma$ is truly a total ordering so $\textsc{TestLocalCycles}$
  rejects when too many samples lie in the same block. 
This is where the algorithm makes two-sided errors though.

We state performance guarantees of $\textsc{TestLocalCycles}$ in
  the following lemma:
 
\begin{algorithm}[t!]
\caption{$\textsc{TestLocalCycles}(<_\sigma,\calD,\eps,\calS)$}\label{alg:test-local-cycles}
\begin{algorithmic}[1]
\Require{Oracle access to $<_\sigma$, sampling access to $\calD,$ $\eps>0$ and a sketch $\calS$ consistent with $<_\sigma$}
\State Draw $O(\sqrt{n}/\eps)$ edges $S$ from $\calD$ and draw $O(\sqrt{n}/\eps)$ elements $T$ from $\calD^*$
\State For every element $u$ in $T$ or an edge  of $S$,
  run $\textsc{FindBlock}(<_\sigma,\calS,u)$.
\State \textbf{reject} if any block has more than $1000\log n$ elements from $T$ 
\For{every $(u,v)\in S$ and $w\in T$ such that $\FindBlock$ puts 
  them in the same block}
\State Query $\{u,w\}$ and $\{v,w\}$ and \textbf{reject} if $u,v,w$ form a directed triangle in $G_\sigma$
\EndFor
\State \textbf{accept}
\end{algorithmic}
\end{algorithm}


\begin{lemma}
\label{lem:total-ordering-blocks-small}
$\textsc{TestLocalCycles}$ makes $\tilde{O}(\sqrt{n}/\eps)$ queries.

Suppose that $<_\sigma$ is a total ordering and $\calS$  is a sketch 
  that is consistent with $<_\sigma$ and satisfies (\ref{eq:goodcond}). Then \textsc{TestLocalCycles} accepts with probability at least $1 - o_n(1)$.

Suppose that any feedback edge set $R$ of local cycles in $G_\sigma$ has $\calD(R)\ge \eps/2$. Then it rejects with probability at least $1-o_n(1)$.
\end{lemma}
\begin{proof}
The query complexity follows from the fact that for any edge $(u,v)\in S$, there are at most $O(\log n)$ many $w\in T$ that lie in the same block; otherwise the procedure rejects on line 3.
So the number of potential triangles that we need to check is no more than $O(|S|\log n)=\tilde{O}(\sqrt{n}/\eps)$.

The no case follows directly from Lemma \ref{lem:localcycles} and Lemma \ref{lem:bipbirthday}.

For the yes case, we assume that the sketch $\calS$ satisfies  
\begin{equation*}
\Pr_{u \sim \calD^*} \left[ s^{(\ell)} <_\sigma u <_\sigma s^{({\ell+1})}  \right] \leq \frac{100 \eps \log n }{\sqrt{n}}
\end{equation*}
	for all $\ell$. 
Note that we only reject when $T$ contains more than $1000 \log n$ points from some block. For the $\ell$-th block, by a Chernoff bound, the probability of having more than $900\log n$ points $u\in T$ with $s^{(\ell)}<_\sigma u <_\sigma s^{(\ell+1)}$ is  at most $n^{-9}$.
So by a union bound, this does not happen with probability $1-o_n(1)$ for all blocks, in which case the number of points sampled in each block is no more than $900\log n+1<1000\log n$ even after counting the left end point of the block.
\end{proof}





\subsection{Putting it all together: An $\tilde{O}(\sqrt{n}/\eps)$ tester for total orderings}
We now have everything we need to analyze our testing algorithm 
  $\textsc{TestTotalOrdering}$.

\begin{algorithm}[t!]
\caption{$\textsc{TestTotalOrdering}(<_\sigma,\calD,\eps)$}\label{alg:total-order-test}
\begin{algorithmic}[1]
\Require{Oracle access $<_\sigma$, sampling access to $\calD$ and $\eps>0$}
\State Run $\textsc{Sketch}(<_\sigma,\calD,\eps)$ and \textbf{reject} if it rejects; otherwise let $\calS$ be its output
\State Run $\textsc{TestLongEdges}(<_\sigma,\calD,\eps, \calS)$ and 
  \textbf{reject} if it rejects
\State Run $\textsc{TestLocalCycles}(<_\sigma,\calD,\eps, \calS)$ and 
  \textbf{reject} if it rejects
\State \textbf{accept}
\end{algorithmic}
\end{algorithm}

\begin{proof}[Proof of Theorem \ref{thm:total}]
The query complexity is trivial.

	When $<_\sigma$ is a total ordering, 
    $\Sketch$ always returns a sketch $\calS$ consistent with $<_\sigma$
    and $\calS$  in addition satisfies (\ref{eq:goodcond}) with probability at least $1-o_n(1)$.
$\textsc{TestLongEdges}$ never rejects as it is one-sided. On the other hand, when $\calS$ satisfies (\ref{eq:goodcond}), by \Cref{lem:total-ordering-blocks-small}, \textsc{TestLocalCycles} accepts with probability at least $1-o_n(1)$. So the algorithm accepts with probability  $1-o_n(1)$ overall.

	Suppose now that $<_\sigma$ is $\eps$-far from total orderings with respect to $\calD$. Assume without loss of generality that $\Sketch$ returns a sketch $\calS$ consistent with $<_\sigma$; otherwise we are trivially done.~By Corollary \ref{coro1}, either any feedback edge set 
   of long cycles in $G_\sigma$ has mass at least $\eps/2$, in which case $\textsc{TestLongCycles}$ rejects with probability at least $0.99$ by Lemma \ref{lem:testlongcycles}, or any feedback edge set of local cycles has mass at least $\eps/2$, in which case $\textsc{TestLocalCycles}$ rejects with probability at least $1-o_n(1)$ by Lemma \ref{lem:total-ordering-blocks-small}. So the algorithm rejects with probability at least $2/3$ overall.
\end{proof}


\renewcommand{\algorithmicrequire}{\textbf{Input:}}

\def\Preprocess{\textsc{Preprocess}}
\def\nil{\mathsf{nil}}
\def\FindBlock{\textsc{FindBlock}}
\def\MaxIndex{\textsc{MaxIndex}}
\def\calL{\mathcal{L}}
\def\FindBigBlocks{\textsc{FindBigBlocks}}

\section{Testing Algorithm for Monotone Decision Lists}\label{sec:mdl}

In this section, we present a distribution-free testing algorithm for 
  testing monotone decision lists with $\tilde{O}(n^{11/12}/\eps^2)$ queries and running time.

\begin{theorem}\label{theo:monotonealg}
There is a two-sided, adaptive distribution-free testing algorithm for monotone~decision lists that makes $\tilde{O}(n^{11/12}/\eps^2)$ queries and has the same running time. 
\end{theorem}
  
It will be used in the next section to obtain a testing algorithm
  for general decision lists via a direct reduction, while losing an 
  extra factor of $1/\eps$.
We will focus on the query complexity of the algorithm in this section; 
  its time complexity upper bound
  follows from a standard implementation.

Similar to some of the procedures from the last section, 
  many of the procedures in this section
  (especially those in Sections \ref{sec:preprocess} and \ref{sec:maxindex}) are
  developed to
  extract structural information about an unknown input in the yes case, 
  here a monotone decision list $f:\{0,1\}^n\rightarrow \{0,1\}$.
So we encourage the reader to think about this case 
  when going through them.
Of course, these procedures will be executed on functions 
  that are not monotone decision lists.
This is why many of the lemmas about performance guarantees
  of these procedures consist of three parts:
1) the query complexity; 2) the performance guarantees
  when the function $f$ is a monotone decision list;
  and 3) the performance guarantees when $f$ is just an arbitrary function.

\subsection{Preprocessing}\label{sec:preprocess}

\begin{algorithm}[t!]
\caption{$\textsc{Preprocess}(f,\calD,\eps)$}\label{alg:preprocess}
\begin{algorithmic}[1]
\Require {Oracle access to $f:\{0,1\}^n\rightarrow \{0,1\}$, sampling access to $\mathcal{D}$ and $\eps>0$}
\State Draw a set $T^*$ of $\smash{ n^{1-\delta/2} /\eps}$ points from $\calD$ and 
  let $T\gets T^*\setminus \{0^n\}$
\State \textbf{accept} if $T$ 
  is either empty, contains $0$-strings only,
  or contains $1$-strings  of $f$  only 
\If{\Sketch$(f,T)=\nil$}
\State \textbf{reject}
\Else\ (letting $\calS=(s^{(1)},\ldots,s^{(k)})$ be the sketch returned) 
\State Run \textsc{FindBigBlocks}$(f,\calD,\eps,\calS)$ to obtain 
  $\calL\subseteq [0:k+1]$
\State \Return $(\calS,\calL)$
\EndIf
\end{algorithmic}
\end{algorithm}

Fix $\delta>0$ to be a positive constant, which will be set to be $1/6$ at the end
  to optimize the query complexity of the overall algorithm.

The preprocessing stage, $\Preprocess(f,\calD,\eps)$,
  is described in Algorithm \ref{alg:preprocess}.
At a high level, it either outputs a pair $(\calS,\calL)$, or tells the main algorithm that there is already enough evidence to either accept or reject the input.
Here $\calS$ is a \emph{sketch} consistent with $f$ to be defined next,
  which can be used to partition the set of variables $[n]$ into blocks
  (using a procedure with the same name $\FindBlock$),
  and $\calL$ contains some useful information about sizes of these blocks.

\Preprocess\  starts by  drawing a  set $T^*$ of $ n^{1-\delta/2} /\eps$ many independent samples from $\calD$, and 
  uses $T:=T^*\setminus \{0^n\}$ to build a \emph{sketch} $\calS$ of the underlying function $f$
  (unless when $T$ is either empty or consists of $0$-strings of $f$
  or $1$-strings of $f$ only, in which case
  the main algorithm accepts since either $\calD$ has
  most of its mass on $0^n$, or $f$ is very close to
  the all-$0$ or all-$1$ function).
 
We define sketches as follows:
  
\begin{definition}
A \emph{sketch} $\calS$ is a tuple $\calS=(s^{(1)},\ldots,s^{(k)})$ of strings in $\{0,1\}^n$ for some $k\ge 2$, such that $s^{(\ell)}\ne 0^n$ for all $\ell\in [k]$.
 We say a sketch $\calS$ is \emph{consistent} with a function $f:\{0,1\}^n\rightarrow \{0,1\}$ if 
    $$f(s^{(\ell)})\ne f(s^{(\ell+1)})\quad\text{and}\quad
 s^{(\ell)}\succ_f s^{(\ell+1)},\quad\text{
   for all $\ell\in [k-1]$.}$$ 
\end{definition}


To describe $\Sketch(f,T)$ we start with the following simple deterministic 
  procedure based~on binary search, called $\FindRep(f,X,Y)$ (Algorithm \ref{alg:findrep}), where $X,Y\subseteq \{0,1\}^n$ are two  sets  and $X$~is nonempty.
The goal of $\FindRep$ is to find a string $x^*\in X$ that satisfies 
\begin{equation}\label{eq:hehe}
  f\left(x^*\vee\left(\bigvee_{y\in Y}y\right) \right)=f\left( \bigvee_{z\in X\cup Y } z \right).
\end{equation}
Note that such a string always exists when $f$ is a monotone decision list.

We summarize properties of $\FindRep$ in the following lemma:

  
  
\begin{algorithm}[t!]
\caption{$\FindRep(f,X,Y )$}\label{alg:findrep}
\begin{algorithmic}[1]
\Require Oracle access to $f$, two sets $X,Y\subseteq\{0,1\}^{n}$ and $X$ is nonempty 
\State Let $b\gets f(\bigvee_{z\in X\cup Y} z)$ and $R\gets X$
\While{$|R|>1$}
    \State Partition $R$ into $R_{1}\sqcup R_{2}$ such that $|R_1|=\lfloor |R|/2\rfloor$ and $|R_2|=\lceil |R|/2\rceil$ 
   \If{$f (\bigvee_{z\in R_{1}\cup Y}z )= {b}$} 
      \State Set $R\gets R_{1}$ 
   \Else 
      \State Set $R\gets R_{2}$ 
   \EndIf
\EndWhile
\State \Return the string in the singleton set $R$
\end{algorithmic}
\end{algorithm}

\begin{lemma}
$\FindRep(f,X,Y)$ is deterministic and makes $O(\log |X|)$ queries on $f$.

When $f$ is~a monotone decision list,  $\FindRep$ always returns an $x^*\in X$
   that satisfies
  (\ref{eq:hehe}).

On the other hand, when $f$ is~an arbitrary function, $\FindRep$ always returns an $x^*\in X$ but~$x^*$ does not necessarily satisfy 
  (\ref{eq:hehe}).
\end{lemma} 

We now describe the procedure $\Sketch(f,T)$ (Algorithm \ref{alg:sketch}), where $T\subseteq \{0,1\}^n\setminus \{0^n\}$
  contains at least one $0$-string and at least one $1$-string of $f$.
We summarize its properties below:

\begin{figure}
\begin{algorithm}[H]
\caption{$\textsc{Sketch} 
 (f,T)$}\label{alg:sketch}
\begin{algorithmic}[1]

\Require {Oracle access to $f$ and 
   $T\subseteq\{0,1\}^{n}\setminus\{0^{n}\}$ 
   has at least one $0$-string and  
   one $1$-string of $f$}

\State Let $m=|T|$, $T_{0}\gets\left\{x\in T:f(x)=0\right\}$ and $T_{1}\gets\left\{x\in T:f(x)=1\right\}$ 
\For{$i$ from $1$ to $m $}
\If{both $T_0$ and $T_1$ are nonempty}    
    \State Let $b=f(\vee_{x\in T_{0}\cup T_{1}}x) $; Set $x^{(i)}\gets\textsc{FindRep} (f,T_{b},T_{\overline{b}})$ and $T_b\gets T_b\setminus \{x^{( i)}\}$ 
\Else  
    \State Let $b$ be such that 
    $T_b\ne \emptyset$; Set $x^{(i)}\gets$ an arbitrary string in $T_b$ and 
    $T_{b}\gets T_{b}\setminus \{x^{(i)}\}$
\EndIf
\EndFor
\item Divide $[m]$ into a disjoint union of nonempty intervals $[m]=I_{1}\sqcup \dots \sqcup I_{k}$ such that\vspace{0.1cm}
    \begin{itemize}
    \item[i)] $f(x^{(i)})=f(x^{(j)})$ for all $\ell\in [k]$ and all $i,j\in I_{\ell}$ \vspace{-0.06cm}
    \item[ii)] $f(x^{(i)})\neq f(x^{(j)})$ for all $\ell\in[k-1]$, $i\in I_{\ell}$ and $j\in I_{\ell +1}$  \vspace{0.1cm}
    \end{itemize} 
\State Check if $k\ge 2$ and $\smash{\calS=(s^{(1)},\ldots,s^{(k)})}$ is consistent with $f$, where $\smash{s^{(\ell)}\gets \vee_{i\in I_{\ell}}\hspace{0.05cm}x^{(i)}}$ 
\State  \Return $\calS$ if so and \Return $\nil$ otherwise
\end{algorithmic}
\end{algorithm}
\end{figure}

\begin{lemma}\label{lem:sketch}
$\Sketch(f,T)$ is deterministic and makes $O(|T|\log |T|)$  queries on $f$.

When $f$ is a monotone decision list, it always returns
  a sketch $\calS$ that is consistent with $f$.

When $f$ is an arbitrary function, it returns either $\nil$
  or a sketch $\calS$ and in the latter case,~$\calS$
  is always a sketch consistent with $f$.
\end{lemma}

\begin{proof}
The query complexity is trivial as each run of $\FindRep$ has query complexity $O(\log |T|)$.

When $f$ is a decision list represented by $(\pi,\nu)$, it is easy to verify that for all $i<j$ in $[m]$ with $f(x^{(i)})\ne f(x^{(j)})$, we have $\min_{\pi}(x^{(i)})<\min_{\pi}(x^{(j)}).$
We also have, $k\ge 2$ given that $T$ contains at least one $0$-string and at least one $1$-string of $f$.
It follows that $\Sketch$ always returns a sketch $\calS$, and $\calS$ must be a sketch consistent with $f$.

For the general case, note that $\Sketch$ always verifies whether $\calS$ is a sketch consistent with $f$ or not, and returns $\nil$ when $\calS$ is not. The lemma follows.
\end{proof}


It is clear from Lemma \ref{lem:sketch} that 
  if $\Sketch$ returns $\nil$ in $\Preprocess$, 
  we know for sure that $f$ is not a monotone decision list and 
  thus, should be rejected.
When $\Sketch$ returns a sketch $\calS$ in $\Preprocess(f,\calD)$,
  we know it must be consistent with $f$ and 
  $\Preprocess$ continues by running a procedure called $\FindBigBlocks(f,\calD,\eps,\calS)$, which 
  uses a procedure called $\FindBlock(f,\calS,x)$ that plays a similar role 
  as the $\FindBlock$ in the last section.

\begin{algorithm}[t!]
 \caption{$\textsc{FindBlock}(f,\sketch,x)$}\label{alg:findblock}
\begin{algorithmic}[1]

\Require {Oracle access to $f$, a sketch $\calS=(s^{(1)},\ldots,s^{(k)})$ that is consistent with $f$ and $x\in \{0,1\}^{n}$}
\If{$f(s^{(1)})=f(x)$}
\If{$f(s^{(2)}\vee x)= f(x)$}
   \State \Return $1$  
\ElsIf{$\smash{f(s^{(2\lfloor k/2\rfloor)}\vee x)\ne f(x)}$}
    \State \Return $\smash{2\lfloor k/2\rfloor+1}$ 
\Else 
  \State Binary search to \Return an odd $\ell$ with 
       $f(s^{(\ell-1)}\vee x)\ne f(x)=f(s^{(\ell+1 )}\vee x)$
\EndIf \EndIf
\State The case when $f(s^{(1)})\ne f(x)$ (or equivalently,
  $f(s^{(2)})=f(x)$) is symmetric
 %
\end{algorithmic}
\end{algorithm}

To motivate $\FindBlock$, we make the following observation.
Let $f$ be any function and $\calS$
  be a sketch that is consistent with $f$.
Given $x\in \{0,1\}^n$, 
  there must exist an index $\ell\in [0:k+1]$ such that
  one of the following three conditions holds:
\begin{enumerate}
    \item either $\ell\in [2:k-1]$ such that $f(x)\ne f(s^{(\ell-1)})=f(s^{(\ell+1)})$ and 
    $s^{(\ell-1)}\succ_f x\succ_f s^{(\ell+1)}$;
    \item or $\ell\in \{0,1\}$ such that $f(x)\ne f(s^{(\ell+1)})$ and $x\succ_f s^{(\ell+1)}$;
    \item or $\ell\in \{k,k+1\}$ such that $f(x)\ne f(s^{(\ell-1)})$ and 
    $s^{(\ell-1)}\succ_f x$.
\end{enumerate}
Furthermore, $\ell$ is \emph{unique} when $f$ is a monotone decision list. 

The deterministic 
  procedure \FindBlock$(f,\calS,x)$ (Algorithm \ref{alg:findblock})  
  finds such an $\ell$   efficiently for any given string $x\in \{0,1\}^n$:

\begin{lemma}
\textsc{FindBlock}$(f,\calS,x)$ is deterministic and makes $O(\log k)$  queries.
It always returns an $\ell\in [0:k+1]$ for $x$ as described above,
  which is unique when $f$ is a monotone decision list.
\end{lemma}
Using $\FindBlock$  we partition 
  variables $[n]$ into \emph{blocks} (note that we cannot afford to compute these blocks but they are well defined given that $\FindBlock$ is deterministic):

\begin{definition}[Blocks]  For fixed $f$ and $\sketch$, we define 
the $\ell$-th \emph{block} $B_{f,\calS,\ell}$ with respect to $\calS$ as 
$$
B_{f,\calS,\ell}=\big\{i\in[n]:\textsc{FindBlock}({f,\sketch,e_{i}})=\ell\big\},\quad\text{for each $\ell\in [0:k+1]$}.
$$
We usually write $B_\ell$ to denote $B_{f,\calS,\ell}$
  when $f$ and $\calS$ are clear from the context.
\end{definition}

Before moving to  $\FindBigBlocks$,
  we record a lemma about $\calS$ when $f$ is a
  monotone decision list.
 The definition below and Lemma \ref{lem:scattered} will only be used in the analysis of the yes case.  

\begin{definition}\label{def:scattered}
Let $f$ be a monotone decision list and 
  $\calS$ be a sketch consistent with $f$. 
We say $\calS$ is \emph{scattered} if we have
$$
\Pr_{x\sim \calD} \Big[\FindBlock(f,\calS,x)=k+1\Big]\le 
\frac{10\eps\log n}{n^{1-\delta/2}}$$
and 
for every $\ell\in [k]$, we have (noting that $f(x)=f(s^{(\ell)})$ if $\FindBlock(f,\calS,x)=\ell$) 
$$
\Pr_{x\sim \calD} \Big[\FindBlock(f,\calS,x)=\ell\ \text{and}\ \exists\hspace{0.05cm}i\in [n]: f(e_i)\ne f(x)\ \text{and}\  
x\succ_f e_i\succ_f s^{(\ell)} \Big]\le 
\frac{10\eps\log n}{n^{1-\delta/2}}.$$
\end{definition}

\begin{lemma}\label{lem:scattered}
Let $f$ be a monotone decision list, and 
  $T^*$ be a set of $n^{1-\delta/2}/\eps$ strings 
  drawn from $\calD$ (as on line 1 of $\Preprocess(f,\calD,\eps)$). 
The probability that $\Preprocess(f,\calD,\eps)$ returns a sketch~$\calS$ that is not scattered is $o_n(1)$  over the randomness of $T^*$.
\end{lemma}
\begin{proof}
Let $(\pi,\nu)$ be a representation of $f$. 
It follows from Definition \ref{def:scattered} that
a necessary condition for $\calS$ to be not scattered is that there exists an interval $I$ in $[n]$ and $b\in \{0,1\}$ such that 
$$
\Pr_{x\sim \calD} \Big[f(x)=b\  \text{and}\ {\min}_\pi(x)\in I\Big]> \frac{10\eps \log n}{n^{1-\delta/2}}
$$
and no sample $x\in T^*$ 
  satisfies $f(x)=b$ and $\min_\pi(x)\in I$.
The lemma follows Chernoff bound and a union bound using the fact that there are only $O(n^2)$ many intervals $I$ in $[n]$.
\end{proof}


\begin{algorithm}[t!]
\caption{$\textsc{FindBigBlocks}(f,\calD,\eps,\calS)$}\label{alg:findbigblocks}
\begin{algorithmic}[1]
\Require {Oracle access to $f$, sampling access to $\mathcal{D}$, $\eps>0$ and a sketch $\sketch$ consistent with $f$ }
\State Create and initialize a counter $c_{\ell}\gets 0$ for each $\ell\in[0:k+1]$.
 \For{$n^{1-\delta} $ times}
    \State Sample an $i\sim [n]$ uniformly at random  
    \State Let $\ell\gets \textsc{FindBlock}({f,\sketch},e_{i})$ and update counter $c_{\ell}\gets c_{\ell}+1$
\EndFor
\State Set   $\calL$ to be $$\calL\gets \left\{\ell\in [0:k+1]:c_{\ell}\geq \frac{4\log n}{\eps}\right\}$$ 
\For{$200/\eps $ times} 
\State Set counter $c\gets 0$ and let 
$$N(\calL)\gets \Big\{\ell'\in [0:k+1]: \ell'\notin \calL\ \text{and}\ |\ell'-\ell|=1\ \text{for some $\ell\in \calL$}\Big\}$$ 
\For{$(100/\eps)\log (n/\eps)$ times}
\State Sample a string $x\sim \calD$
\State Let $\ell\gets \FindBlock(f,\calS,x)$ and update counter $c\gets c+1$ if $\ell\in N(\calL)$
\EndFor
\If{$c< 5\log (n/\eps)$}
  \State \Return $\calL$
\Else
  \State Set $\calL\gets \calL\cup N(\calL)$
\EndIf

\EndFor
\State \Return $\calL$\Comment{This line is reached with low probability}
\end{algorithmic}
\end{algorithm}

The preprocessing stage ends by running 
  $\textsc{FindBigBlocks}(f,\calD,\eps,\calS)$ (Algorithm \ref{alg:findbigblocks}):

\begin{lemma} \label{lem:numofbb}
$\textsc{FindBigBlocks}(f,\calD,\eps,\calS)$ makes 
  $\tilde{O}(n^{1-\delta}/\eps^2)$ queries and it always returns a subset
   $\calL\subseteq [0:k+1]$.
With probability at least $1-o_n(1)$, $\calL$ satisfies the 
  following properties\footnote{Note that this holds no matter whether $f$ is a 
  monotone decision list or not.}:
\begin{enumerate}
\item $|\calL|\le n^{1-\delta}$;
\item Let $N(\calL)$ denote the set of neighboring blocks of $\calL$: $$N(\calL):=\big\{\ell'\in [0:k+1]\setminus \calL: |\ell'-\ell|=1\ \text{for some $\ell\in \calL$}\big\}.$$
Then we have $$\Pr_{x\sim\calD}\Big[\FindBlock(f,\calS,x)\in 
N(\calL)\Big]\le 0.1\eps.$$
\item For every $\ell\in [0:k+1]\setminus \calL$, we have 
$$|B_{ \ell}|\le \frac{16n^{\delta}\log n}{\eps}.$$
\end{enumerate}
\end{lemma}

\begin{proof}
We start by considering the $\calL$ computed on line 6. By Chernoff bound and a union bound, we have that with probability at least $1-o_n(1)$, $\calL$ on line 6 satisfies the following conditions:
\begin{enumerate}
\item Every $\ell\in \calL$ satisfies $|B_\ell|\ge (n^\delta \log n)/\eps$, from which we also have (for the $\calL$ on line 6)
$$
|\calL|\le n\cdot \frac{\eps}{n^\delta\log n}=\frac{\eps n^{1-\delta}}{\log n}.
$$
\item Every $\ell\notin \calL$ satisfies 
   $|B_\ell|\le (16 n^\delta \log n)/\eps.$
\end{enumerate}
Condition 3 of the Lemma follow directly. For condition 1, note that the for loop beginning on line 7 only repeats for $O(1/\eps)$ rounds and in each round, $\calL$ can grow by no more than twice of the size of $\calL$ on line 6.
As a result, the final size of $\calL$ is no more than
$$
O\left(\frac{1}{\eps}\right)\cdot \frac{\eps n^{1-\delta}}{\log n}\le n^{1-\delta}.
$$

Finally, for condition 2, again by Chernoff bound and a union bound, with probability at least $1-o_n(1)$, we have for every iteration of the for loop on line 7:
\begin{enumerate}
\item When $c<5\log(n/\eps)$, we have $$\Pr_{x\sim \calD}\Big[\FindBlock(f,\calS,x)\in N(\calL)\Big]\le 0.1\eps;$$ 
\item When $c\ge 5\log(n/\eps)$, the probability is at least $0.01\eps$.
\end{enumerate}
Since we repeat the for loop $200/\eps$ times, it must end with an iteration with $c<5\log (n/\eps)$ (instead of reaching the last line) and this finishes the proof of the lemma.
\end{proof}

We say $\calL$ returned by $\FindBigBlocks$ is \emph{good}
  with respect to $\calS$ if it satisfies all conditions
  stated in Lemma \ref{lem:numofbb}.
We will refer to $\ell\in \calL$ as \emph{big} blocks and 
  $\ell\notin \calL$ as \emph{small} blocks.



\subsection{$\MaxIndex$}\label{sec:maxindex}

Let $f:\{0,1\}^n\rightarrow \{0,1\}$, $\calS=(s^{(1)},\ldots,s^{(k)})$ 
  be a sketch that is consistent with $f$, and 
  $\calL\subseteq [0:k+1]$ be a good set of big blocks with respect to $\calS$.
We describe a
  deterministic procedure $\MaxIndex$ that will play an important 
  role in the main testing algorithm. 

To motivate $\MaxIndex$,
consider the case when $f$ is a monotone decision list (even though it will be ran on general functions).
Given $x\in \{0,1\}^n\setminus \{0^n\}$,
  we would like to find an $i$ such that 
  $f(e_i)=f(x)$ and $f(e_i\vee e_j)=f(e_i)$ for all
  $j\in \supp(x)$ such that $f(e_j)\ne f(x)$, i.e.,
  $i$ is one of the $f(x)$-rule variables 
  that has priority higher than any of the $\overline{f(x)}$-rule
  variables in the support of $x$.
$\MaxIndex(f,\calS,x)$ (Algorithm \ref{alg:monDLmaxind}) achieves this with $\tilde{O}(n^\delta/\eps)$ 
  queries, with a caveat though that it only promises to 
  work when $x$ lies in a block not in $\calL$: $\FindBlock(f,\calS,x)\notin \calL$.


\begin{lemma}\label{lem:maxindex}
$\textsc{MaxIndex}(f,\calS,\calL,x)$ is a deterministic procedure. It makes $O(\log n)$ many queries on $f$ when $\FindBlock (f,\calS,x)\in \calL$ and $\tilde{O}(n^\delta/\eps)$ many queries
  when $\FindBlock (f,\calS,x)\notin \calL$. It returns either~an $i\in \supp(x)$ or $\nil$; whenever it returns an $i\in \supp(x)$, we always have \begin{equation}\label{eq:check}\FindBlock(f,\calS,e_i)=\FindBlock(f,\calS,x)\quad\text{and}\quad f(x)=f(e_i).\end{equation}
Suppose that $f$ is a monotone decision list. 
Then $\MaxIndex$ always returns an $i$. 
Moreover, when $\FindBlock(f,\calS,x)\notin \calL$, the $i\in \supp(x)$ returned  additionally satisfies 
 $$f(e_i\vee e_j)=f(e_i),\quad  \text{for all $j\in \supp(x)$}.$$
\end{lemma}

\begin{proof}
The query complexity part is trivial. For the case with a general Boolean function $f$, note that $\MaxIndex$ always checks (\ref{eq:check}) before returning $i$.

When $f$ is a monotone decision list,   every $e_i$ added to $U$ must satisfy $s^{(\ell-1)}\succ_f e_i\succ_f s^{(\ell+1)}$.~But given that $\calL$ is good and $\ell\notin \calL$
  and the while loop on line 8 is repeated enough number of times, 
  $U$ by the end must contain
  all $i\in \supp(x)$ with $f(e_i)=f(x)$ and $s^{(\ell-1)}\succ_f e_i\succ_f s^{(\ell+1)}$.
At least one of the $e_i$'s in $U$ at the end will satisfy all conditions checked on line 13.
\end{proof} 

\begin{algorithm}[t!]
\caption{$\textsc{MaxIndex}(f,\sketch,\calL,x)$}\label{alg:monDLmaxind}
\begin{algorithmic}[1]
\Require {Oracle access to $f$, a sketch $\calS=(s^{(1)},\ldots,s^{(k)})$ consistent with $f$, $\calL\subseteq [0:k+1]$ that is\newline good with respect to $\calS$, and a string $x\in \{0,1\}^{n}\setminus\{0^{n}\}$}
\State Let $\ell\gets \textsc{FindBlock} (f,\calS,x)$
   and $s^{(\ell+1)}\gets 0^n$ if $\ell\in \{k,k+1\}$
\If{$\ell\in\calL$}\Comment{Case when $\ell\in \calL$}
    \State Let $E\gets \{e_{j}:j\in\supp(x)\}$ and  $e_i\gets \FindRep(f,E,\{s^{(\ell+1)}\})$ 
     
    \State \Return $i$ if $\FindBlock(f,\calS,e_i)=\ell$ and $f(e_i)=f(x)$;
    \Return $\nil$ otherwise
\EndIf
\If{$\ell\notin \calL$}
\State Let $E\gets \{e_{j}:j\in\supp(x)\}$ and $U\gets \emptyset$ \Comment{Case when $\ell\notin \calL$}
\While{$|U|< (16n^{\delta}\log n)/\eps$ and $f(s^{(\ell+1)}\vee(\vee_{e\in E}\hspace{0.05cm}e))=f(x)$}  
\State Set $\smash{z\gets \textsc{FindRep}(f,E,\{s^{(\ell+1)}\})}$
\State Update $U\gets U\cup \{z\}$ and $E\gets E\setminus\{z\}$ 
\EndWhile
\For{each $e_i\in U$} 
  \State \Return $i$
    if  $\FindBlock(f,\calS,e_i)=\ell$,
    $f(e_i)=f(x)$ and 
    $f(e_i\vee (\vee_{e\in E}\hspace{0.05cm} e))=f(x)$
\EndFor
\State \Return $\nil$
\EndIf
 \end{algorithmic}
\end{algorithm}


\def\MonoDL{\textsc{MonotoneDL}}
\def\eps{\epsilon}
\def\FindBlock{\textsc{FindBlock}}
\def\MaxIndex{\textsc{MaxIndex}}
\def\nil{\textsf{nil}}

\subsection{The auxiliary graph and classification of its cycles}

After the preprocessing stage, let $\calS=(s^{(1)},\ldots,s^{(k)})$ be a sketch that is 
  consistent with $f$, and $\calL$ $\subseteq [0:k+1]$ be a good set of big blocks with respect to $\calS$.
(When analyzing the yes case later, we will add the condition that $\calS$ is scattered as well.)
Given $\calS$ and $\calL$, 
\begin{align*}\FindBlock(f,\calS,\cdot)&:\{0,1\}^n\rightarrow [0:k+1]\quad\text{and} \\
\MaxIndex(f,\calS,\calL,\cdot)&:\{0,1\}^n\setminus \{0^n\}\rightarrow [n]\cup \{\nil\}
\end{align*}are two well-defined 
  deterministic maps.~We use these two maps to classify cycles in the following
  auxiliary directed graph $H$:

\begin{definition}
We write $H$ to denote the directed (bipartite) graph with vertex set 
$$V(H)=\big\{(u,W):u\in[n]\ \text{and}\ W\subseteq [n]\big\},$$ 
and there is a directed edge from $(u,W_{1})$ to $(v,W_{2})$ in $H$ if and only if $v\in W_{1}$ and $f(e_u)\ne f(e_v)$.
\end{definition}

We will refer to $H$ as the \emph{auxiliary} graph.
The following definition shows how $\calS$ and $\calL$ together
  can be used to define a map $\varphi$ from $\{0,1\}^n$ to  $V(H)\cup \{\star,\nil\}$, which in turn induces a probability distribution  
  over $V(H)\cap \{\star,\nil\}$ from $\calD$:

\begin{definition}The map $\varphi_{f,\sketch,\calL}:\{0,1\}^n \rightarrow V(H)\cup \{\star, \nil\}$ is defined
as follows: $\varphi_{f,\calS,\calL}(0^n)=\star$; 
$$\varphi_{f,\sketch,\calL}(x):=\Big (\MaxIndex({f,\calS,\calL},x),\,\big\{i\in \supp(x): f(e_{i})\neq f(x)\big\}\Big)$$
if $x\ne 0^n$ and $\textsc{MaxIndex}(f,\calS,\calL,x)\ne \nil$; and $\varphi_{f,\calS,\calL}(x)=\nil$ otherwise.

For convenience, we will just write $\varphi$ for
  $\varphi_{f,\calS,\calL}$ when its subscripts are clear
  from the context.
Given $\varphi$ and $\calD$, we write 
  $\mathcal{D}\circ \varphi^{-1}$ to denote the push-forward of the probability distribution $\mathcal{D}$ by $\varphi$, i.e., for each $(u,W)\in V(H)$,
\begin{align*}
\mathcal{D}\circ\varphi^{-1} (u,W)&=\sum_{x\in \{0,1\}^n\setminus \{0^n\}} \calD(x)\cdot 1\big[\varphi (x)=(u,W)\big],\\[0.6ex]
\mathcal{D}\circ\varphi^{-1} (\nil)&=\sum_{x\in \{0,1\}^n\setminus \{0^n\}} \calD(x)\cdot 1\big[\varphi(x)=\nil\big]\quad\text{and}\quad\calD\circ\varphi^{-1}(\star)=\calD(0^n).
\end{align*} 
\end{definition}


%
The following lemma shows that, when $f$ is $\eps$-far
  from monotone decision lists with respect to $\calD$,
  any feedback vertex
  set of $H$ must have mass at least $\Omega(\eps)$
  in $\calD\circ\varphi^{-1}$. 





\begin{lemma}\label{prop:pushforward}
Suppose that $f$ is $\varepsilon$-far from monotone decision lists with respect to $\mathcal{D}$,
  $\calS$ is a sketch consistent with $f$, and
  $\calL$ is a good set of big blocks with respect to $\calS$.
Then either 
$ 
\calD\circ\varphi^{-1} (\emph{\nil})\ge \eps/2,
$ 
or we have $\mathcal{D}\circ \varphi^{-1}(U)\geq \varepsilon/2$ for any feedback vertex set $U\subseteq V(H)$ of $H$.
\end{lemma}
\begin{proof}
Assume that $\calD\circ\varphi^{-1}(\nil)<\eps/2$.
{Let $I$ denote the set of $i\in [n]$ such that $(i,W)\in V(H)\setminus U$ for some $W$.
As $H$ after removing $U$ is cycle-free, it induces a partial order over $I$. 
Consider any total order over $I$ consistent with that partial order, which together with values of $f(e_{i})$ (for $i\in I$) and $f(0^n)$ (for the case when $x_i=0$ for all $i\in I$) defines a monotone decision list $g$ on $\{0,1\}^{n}$.

For any $x\in\{0,1\}^{n}\setminus \{0^n\}$ such that 
  $\varphi (x)\in V(H)\setminus U$, we show that $g(x)=f(x)$. 
To see this is the case, taking any such $x$, we note that by Lemma \ref{lem:maxindex} about $\MaxIndex$,
$$
\varphi(x)_1\in \supp(x)\quad\text{and}\quad
f(x)=f\big(e_{\varphi(x)_1}\big).
$$  
On the other hand, assume for a contradiction that 
  in the monotone decision list $g$, $g(x)$ is set to be 
  $\overline{f(x)}=f(e_i)$ because of some variable $i\in I$.
(Note that it cannot be the case that  
  $x_i=0$ for all $i\in I$, since we already know that $\varphi(x)_1\in \supp(x)$
  and $\varphi(x)_1\in I$ using $\varphi(x)\in V(H)\setminus U$.)
Then we have that $H$ after removing $U$ has an edge
  from $\varphi(x)$ to $(i,W)$ for some $W$, as $i\in I$.
This~implies that $i$ is dominated by
  $\varphi(x)_1$ in the total order over $I$ and thus, $g(x)$ cannot be set according to $i$, a contradiction.}
  
Finally, given that $f(x)=g(x)$ for all $x\in\{0,1\}^{n}\setminus \{0^n\}$ such that 
  $\varphi (x)\in V(H)\setminus U$, we have
\[ \eps\le \mathrm{dist}_{\mathcal{D}}(f,\text{monDL})\leq \mathcal{D}\big(\{x:g(x)\neq f(x)\}\big)\leq \varepsilon/2+\mathcal{D}\circ \varphi ^{-1}(U).\]
This finishes the proof of the lemma.
\end{proof}

We further classify cycles of $H$ into six types using $\calS,\calL$ and the map $\FindBlock(f,\calS,\cdot)$.
It follows from Lemma \ref{prop:pushforward} that, for some~$c\in [5]$, 
  any vertex feedback set of type-$c$ cycles in 
  $H$ must have mass  $\Omega(\eps)$ in $\calD\circ\varphi^{-1}$
  (it will become clear that we don't need to deal with type-$0$ cycles).
Our main testing algorithm then consists of five procedures, each handling 
  one type of cycles.

\begin{definition}[Types of cycles in $H$]
Let $C$ be a directed cycle in $H$. We say 
\begin{enumerate}
\item[0.] $C$ is of type $0$ if it contains a vertex $(u,W)$ with $\FindBlock(f,\calS,e_u)\in N(\calL)$;

\item[1.] $C$ is of type 1 if it contains a directed edge $(u,W_{1})\rightarrow(v,W_{2})$ that satisfies
$$\textsc{FindBlock}(f,\sketch,e_v)\le \textsc{FindBlock}(f,\sketch,e_u)-2;$$

\item[2.] $C$ is of type 2 if it contains a directed edge $(u,W_1)\rightarrow (v,W_2)$ that satisfies
$$
\FindBlock(f,\calS,e_v)=\FindBlock(f,\calS,e_u)-1
$$
and both $\FindBlock(f,\calS,e_u),\FindBlock(f,\calS,e_v)\in \calL$;
\item[3.] $C$ is of type 3 if it contains a directed  edge $(u,W_{1})\rightarrow(v,W_{2})$ that satisfies
$$\big|\textsc{FindBlock}(f,\sketch,e_u)-
\textsc{FindBlock}(f,\sketch,e_v)\big|=1$$
and 
$\textsc{FindBlock}(f,\sketch,e_u),
\textsc{FindBlock}(f,\sketch,e_v)\notin \calL\cup N(\calL)$
and $f(e_u\vee e_v)\ne f(e_u)$;

\item[4.] $C$ is of type $4$ if it contains 
  two consecutive edges $(u,W_1)\rightarrow (v,W_2)\rightarrow
    (w,W_3)$ such that
\begin{equation}\label{eq:useful1}
\textsc{FindBlock}(f,\sketch,e_w)+2 =
\textsc{FindBlock}(f,\sketch,e_v)+1 =
\textsc{FindBlock}(f,\sketch,e_u) 
\end{equation}
and $\textsc{FindBlock}(f,\sketch,e_u),
\textsc{FindBlock}(f,\sketch,e_v),\textsc{FindBlock}(f,\sketch,e_w)\notin \calL\cup N(\calL)$ and
$$
f(e_u\vee e_v)=f(e_u)\quad\text{and}\quad
f(e_v\vee e_w)=f(e_v);
$$

\item[5.] $C$ is of type $5$ if it (1) has length at least $4$, (2) satisfies $\FindBlock(f,\calS,e_u)\notin \calL\cup N(\calL)$ for\\
  all $(u,W)$ in $C$, (3) satisfies $f(e_u\vee e_v)=f(e_u)$ for all
  edges $(u,W_1)\rightarrow (v,W_2)$ in $C$, and (4)
\begin{equation}\label{eq:useful2}
 \max_{(u,W)\in C} \FindBlock(f,\calS,e_u)=
 \min_{(u,W)\in C}
\FindBlock(f,\calS,e_u)+1.
\end{equation}
\end{enumerate}
\end{definition}

The following lemma shows that every cycle in $H$ falls into at least one of the five types:

\begin{lemma}\label{lem:classfication}
Any cycle in $H$ must be of type $c$ for at least one $c\in \{0,1,\ldots,5\}$. 
\end{lemma}
\begin{proof}
Let $C$ be a cycle in $H$ with the following vertex sequence:
$$(u_{1},W_{1}),(u_{2},W_{2}),\dots,(u_{\ell},W_{\ell}),(u_{\ell+1},W_{\ell+1})=(u_{1},W_{1}),$$
i.e. $u_{r+1}\in W_{r}$ for all $r\in[\ell]$.
Given that $H$ is bipartite, $\ell\ge 2$ and must be an even number.

We start by showing that if $C$ is not type $0$, $1$, or $2$ then it must be the case that 
\begin{itemize}
    \item either $\FindBlock(f,\calS,e_u)\in \calL$
  for all $(u,W)$ in $C$; 
  \item or $\FindBlock(f,\calS,e_u)\notin \calL\cup N(\calL)$ for all $(u,W)$ in $C$.
\end{itemize}




Otherwise, given the definition of $N(\calL)$, there must be 
  $(u,W_1) \in \calL$ and a $(v,W_2) \not \in \calL$ such that 
$$
\FindBlock(f,\calS,e_u)\le \FindBlock(f,\calS,e_v)-2.
$$
However, the part of the cycle $C$ from $(v,W_2)$ to $(u,W_1)$
  can never skip a block whenever it goes down in blocks (since $C$ is
  not type $1$).
Therefore, the cycle must visit every block between that of $u$
  and that of $v$.
But given that one of them is in $N(\calL)$, $C$ must be of type-$0$,
  a contradiction.
  
It is easy to show that, when $\FindBlock(f,\calS,e_u)\in \calL$
  for all $(u,W)$ in $C$, $C$ must be of type 2. So we are left with the case when $\FindBlock(f,\calS,e_u)\notin \calL\cup N(\calL)$ for all $(u,W)$ in $C$.




Assume that the cycle $C$ is not type 3 or type 4. We finish the proof
  by showing that $C$ must be of type 5.
First, given that $C$ is not of type 3 or type 4, there cannot be two consecutive edges $(u,W_1)\rightarrow (v,W_2)
  \rightarrow (w,W_3)$ that satisfy (\ref{eq:useful1}).
Combining this with the assumption that the cycle is not type 1, $C$
  must satisfy (\ref{eq:useful2}).
We also note that $C$ cannot be of length 2 (otherwise, it is of type 3),
  and it must satisfy $f(e_u\vee e_v)=f(e_u)$ for all
  edges $(u,W_1)\rightarrow (v,W_2)$ in $C$ (otherwise, it is of type 3).
This finishes the proof of the lemma.
\end{proof}

\begin{corollary}\label{coro:pushforward}
\emph{Suppose $f$ is $\eps$-far from monotone decision lists with respect to $\calD$, $\calS$ is a sketch
consistent with $f$, and $\calL$ is a good set of blocks with respect to $\calS$. Then either 
$ 
\calD\circ\varphi^{-1} (\emph{\nil})\ge \eps/2
$, 
or for some $c\in [5]$, 
  any feedback vertex set $U$ of type-$c$ cycles in $H$ has $\mathcal{D}\circ \varphi ^{-1}(U)\geq \Omega(\varepsilon)$.}
\end{corollary}
\begin{proof}
This follows from Lemma \ref{prop:pushforward}, Lemma \ref{lem:classfication},
  as well as the fact that the mass of $\calD\circ\varphi^{-1}$ on 
  $(u,W)$ with $\FindBlock(f,\calS,e_u)\in N(\calL)$ is at most $0.1\eps$
  given that $\calL$ is good.
\end{proof}

\subsection{The main testing algorithm and proof of Theorem \ref{theo:monotonealg}}

The testing algorithm $\textsc{MonotoneDL}(f,\calD,\eps)$ for monotone decision lists is given in Algorithm~\ref{alg:monDLtester}. After running the preprocessing stage to obtain a sketch $\calS$
  and a set $\calL$ of big blocks,
  the algorithm quickly  checks whether $\calD\circ \varphi^{-1}(\nil)\ge \eps/2$ or not in 
  line 8.
The rest of the algorithm then consists of  one procedure for each 
  of the five types of cycles.

We list performance guarantees of these 
  procedures in lemmas below.
In all lemmas we assume 
\begin{enumerate}
  \item $f:\{0,1\}^n\rightarrow \{0,1\}$ and 
  $\calD$ is a probability distribution over $\{0,1\}^n$; 
  \item $\calS=(s^{(1)},\ldots,s^{(k)})$
  is a sketch that is consistent with $f$; and
  \item $\calL$ is a good set of big blocks with respect to $\calS$.
\end{enumerate}
In each lemma we describe performance guarantees of the procedure when $f$ is a monotone decision list and when any  feedback vertex set for type-$c$ cycles in $H$, for some $c\in [5]$, is at least $\Omega(\eps)$.

\begin{algorithm}[t!]
\caption{$\textsc{MonotoneDL}(f,\calD,\eps)$}\label{alg:monDLtester}
\begin{algorithmic}[1]
\Require {Oracle access to $f:\{0,1\}^{n}\rightarrow\{0,1\}$, sampling access to $\mathcal{D}$ and $\eps>0$}

\State Run $\Preprocess(f,\calD,\eps)$
\If{it accepts or rejects}
    \State \Return the same answer
\Else
    \State Let $(\calS,\calL)$ be the pair it returns
\EndIf
\For{$O(1/\eps)$ times}
\State Draw $x\sim \calD$ and run $\MaxIndex(f,\calS,x)$;
  \textbf{reject}
  if it returns $\nil$
\EndFor
\State \textsc{TestType-1}$(f,\calD,\eps,\calS,\calL)$ and \textbf{reject} if it rejects 

\State \textsc{TestType-2}$(f,\calD,\eps,\calS,\calL)$ and \textbf{reject} if it rejects 
\State \textsc{TestType-3}$(f,\calD,\eps,\calS,\calL)$ and 
  \textbf{reject} if it rejects

\State \textsc{TestType-4}$(f,\calD,\eps,\calS,\calL)$ and 
  \textbf{reject} if it rejects

\State \textsc{TestType-5}$(f,\calD,\eps,\calS,\calL)$ and 
  \textbf{reject} if it rejects
\State \textbf{accept} 
\end{algorithmic}
\end{algorithm}

\begin{lemma}\label{lemma:type1}
$\textsc{TestType-1} $ makes $\tilde{O}(n^{0.5+\delta}/\eps^2)$ queries  
 and always  accepts when $f$ is a monotone decision list.
If any  feedback vertex set of type-$1$ cycles 
  in $H$ has probability mass $\Omega(\eps)$ in $\calD\circ\varphi^{-1}$, then 
  $\textsc{TestType-1}$ rejects
  with probability at least $0.9$.
\end{lemma}

\begin{lemma}\label{lemma:type2}
    $\textsc{TestType-2}$ makes $\tilde{O}(n^{1-\delta/2}/\eps)$ queries.

When $f$ is a monotone decision list and $\calS$ is scattered, 
  it rejects with probability at most $o_n(1)$.

If any feedback vertex set of type-$2$ cycles 
  in $H$ has probability mass at least $\Omega(\eps)$ in $\calD\circ\varphi^{-1}$, 
  then $\textsc{TestType-2}$ rejects
  with probability at least $0.9$.
\end{lemma}

\begin{lemma}\label{lemma:type3}
    $\textsc{TestType-3}$ makes $\tilde{O}(n^{0.5+\delta}/\eps^2)$ queries and always  accepts when $f$ is a monotone decision list.
If any  feedback vertex set of type-$3$ cycles 
  in $H$ has probability mass $\Omega(\eps)$ in $\calD\circ\varphi^{-1}$, 
  then \textsc{TestType-3} rejects
  with probability at least $0.9$.
\end{lemma}

\begin{lemma}\label{lemma:type4}
$\textsc{TestType-4}$ makes $\tilde{O}(n^{2/3+\delta}/\eps^2)$ queries and  always accepts when $f$ is a monotone decision list.
If any feedback vertex  set of type-$4$ cycles 
  in $H$ has probability mass $\Omega(\eps)$ in $\calD\circ \varphi^{-1}$, 
  then $\textsc{TestType-4}$ rejects
  with probability at least $0.9$.
\end{lemma}

\begin{lemma}\label{lemma:type5}
$\textsc{TestType-5}$ makes $\tilde{O}(n^{3/4+\delta}/\eps^2)$ queries and always accepts when $f$ is a monotone decision list.
If any  feedback vertex set of type-$5$ cycles 
  in $H$ has probability mass $\Omega(\eps)$ in $\calD\circ\varphi^{-1}$, 
  then $\textsc{TestType-5}$ rejects
  with probability at least $0.9$.
\end{lemma}

We prove these five lemmas in Section \ref{sec:mainlemmaanalysis}.
Theorem \ref{theo:monotonealg} follows directly:
\begin{proof}[Proof of Theorem \ref{theo:monotonealg}]
The overall query complexity of $\textsc{MonotoneDL}$ is 
$$
\tilde{O}\left(\frac{n^{1-\delta/2}}{\eps}\right)+ 
\tilde{O}\left(\frac{n^{1-\delta}}{\eps^2}\right)+\tilde{O}\left(\frac{n^{3/4 +\delta}}{\eps^2}\right)=\tilde{O}\left(\frac{n^{11/12}}{\eps^2}\right),
$$
when $\delta$ is set to be $1/6$.

When $f$ is a monotone decision list, the only possibility for it to be rejected is by $\textsc{TestType-2}$.
But by Lemma \ref{lem:scattered}, $\calS$ is not scattered with probability $o_n(1)$ and when it is scattered, by Lemma \ref{lemma:type2}, $\textsc{TestType-2}$ rejects with probability $o_n(1)$.

When $f$ is $\eps$-far from monotone decision lists with respect to $\calD$, it is accepted by $\Preprocess$ with probability $o_n(1)$, given that $\calD$ cannot have more than $1-\eps$  mass on $0^n$ and that $f$ cannot be $\eps$-close to the all-$0$ or all-$1$ function with respect to $\calD$. Therefore, with probability at least $1-o_n(1)$ either $\textsc{MonotoneDL}$ already rejected $(f,\calD)$ or it reaches line 7 with a sketch $\calS$ consistent with $f$ and an $\calL$ that is good with respect to $\calS$ by Lemma \ref{lem:numofbb}.
When this happens, 
either the probability of $\MaxIndex(f,\calS,x)=\nil$ as $x\sim \calD$ is at least $0.5\eps$, in which case line 8 rejects with probability at least $0.9$, or the rest of $\textsc{MonotoneDL}$ rejects with probability at least $0.9$ by Lemma \ref{lemma:type1}, \ref{lemma:type2}, \ref{lemma:type3},
\ref{lemma:type4}
and \ref{lemma:type5}. This finishes the proof of the theorem.
\end{proof}

\subsection{Main procedures and their analyses}\label{sec:mainlemmaanalysis}


\begin{proof}[Proof of Lemma \ref{lemma:type1}]
The procedure $\textsc{TestType-1}(f,\calD,\eps,\calS,\calL)$ is described in Algorithm \ref{alg:type1}. 

First we show that $\textsc{TestType-1}$ never rejects when $f$ is a monotone decision list. 
To see this is the case, assume for a contradiction that the event occurs on $x\in P$ and $y\in Q$. Let $$i=\MaxIndex(f,\calS,\calL,y),\quad \ell=\FindBlock(f,\calS,x),\quad\text{and}\quad \ell'=\FindBlock(f,\calS,y).$$ We have by Lemma \ref{lem:maxindex} that 
$$
\FindBlock(f,\calS,e_i)=\FindBlock(f,\calS,y)=\ell'\le \ell-2.
$$
Using $i\in \supp(x)$, we have 
$ 
x \succ_f e_i\succ_f s^{(\ell'+1)}$  which contradicts with $s^{(\ell-1)}\succ_f x$ and $\ell\ge \ell'+2$.

Now we consider the case when any  feedback vertex set of type-1 cycles in 
  $H$ has mass at least $\Omega(\eps)$ in $\calD\circ\varphi^{-1}$.
To this end, we apply the birthday paradox lemma (\Cref{lem:bipbirthday}) 
  on the following bipartite graph $G$.
The left side of $G$ has vertices $U:=[n]$;  
  the right side has vertices $V:=V(H)$;
  and $(v,(u,W))$ is an edge in $G$ if and only if $v\in W$, $f(e_u)\ne f(e_v)$,
  and 
$$
{\FindBlock}(f,\calS,e_v)\le {\FindBlock}(f,\calS,e_u)-2.
$$
To apply \Cref{lem:bipbirthday}, the distribution $\mu$ over $U\cup \{\#\}$ is defined as 
  $$\mu(u)=\sum_{(u,W)\in V(H)}\calD\circ\varphi^{-1}(u,W) 
  \quad\text{and}\quad \mu(\#)=\calD\circ\varphi^{-1}(\star)+\calD\circ\varphi^{-1}(\nil).$$
The distribution $\nu$ over $V\cup \{\#\}$ is exactly $\calD\circ\varphi^{-1}$
  except that $\nu(\#)=\calD\circ\varphi^{-1}(\star)+
    \calD\circ\varphi^{-1}(\nil).$
It follows that any vertex cover of $G$ must have total weight at least $\Omega(\eps)$.
The no-part of the lemma follows directly from \Cref{lem:bipbirthday}.
\end{proof}

\begin{algorithm}[t!]
\caption{$\textsc{TestType-1}(f,\calD,\eps,\calS,\calL)$}\label{alg:type1}
\begin{algorithmic}[1]
\State Draw two sets $P^*,Q^*$ of $O(\sqrt{n}/\varepsilon)$ samples from $\mathcal{D}$; let $P\gets P^*\setminus \{0^n\}$ and $Q\gets Q^*\setminus \{0^n\}$
\State For each $x\in P\cup Q  $, compute $\textsc{FindBlock}({f,\sketch},x)$ and $\MaxIndex(f,\calS,\calL,x)$
\State \textbf{reject} if there exist $x\in P $ and $y\in Q $ such that $f(x)\ne f(y)$,
         $$\textsc{FindBlock}(f,\sketch,y)\le \textsc{FindBlock}(f,\sketch,x)-2\quad\text{and}\quad
         \MaxIndex({f,\sketch},\calL,y)\in\mathrm{supp}(x)$$ 
\State \textbf{accept} otherwise 
\end{algorithmic}
\end{algorithm}

\begin{algorithm}[t!]
\caption{$\textsc{TestType-2}(f,\calD,\eps,\calS,\calL)$}\label{alg:type2}
\begin{algorithmic}[1]
\State Draw a set $P^*$ of  samples and a set $Q^*$ of  samples from $\mathcal{D}$ of size given as follows:
$$
\frac{n^{\delta/2}}{\varepsilon\log^2 n}\quad\ \text{and}\ \quad 
\frac{n^{1-\delta/2}\log^3 n}{\eps}, \quad\text{\ \ respectively}
$$ 
\State Let $P\gets P^*\setminus \{0^n\}$ and $Q\gets Q^*\setminus \{0^n\}$
\State For each $x\in P\cup Q$, compute $\FindBlock(f,\calS,x)$
\State For each $x\in P\cup Q$ with $\FindBlock(f,\calS,x)\in \calL$, compute $\MaxIndex(f,\calS,\calL,x)$
\State \textbf{reject} if there exist $x\in P $ and $y\in Q $ such that \vspace{0.15cm}
\begin{enumerate}
    \item[i)] $\textsc{FindBlock}({f,\sketch},y)=
         \textsc{FindBlock}(f,\sketch,x)-1$;
    \item[ii)] $\textsc{FindBlock}({f,\sketch},x),
         \textsc{FindBlock}(f,\sketch,y)\in \calL$; and 
    \item[iii)] $\MaxIndex(f,\calS,\calL,y)\in \supp(x)$\vspace{0.15cm}
\end{enumerate}
\State \textbf{accept} otherwise 
\end{algorithmic}
\end{algorithm}

\begin{algorithm}[t!]
\caption{$\textsc{TestType-3}(f,\calD,\eps,\calS,\calL)$}\label{alg:type3}
\begin{algorithmic}[1]
\State Draw two sets $P^*,Q^*$ of $O(\sqrt{n}/\eps)$ samples  from $\calD$; let $P\gets P^*\setminus \{0^n\}$ and $Q\gets Q^*\setminus \{0^n\}$ 
\State Compute $\textsc{FindBlock}({f,\sketch},x)$ and $\textsc{MaxIndex}(f,\sketch,\calL,x)$ for each $x\in P\cup Q$
\State \textbf{reject} if there exist $x \in P$ and $y\in Q$ such that \vspace{0.15cm}
        \begin{itemize}
        \item[i)] $\big|\textsc{FindBlock}(f,\sketch,x)- \textsc{FindBlock}(f,\sketch,y)\big|=1$; 
        \item[ii)] 
        $\textsc{FindBlock}(f,\sketch,x),\textsc{FindBlock}(f,\sketch,y) \notin \calL\cup N(\calL)$;  
        \item[iii)] $\textsc{MaxIndex}(f,\sketch,\calL,y)\in \mathrm{supp}(x)$; and
        \item[iv)] $f(e_u\vee e_v)\ne f(e_u)$,
         where $u=\textsc{MaxIndex}(f,\sketch,\calL,x)$ and $v=\textsc{MaxIndex}(f,\sketch,\calL,y)$\vspace{0.15cm} 
        \end{itemize}
\State \textbf{accept} otherwise
\end{algorithmic}
\end{algorithm}

\begin{algorithm}[t!]
\caption{$\textsc{TestType-4}(f,\calD,\eps,\calS,\calL)$}\label{alg:type4}
\begin{algorithmic}[1]
\State Draw a set $P^* $ of $O({n^{2/3}}/\eps)$ samples  from $\calD$; let $P\gets P^*\setminus \{0^n\}$ 
\State Compute $\textsc{FindBlock}({f,\sketch},x)$ and $\textsc{MaxIndex}(f,\sketch,\calL,x)$ for each $x\in P $
\State \textbf{reject} if there exist $x,y,z \in P$  such that \vspace{0.15cm}
        \begin{itemize}
        \item[i)] $ \textsc{FindBlock}(f,\sketch,z)+2=\textsc{FindBlock}(f,\sketch,y)+1=
        \textsc{FindBlock}(f,\sketch,x) $; 
        \item[ii)] 
        $\textsc{FindBlock}(f,\sketch,x),\textsc{FindBlock}(f,\sketch,y),\textsc{FindBlock}(f,\sketch,z)\notin \calL\cup N(\calL)$; and
        \item[iii)]   $f(e_u\vee e_v)= f(e_u)$ and $f(e_v\vee e_w)=f(e_v)$: $u,v$ and $w$ are $\MaxIndex$ of $x,y$ and $z$\vspace{0.15cm}
        \end{itemize}
\State \textbf{accept} otherwise
\end{algorithmic}
\end{algorithm}

\begin{algorithm}[t!]
\caption{\textsc{TestType-5}$(f,\calD,\eps,\calS,\calL)$}\label{alg:type5}
\begin{algorithmic}[1]
\State Draw a set $P^*$ of $O(n^{3/4}/\eps)$ samples from $\mathcal{D}$; let $P\gets P^*\setminus \{0^n\}$
\State Compute $\textsc{FindBlock}_{f,\sketch}(x)$ and $\textsc{MaxIndex}_{f,\sketch}(x)$ for each $x\in P$ 
    \State \textbf{reject} if there exist  $x^1,x^2,x^3,x^4\in P$ such that \vspace{0.15cm}
        \begin{itemize}
        \item[i)] Let $\ell_1,\ell_2,\ell_3$ and $\ell_4$ be the $\FindBlock$ of $x^1,x^2,x^3$ and $x^4$, respectively
        \item[ii)] $\ell_{1}=\ell_{3}=\ell_{2}+1=\ell_{4}+1$ and $\ell_1,\ell_2,\ell_3,\ell_4\notin \calL\cup N(\calL)$
        \item[iii)] Let $u_1,u_2,u_3$ and $u_4$ be the $\MaxIndex$ of $x^1,x^2,x^3$ and $x^4$, respectively
        \item[iv)] $ f(e_{u_{1}}\vee e_{u_{2}})=f(e_{u_{3}}\vee e_{u_{4}})=0$ and $f(e_{u_{2}}\vee e_{u_{3}})=f(e_{u_{4}}\vee e_{u_{1}})=1$.\vspace{0.15cm}
        \end{itemize}
    \State \textbf{accept} otherwise
\end{algorithmic}
\end{algorithm}

\begin{proof}[Proof of Lemma \ref{lemma:type2}]
 $\textsc{TestType-2} $ is described in Algorithm \ref{alg:type2}.
It uses $\tilde{O}(n^{1-\delta/2}/\eps)$  queries as each call to $\MaxIndex$ uses $O(\log n)$ queries 
  when $\FindBlock(f,\calS,x)\in \calL$ given that $\calL$ is good.

For the yes case,   if $\textsc{TestType-2}$ rejects because of $x\in P $ and $y\in Q $, it must be either
\begin{enumerate}
    \item $\FindBlock(f,\calS,x)=k+1$; or
\item Letting $\ell=\FindBlock(f,\calS,x) 
\in [k]$, $\exists\hspace{0.05cm}i\in [n]$ such that $f(e_i)\ne f(x)$ and  
$x\succ_f e_i\succ_f s^{(\ell)}$.
\end{enumerate}
As $\calS$ is scattered, it follows from $|\calL|\le n^{1-\delta}$ that the probability of $P$ having such an $x$ is at most
$$
\frac{10\eps \log n}{n^{1-\delta/2}}\cdot n^{1-\delta}\cdot \frac{n^{\delta/2}}{\eps\log^2 n}=o_n(1).
$$

The proof of the no case is similar to that of Lemma \ref{lemma:type1}.
Assume that any  feedback vertex set of type-2 cycles in 
  $H$ has mass at least $\Omega(\eps)$ in $\calD\circ \varphi^{-1}$.
The left side of $G$ has vertices $U:=[n]$ and 
  the right side has vertices $V:=V(H)$,
  and $(v,(u,W))$ is  in $G$ iff $v\in W$, $f(e_u)\ne f(e_v)$
  and 
$$
{\FindBlock}(f,\calS,e_v)= {\FindBlock}(f,\calS,e_u)-1$$ \text{and} 
${\FindBlock}(f,\calS,e_u), {\FindBlock}(f,\calS,e_v)\in \calL$.

The distribution $\mu$ over $U\cup \{\#\}$  and 
  $\nu$ over $V\cup \{\#\}$ is defined in the same way as
  those in the proof of Lemma \ref{lemma:type1}.
It follows that any vertex cover of $G$ must have total weight at least $\Omega(\eps)$.
The no-part of the lemma follows directly from \Cref{lem:bipbirthday}.
\end{proof}

\begin{proof}[Proof of Lemma \ref{lemma:type3}]
The procedure $\textsc{TestType-3}$ is described in Algorithm \ref{alg:type3}.
For the analysis of its query complexity, the only nontrivial part is to observe that we don't need $|P|\times |Q|$ many queries to evaluate $f(e_u\vee e_v)$ on line iv) but $$|P|\cdot \frac{32n^\delta \log n}{\eps}$$ queries suffice.
This is because 
  $$\FindBlock(f,\calS,x)=\FindBlock(f,\calS,e_u)\quad\text{and}\quad 
  \FindBlock(f,\calS,y)=\FindBlock(f,\calS,e_v)$$ and thus,
for any $u$, the number of $v$ that can satisfy i) and ii) 
  is no more than by $2\cdot (16n^\delta \log n)/\eps$
  given that all these are small blocks outside of $\calL\cup N(\calL)$.

It is easy to verify that $\textsc{TestType-3}$ never rejects when $f$ is a monotone decision list. Assume for a contradiction that the event occurs on $x\in P$ and $y\in Q$ with $u,v$ given in iv). 
Given that $v\in \supp(x)$, we should have $f(e_u\vee e_v)=f(e_u)$ by Lemma \ref{lem:maxindex}, a contradiction.

The proof of the no case is similar to that of the previous two lemmas. We apply the birthday paradox lemma (\Cref{lem:bipbirthday}) 
  on the following bipartite graph $G$.
The left side has vertices $U=[n]$ and 
  the right side of $G$ has vertices $V=V(H)$,
  and $(v,(u,W))$ is an edge in $G$ if and only if $v\in W$, $f(e_u)\ne f(e_v)$, $|\FindBlock(f,\calS,e_u)-\FindBlock(f,\calS,e_v)|=1$, $$\FindBlock(f,\calS,e_u),\FindBlock(f,\calS,e_v)\notin \calL\cup N(\calL),$$ and $f(e_u\vee e_v)\ne f(e_u)$.
The rest of the proof is similar so we skip the details.
\end{proof}

\begin{proof}[Proof of Lemma \ref{lemma:type4}]
The procedure $\textsc{TestType-4}(f,\calD,\calS,\calL)$ is described in Algorithm \ref{alg:type4}.
For the analysis of its query complexity, we make the same observation as in the previous lemma that the number of queries needed for $f(e_u\vee e_v)$ and $f(e_v\vee e_w)$ in iii) is at most $|P|\cdot O((n^\delta \log n)/\eps)$.

For the yes case, we assume for a contradiction that $f$ is a monotone decision list but there are $x,y,z\in P$ and $u,v,w$
  that satisfy i), ii) and iii).
Then we have $e_u\succ_f e_v\succ_f e_w$ but \begin{align*} 
\FindBlock(f,\calS,e_u)&=\FindBlock(f,\calS,x)\\&=\FindBlock(f,\calS,z)+2=\FindBlock(f,\calS,e_w)+2,
\end{align*}
a contradiction.

For the no case, we apply the hypergraph birthday paradox lemma
  (Lemma \ref{lem:hypbirthday}). 
We consider the following $3$-uniform hypergraph $G$ over $[n]$:
$\{u,v,w\}$ is an edge of $G$ if and only if $$\FindBlock(f,\calS,e_w)+2=\FindBlock(f,\calS,e_v)+1=\FindBlock(f,\calS,e_u),$$
$\FindBlock(f,\calS,e_u), \FindBlock(f,\calS,e_v),\FindBlock(f,\calS,e_w)\notin \calL\cup N(\calL)$, $f(e_u\vee e_v)=f(e_u)$ and $f(e_v\vee e_w)=f(e_v)$.
The distribution $\mu$ over $[n]\cup \{\#\}$ is defined naturally as 
$$\mu(u)=\sum_{(u,W)\in V(H)} \calD\circ\varphi^{-1}(u,W)\quad\text{and}\quad
\mu(\#)=\calD\circ\varphi^{-1}(\star)+\calD\circ\varphi^{-1}(\nil).
$$
It follows from the assumption of the lemma that any vertex cover of $G$ must have mass at least $\Omega(\eps)$.
The no part of the lemma follows directly from Lemma \ref{lem:hypbirthday}.
\end{proof}

\begin{proof}[Proof of Lemma \ref{lemma:type5}]
The procedure $\textsc{TestType-5}$ is described in Algorithm \ref{alg:type5}. The analysis of its query complexity is similar to that of the previous two lemmas.

To see that $\textsc{TestType-5}$ always rejects when $f$ is a monotone decision list,  
it is easy to verify that iv) cannot be consistent with any monotone decision list since $f(e_{u_1}\vee e_{u_2})=f(e_{u_3}\vee e_{u_4})=0$ implies that 
$f(e_{u_1}\vee e_{u_2}\vee e_{u_3}\vee e_{u_4})=0$ but the second part implies it to be $1$.


For the no case, we consider the following 4-uniform hypergraph $G$ over $[n]$: $\{u_{1},u_{2},u_{3},u_{4}\}$ is an edge of $G$ if and only if
\begin{align*}
& \FindBlock(f,\calS,e_{u_{1}})=\FindBlock(f,\calS,e_{u_{3}})\\ 
&\hspace{1cm}=\FindBlock(f,\calS,e_{u_{2}})+1=\FindBlock(f,\calS,e_{u_{4}})+1
\end{align*}
and $f(e_{u_{1}}\vee e_{u_{2}})=f(e_{u_{3}}\vee e_{u_{4}})=0$ and $f(e_{u_{2}}\vee e_{u_{3}})=f(e_{u_{4}}\vee e_{u_{1}})=1$. The distribution $\mu$ over $[n]\cup\{\#\}$ is defined naturally as
$$\mu(u)=\sum_{(u,W)\in V(H)} \calD\circ\varphi^{-1}(u,W)\quad\text{and}\quad
\mu(\#)=\calD\circ\varphi^{-1}(\star)+\calD\circ\varphi^{-1}(\nil).
$$
It follows from the assumption of the lemma and \Cref{lem:alternating-4} that any vertex cover of $G$ must have mass at least $\Omega(\eps)$.
The no part of the lemma follows directly from Lemma \ref{lem:hypbirthday}. 
\end{proof}

\begin{lemma}\label{lem:alternating-4}
Let $G=(U,V,E)$ be a complete bipartite graph with an edge labeling $\phi:E\rightarrow\{0,1\}$. For any integer $k\geq 2$, we call a sequence of vertices $(u_{1},u_{2},\dots,u_{2k})$ an alternating $2k$-cycle in $(G,\phi)$ if the following holds:
\begin{itemize}
\item 
$u_{1},\dots,u_{2k-1}\in U$ and $u_{2},\dots,v_{2k}\in V$.

\item $\phi(\{u_{2i-1},u_{2i}\})=0$ and $\phi(\{u_{2i},u_{2i+1}\})=1$ for all $i\in[k]$, where $u_{2k+1}:=u_{1}$.
\end{itemize}
If $(G,\phi)$ has an alternating $2k$-cycle, then it also has an alternating 4-cycle.
\end{lemma}

\begin{proof}
Assume that $k\geq 3$ is the smallest integer such that $(G,\phi)$ contains an alternating $2k$-cycle, and let $(u_{1},\dots,u_{2k})$ be such a cycle. On one hand, if $\phi(\{u_{2},u_{2k-1}\})=1$, then $(u_{1},u_{2},u_{2k-1},u_{2k})$ forms an alternating 4-cycle. On the other hand, if $\phi(\{u_{2},u_{2k-1}\})=0$, then $(u_{2k-1},u_{2},u_{3},\dots,u_{2k-2})$ forms an alternating $2(k-1)$-cycle. In both cases we find a contradiction to the assumed minimality of $k$, so there must exist an alternating 4-cycle in $(G,\phi)$.
\end{proof}




\def\MonotoneDL{\textsc{MonotoneDL}}
\def\DL{\textsc{DecisionList}}
\def\TestDL{\textsc{TestDL}}

\section{Testing Algorithm for Decision Lists}\label{sec:gdl}

We prove Theorem \ref{theo:mainalg} by giving a  reduction from the problem of testing decision lists to that~of testing monotone decision lists. We start with a standard amplification on $\MonotoneDL$ to get
  $\MonotoneDL^*$ with the following properties:
\begin{flushleft}\begin{enumerate}
\item The number of queries made by $\MonotoneDL^*(f,\calD,\eps)$, denoted by $N_{n,\eps}$, is $ \tilde{O}(n^{11/12}/\eps^2)$;
\item For any $(f,\calD)$ such that $f$ is a monotone 
  decision list,  $\MonotoneDL^*(f,\calD,\eps)$ accepts with
  probability at least $1-o_n(1)$; and
\item For any $(f,\calD)$ such that $f$ is $\eps$-far from monotone 
  decision lists with respect to $\calD$, we have that $\MonotoneDL^*(f,\calD,\eps)$ rejects with
  probability at least $1-o_n(1)$.
\end{enumerate}\end{flushleft}
The algorithm $\DL(f,\calD,\eps)$ for testing general decision lists is described in Algorithm \ref{alg:DLtester}.
It uses a procedure called $\textsc{CheckDL}$ described in Algorithm \ref{alg:CheckDL}.

The high-level idea behind the reduction is that when $f$ is a decision list represented by $(\pi,\mu,\nu)$,
if we happen to know the minimum element in the decision list
  (i.e., $r\in \{0,1\}^n$ such that $r_{\pi(j)}\ne \mu_j$ for all $j\in [n]$ or equivalently, $\min_{\pi,\mu}(r)=n+1$, then $g$ defined as $g(x):=f(x\oplus r)$ would become a monotone  
  decision list on which we can run $\MonotoneDL^*$.
Of course, it is not clear how to find the minimum element
  efficiently, but we will show that it suffices to work with an element that is \emph{close} to being the minimum.

\begin{algorithm}[!t]
\caption{$\textbf{DecisionList}(f,\calD,\eps)$}\label{alg:DLtester}
\begin{algorithmic}[1]
\Require {Oracle access to $f$, sampling access to $\mathcal{D}$, and a distance parameter $\eps>0$}
\For{$100/\eps$ rounds}
\State Draw $r\sim \calD$ and set counter $c\gets 0$
\For{$100\log (n/\eps)$ rounds}
  \State Run $\textsc{CheckDL}(f,\calD,\eps,r)$ and set $c\gets c+1$ if it accepts
\EndFor
\State \textbf{accept} if $c\ge \log (n/\eps)$
\EndFor
\State \textbf{reject}
\end{algorithmic}
\end{algorithm}

In more details, let's consider the case when $f$ is a decision list
  and is represented by $(\pi,\mu,\nu)$.
Since we repeat the main loop of $\DL$ (which starts on line 1) for $100/\eps$
  times, it is not hard to show that with probability at least $0.9$,
  at least one of the  $r$ sampled on line 2 satisfies
\begin{equation}\label{condonr}
\Pr_{x\sim \calD}\left[\min_{\pi,\mu}(r)<\min_{\pi,\mu}(x) \right]\le 0.1\eps.
\end{equation}
To see this is the case, let $j^*\in [n+1]$ be the largest integer such that 
$$
\Pr_{x\sim\calD} \left[\min_{\pi,\mu}(x)\ge j\right]\ge 0.1\eps.
$$ 
Then at least one of the $r$'s sampled on line 2 satisfies $\min_{\pi,\mu}(x)\ge j^*$ with 
  probability at least $0.9$ and any such $r$ satisfies (\ref{condonr}).

Assuming that $r$ satisfies (\ref{condonr}). 
The simple subroutine $\TestDL$ described in Algorithm \ref{alg:TestDL} can help us test whether $f$ is 
  a decision list with almost one-sided error.
Its performance guarantees~are stated in the following lemma:

\begin{lemma}\label{lem:TestDL}
$\TestDL$ makes $ \tilde{O}(n^{11/12}/\eps^2)$ many queries.

If $f$ is $\eps$-far from   decision lists with respect to $\calD $,
  then for any strings $r,z\in \{0,1\}^n$, $\TestDL$ rejects with probability at least
  $1-o_n(1)$.

Suppose $f$ is a decision list represented by $(\pi,\mu,\nu)$, $r$ satisfies (\ref{condonr}) and $z$ satisfies $\min_{\pi,\mu}(z)>\min_{\pi,\mu}(r)$.
Then $\TestDL$ accepts with probability at least $1-o_n(1)$.
\end{lemma}
\begin{proof}
The query complexity part is trivial. When $f$ is $\eps$-far 
  from decision lists with respect to $\calD$, note that either $\dist_{\calD\oplus z}(g,h)$ is at least $0.4\eps$, in which case
  line 6 rejects with probability  $1-o_n(1)$, or
  $\dist_{\calD\oplus z}(g,h)\le 0.4\eps$ and thus, by triangular inequality,
  $h$ is at least $0.6\eps$-far from decision lists with
  respect to $\calD\oplus z$.
Therefore, $\MonotoneDL^*$ on line 7 rejects with probability $1-o_n(1)$.

On the other hand, suppose that $f$ is a decision list represented by
  $(\pi,\mu,\nu)$ and $r,z\in \{0,1\}^n$ satisfy (\ref{condonr}) and $\min_{\pi,\mu}(z)>\min_{\pi,\mu}(z)$.
It is easy to verify that  $h$ is a monotone decision list. Also by  (\ref{condonr}), $\dist_{\calD\oplus z}(g,h)\le 0.1\eps$ so
  line 6 continues with probability at least $1-o_n(1)$.
Finally, line 7 accepts with probability at least $1-o_n(1)$ given that $h$ is a monotone decision list. 
\end{proof}
\begin{algorithm}[!t]
\caption{$\textbf{TestDL}(f,\calD ,\eps,r,z)$}\label{alg:TestDL}
\begin{algorithmic}[1]
\Require {Oracle access to $f $, sampling access to $\calD $, $\eps>0$ and $r, z\in \{0,1\}^n$ }
\State Let $b=f(r)$, $g $ be the function $g (x)=g(x\oplus z)$ and 
  $h$ be the function defined as follows:
 $$
h(x)=\begin{cases}
b & \text{if $g (x)=b$}\\
\overline{b} & \text{if $g(x)=\overline{b}$ and $x\succ_g r\oplus z$}\\
b & \text{if $g(x)=\overline{b}$ and $r\oplus z\succ_g x $}
\end{cases}
$$
\State Set a counter $c\gets 0$
\For{$ (10\log n)/\eps $ rounds}
\State Draw $x\sim \calD\oplus z$ and increment $c$ if $g(x)\ne h(x)$
\EndFor
\State \textbf{reject} if $c\ge (2\log n)/\eps$ and continue otherwise
\State Run $\MonotoneDL^*(h,\calD\oplus z,\eps/2)$; \textbf{accept} if
  it accepts; \textbf{reject} if it rejects 
\end{algorithmic}
\end{algorithm}

Given $\TestDL$, the challenge for the case when $f$ is a decision list
  is to find a $z\in \{0,1\}^n$ with
  $\min_{\pi,\mu}(z)>\min_{\pi,\mu}(r)$.
This is done in $\textsc{CheckDL}$ (Algorithm \ref{alg:CheckDL}),
where the deterministic binary search subroutine called 
  $\textsc{IndexSearch}$ (Algorithm \ref{alg:IndexSearch}) will play an important role:

\begin{algorithm}[!t]
\caption{$\textbf{IndexSearch}(f,r,y)$}\label{alg:IndexSearch}
\begin{algorithmic}[1]
\Require {Oracle access to $f $ and $r,y\in \{0,1\}^n$ with $f(r)\ne f(y)$}
\State Let $b=f(r)$ and  $g(x):=f(x\oplus r)$;
below we write $g(T)$ to denote $g(x)$ where $x_i=1$ iff $i\in T$
\State Run deterministic binary search to look for an $i\in \supp(y\oplus r)$ such that $g(e_i)=\overline{b}$ 
\State \textbf{return} this $i$ if found and continue if the binary search fails
\State Note that the binary search can fail if in one round, both branches 
  evaluate to $b$ in $g$. In more\newline details, let $T_0,T_1,\ldots,T_p$ be the 
  sequence of subsets of $\supp(y\oplus r)$ followed by the binary search\newline  such that $T_0=\supp(y\oplus r)$,
  $T_{i+1}\subset T_i$ with $|T_{i+1}|\le \lceil |T_i|/2\rceil$ and 
  $g(T_0)=\cdots=g(T_p)=\overline{b}$
  but\newline both subsets that $T_p$ splits into evaluate to $b$ in $g$ so 
  the binary search fails
\If{at least one of $g(T_0\setminus T_1),\ldots,g(T_{p-1}\setminus T_p)$ is $\overline{b}$, say $g(T^*)=\overline{b}$}
\State Run binary search on $T^*$ to look for an $i\in T^*$ such that $g(e_i)=\overline{b}$
\State \textbf{return} this $i$ if found; \textbf{return} $\nil$ if binary search fails again
\Else
\State \textbf{return} $\nil$
\EndIf
\end{algorithmic}
\end{algorithm}

\begin{lemma}\label{lem:IndexSearch}
$\textsc{IndexSearch}(f,r,y)$ is deterministic and uses $O(\log n)$ queries on $f$.

When $f$ is a decision list represented by $(\pi,\mu,\nu)$ and $r,y$ satisfy 
  $f(r)\ne f(y)$ and $\min_{\pi,\mu}(r)>\min_{\pi,\mu}(y)$, it always 
  returns an $i$ with $\pi^{-1}(i)\le \min_{\pi,\mu}(r)$ such that $f(r^{(i)})\ne f(r)$.
  
When $f$ is a decision list represented by $(\pi,\mu,\nu)$ and $r,y$ satisfy
   $f(r)\ne f(y)$ and $\min_{\pi,\mu}(r)<\min_{\pi,\mu}(y)$, 
   it returns either $\pi(\min_{\pi,\mu}(r))$ or $\emph{\nil}$.
\end{lemma}
\begin{proof}
The query complexity is trivial. The case when $\min_{\pi,\mu}(r)<\min_{\pi,\mu}(y)$ is also trivial because the 
  algorithm always returns either $\nil$ or some $i\in \supp(y\oplus r)$ 
  such that $\overline{b}=g(e_i)=f(r^{(i)})$.
But given the assumption that $\min_{\pi,\mu}(r)<\min_{\pi,\mu}(y)$, it must 
  be the case that $i=\min_{\pi,\mu}(r)$.

Next consider the case when $\min_{\pi,\mu}(r)>\min_{\pi,\mu}(y)$, and let 
  $i^*=\pi(\min_{\pi,\mu}(r))$.
If an $i$ is~found by the first binary search on line 3, we are done as any $i$ with $f(r^{(i)})\ne f(r)$ has $\pi^{-1}(i)\le \min_{\pi,\mu}(r)$.
Otherwise, the first binary search fails and this can happen only when  
 $i^*\in T_p$
  and $\pi(\min_{\pi,\mu}(y))\notin T_p$.
As a result, the $T^*$ we look for on line 5 must exist and it satisfies that 
  $i^*\notin T^*$.
Therefore, the second binary search always finds an $i$ such that $f(r^{(i)})=\overline{b}\ne f(r)$.
\end{proof}

We are now ready to prove Theorem \ref{theo:mainalg}:

\begin{proof}[Proof of Theorem \ref{theo:mainalg}]
The correctness of $\DL$ for the case when $f$ is $\eps$-far from
  decision lists with respect to $\calD$ follows
  from the following claim:

\begin{algorithm}[!t]
\caption{$\textbf{CheckDL}(f,\calD,\eps,r)$}\label{alg:CheckDL}
\begin{algorithmic}[1]
\Require {Oracle access to $f $, sampling access to $\mathcal{D}$, $\eps>0$ and $r\in \{0,1\}^n$}
\State Run $\MonotoneDL^*(g,\calD,\eps)$ with
  $g(x):=f(x\oplus r)$ and \textbf{accept} if it accepts
\State Otherwise, let $b=f(r)$ and $S$ be the set of strings queried by $\MonotoneDL^*$ on line 1 
\State Run $\Sketch(g,S)$ but stop on line 9 to obtain  sequence $X=(x^{(1)},\ldots,x^{(m)})$ and  $I_1,\ldots,I_k$ 

\State 
Let 
  $x^*$ denote the last  $\overline{b}$-string of $g$ in $X$ 
\State Run $\TestDL(f,\calD,\eps,r,x^*\oplus r)$; \textbf{accept} if it accepts and continue
  otherwise
\State Let $A$ (or $B$) denote the last interval of $\overline{b}$-strings (or $b$-strings) of $g$ in $X$  
\State Run $\textsc{IndexSearch}(f,r,x\oplus r)$ for every $x\in A$
\If{$\textsc{IndexSearch}(f,r,z\oplus r)=\nil$ for some $z\in A$}
  \State Pick any such $z\in A$
  \State Run $\TestDL(f,\calD,\eps,r,z\oplus r)$; \textbf{accept}
  if it accepts and \textbf{reject} if it rejects
\ElsIf{$\textsc{IndexSearch}(f,r,x\oplus r)=i$ for some $x\in A$ such that  
  $y_{i}=1$ for some $y\in B$}
  \State Pick any such $i$ and run $\TestDL(f,\calD,\eps,r,r^{(i)})$;
  \textbf{accept} if it accepts; \textbf{reject} if it rejects
\Else
  \State \textbf{reject}
\EndIf
\end{algorithmic}
\end{algorithm}

\begin{claim}\label{lem:general2}
\emph{Suppose $f$ is $\eps$-far from decision lists with respect to $\calD$.
Then for any string $r\in \{0,1\}^n$, $\textsc{CheckDL}(f,\calD,\eps,r)$ rejects with probability
  at least $1-o_n(1)$.}
\end{claim}

By a Chernoff bound followed by a union bound over the $100/\eps$
  rounds, $\DL$ accepts in this case with probability at most
$$
\frac{100}{\eps}\cdot \exp(-\Omega(\log (n/\eps)))< 1/3.
$$

\begin{proof}[Proof of Claim \ref{lem:general2}]
Given that $f$ is $\eps$-far from decision lists with respect to $\calD$, we have that 
  $g(x):=f(x\oplus r)$ is also $\eps$-far from decision lists and in particular, $\eps$-far from monotone decision lists.
As a result, line 1 of $\textsc{CheckDL}$ accepts with probability $o_n(1)$.

Other than line 1, $\textsc{CheckDL}$ accepts  when one of the 
  three executions of $\TestDL$ accepts but by Lemma \ref{lem:TestDL}, this also happens with probability $o_n(1)$.
This finishes the proof of the claim.
\end{proof}

The correctness of the algorithm for the case when $f$ is a decision list follows from
  the following claim about $\textsc{CheckDL}$:

\begin{claim}\label{lem:general1}
\emph{Suppose $f$ is a decision list represented by $(\pi,\mu,\nu)$, and 
  $r\in \{0,1\}^n$ satisfies (\ref{condonr}),
Then $\textsc{CheckDL}(f,\calD,\eps,r)$ accepts with probability
  at least $0.1$.}
\end{claim}

For this case, it follows by Claim \ref{lem:general1} that $\DL$ accepts with probability at least
$$
0.9\cdot \big(1-\exp (-\Omega( \log (n/\eps)) )\big)> 2/3.
$$

\begin{proof}[Proof of Claim \ref{lem:general1}]
Let $(\pi,\mu,\nu)$ be a representation of the input 
  decision list $f$ and
let $b=f(r)$ and $j^*=\min_{\pi,\mu}(r)\in [n+1]$.
Let $g$ be the function with $g(x):=f(x\oplus r)$.
Then $g$ is also a decision list and can be represented by
  $(\pi,\mu',\nu)$ with $\mu'=\mu\oplus r$.
Let
$g^*$ be the function defined as
$$
g^*(x):=\begin{cases} 
b & \text{if $g(x)=b$}\\
\overline{b} & \text{if $g(x)=\overline{b}$ and $\min_{\pi,\mu'}(x)<j^*$}\\
b & \text{if $g(x)=\overline{b}$ and $\min_{\pi,\mu'}(x)>j^*$}
\end{cases}
$$
Note that these three cases cover all possible $x$ since  $\min_{\pi,\mu'}(x)=j^*$ implies that $g(x)=b$.
Let $g^\dagger$ be the function defined as 
$$
g^\dagger(x):=\begin{cases}
\overline{b} & \text{if $g(x)=\overline{b}$}\\
b & \text{if $g(x)=b$ and $\min_{\pi,\mu'}(x)\le j^*$}\\
\overline{b} & \text{if $g(x)=b$ and $\min_{\pi,\mu'}(x)> j^*$}
\end{cases}
$$

We show that $g^*$ is a monotone decision list.
To see this is the case,
  it is easy to verify that $g^*$ can be represented as $(\pi,\mu',\nu^*)$
  with $\nu^*_j=b$ for all $j\ge j^*$.
This can be equivalently represented as $(\pi, 1^n,\nu^*)$,
  a monotone decision list (basically every $x$ with
  $\min_{\pi,\mu'}(x)\ge j^*$ always gets $g^*(x)=b$).
It is also easy to verify that $f^\dagger$ is a monotone decision list as well.
  
Consider $\MonotoneDL^*(g,\calD,\eps)$ on line $2$. 
Let $\calS$ denote the distribution over $N_{n,\eps}$-subsets
  of $\{0,1\}^n$ as the (random) set of queries made by $\MonotoneDL^*(g,\calD,\eps)$.
We show that either
$$
\Pr\Big[\MonotoneDL^*(g,\calD,\eps)\ \text{accepts}\Big]
$$
is at least $0.2$, or
$$
\Pr_{S\sim \calS}\left[\text{$S$ has an $x $ with $g(x)=\overline{b}$ and 
  $\min_{\pi,\mu'}(x)>j^*$, and a $y $ with $g(y)=b$ and $\min_{\pi,\mu'}(y)>j^*$}\right]
$$
is at least $0.2$.
This is because, if the second probability is at least $0.2$ then we are done.
Otherwise either with probability at least $0.4$, $\MonotoneDL^*(g,\calD,\eps)$ behaves exactly the same as if it runs on $(g^*,\calD)$ and it follows that the first probability is 
  at least $0.4-o_n(1)$ given that $g^*$ is a monotone decision list, or 
  with probability at least $0.4$, $\MonotoneDL^*(g,\calD,\eps)$ behaves exactly the same as if it runs on $(g^\dagger,\calD)$ and thus, the first probability is at least $0.4-o_n(1)$ as well.

It suffices to show that whenever $S$ contains an $x$ with $g(x)=\overline{b}$ such that 
  $\min_{\pi,\mu'}(x)>j^*$ and a $y$ with $g(y)=b$ such that $\min_{\pi,\mu'}(y)>j^*$, the rest of the procedure (lines 3--15 of $\textsc{CheckDL}$) accepts with
  probability at least $1-o_n(1)$.

First, if $x^*$ on line 4 satisfies $\min_{\pi,\mu'}(x^*)>\min_{\pi,\mu'}(r)$
then we are done because  the $\TestDL$~on line 5 accepts with probability $1-o_n(1)$.
So we assume below that $\min_{\pi,\mu'}(x^*)<\min_{\pi,\mu'}(r)$ (they cannot be equal because $x^*$ is a $\overline{b}$-string of $g$).

From this we can infer that $A$ must contain a string $x^\dagger$ such that 
  $g(x^\dagger)=\overline{b}$ and $\min_{\pi,\mu'}(x^\dagger)>j^*$ and 
  $B$ must contain a string $y^\dagger$ such that 
  $g(y^{\dagger})=b$ and $\min_{\pi,\mu'}(y^\dagger)>j^*$.
To see this is the case, we just use the following property of the sorted sequence $X$ together with the intervals $I_1,\ldots,I_k$: If $x$ is any
  string in $I_t$ for some $t$ such that $\min_{\pi,\mu'}(x)<j^*$ (such as the $x^*$ above),
  then all strings $z$ in $I_1,\ldots,I_{t-1}$ (i.e., intervals before $I_t$) must satisfy $\min_{\pi,\mu'}(z)<j^*$ as well (though this does not apply to strings
  in the same interval $I_t$ as $x$.)

Now if $\textsc{CheckDL}$  enters line 9, $z$ must satisfy 
  $\min_{\pi,\mu'}(z)>j^*$ by Lemma \ref{lem:IndexSearch} and thus,
  we are done because  $\TestDL$ on line 9 accepts with probability $1-o_n(1)$.
Otherwise we must have that $\textsc{IndexSearch}(f,r,x^{\dagger}\oplus r)$  
 returns $i^*=\pi(j^*)$
  by Lemma \ref{lem:IndexSearch}.
And checking among all indices returned by $\textsc{IndexSearch}(f,r,x\oplus r)$, $x\in A$, $i^*$ is the only index that can have $y_{i^*}=1$ for some $y\in B$
  (using that $y^\dagger \in B$ and that  all strings $y\in B$ satisfy
  $\min_{\pi,\mu'}(y)\ge j^*$ because $x^{\dagger}\in A$).
As a result, $\textsc{TestDL}$ must enter line 12 with $i=i^*=\pi(j^*)$ and thus,
  $\min_{\pi,\mu}(r^{(i)})>\min_{\pi,\mu}(r)$.
So we are done because $\TestDL$ on line 12 accepts with probability $1-o_n(1)$.
\end{proof}

It suffices to upperbound the number 
  of queries made by $\DL$, which is at most
$$
O\left(\frac{1}{\eps}\right)\cdot O\left(\log\left(\frac{n}{\eps}\right)\right)\cdot \tilde{O}\left(N_{n,\eps/2}\right)
= \tilde{O}\left(\frac{n^{11/12}}{\eps^3}\right).
$$
This finishes the proof of Theorem \ref{theo:mainalg}.
\end{proof}
\section{Birthday Paradox Lemmas}\label{sec:birthday}

\subsection{Bipartite Birthday Paradox}
We prove the bipartite birthday paradox lemma (\Cref{lem:bipbirthday}) in this subsection. We first present a proof of the classical birthday paradox (\Cref{lem:birth_classical_bip}) where we bound the probability of 2 people sharing a birthday. Note that for our purposes, we assume that each person's birthday follows a same (possibly non-uniform) distribution on a finite set.

\begin{lemma}\label{lem:birth_prelim_bip}
Let $\mathcal{D}$ be a probability distribution over a ground set $B$, and let $T\subseteq B$ be a subset such that $\mathcal{D}(T)\geq \varepsilon$ and $\mathcal{D}(\{t\})\geq p$ for each element $t\in T$, for some $\eps, p>0$. Let $X_1,\ldots,X_m$ be a sequence of $m$ samples drawn independently from $\mathcal{D}$ and let $S$ be the set they form. We have
$$\Pr_{S}\left[\mathcal{D}(S\cap T)\leq \min\left\{\frac{\varepsilon pm}{4},\frac{\varepsilon}{2}\right\}\right]\leq e^{-\varepsilon m/{ 16}}.$$
\end{lemma}
\begin{proof}
Let $S_{j}=\{X_{1},\dots,X_{j}\}\cap T$. For each $j\in [m]$, we define a random variable
$$Y_{j}=\begin{cases}
1 &\text{if }\mathcal{D}(T\setminus S_{j-1})<\varepsilon/2\\
\mathcal{D}(S_{j}\setminus S_{j-1}) &\text{otherwise.}
\end{cases}$$
We clearly have $\Pr[Y_{j}\geq p\mid X_{1},\dots,X_{j-1}]\geq \varepsilon/2$. So $p^{-1}(Y_{1}+\cdots+Y_{m})$ stochastically dominates the sum of $m$ independent Bernoulli random variables with parameter $\varepsilon/2$. By the multiplicative Chernoff bound, we have
$$\Pr\Big[p^{-1}(Y_{1}+\dots+Y_{m})\leq \varepsilon m/4\Big]\leq e^{-\varepsilon m/{ 16}}.$$
The lemma follows from the fact that  $Y_{1}+\dots+Y_{m}=\mathcal{D}(S\cap T)$ on the event $\mathcal{D}(S\cap T)\leq \varepsilon /2$. 
\end{proof}

\begin{lemma}[Classical birthday paradox]\label{lem:birth_classical_bip}
Let $n$ be a positive integer, and let $\mathcal{D}$ be a distribution over $[n+1]$. Let $p_{i}=\mathcal{D}(\{i\})$ for each $i\in [n+1]$, and let $\eps=p_{1}+\dots+p_{n}$. Let $S$ and $S'$ be two sets of samples of size $m$ and $m'$, respectively, drawn independently from $\mathcal{D}$, with $m$ and $m'$ satisfying $m\cdot m'\geq 100n/\varepsilon^{2}$ and that $m,m'\geq { 200}/\varepsilon$. Then with probability at least $0.99$, there exists an $i\in [n]$ such that $i$ appears in both $S$ and $S'$.
\end{lemma}
\begin{proof}
Let $T=\{i\in [n]:p_{i}\geq \varepsilon/(2n)\}$. Since $\mathcal{D}([n]\setminus T)<n\cdot \varepsilon/(2n)=\varepsilon/2$ and $\mathcal{D}([n])=\varepsilon$, it follows that $\mathcal{D}(T)\geq \varepsilon/2$. Applying \Cref{lem:birth_prelim_bip}, we have
$$\Pr_{S}\left[\mathcal{D}(S\cap T)\leq \min\left\{\frac{\varepsilon^{2}m}{16n},\frac{\varepsilon}{4}\right\}\right]\leq e^{-\varepsilon m/{ 32}}\leq \frac{1}{500}.$$
Furthermore, we have
\begin{align*}
&\quad \Pr_{S'}\left[S'\cap (S\cap T)= \emptyset\hspace{0.06cm}\middle|\hspace{0.06cm}\mathcal{D}(S\cap T)\geq \min\left\{\frac{\varepsilon^{2}m}{16n},\frac{\varepsilon}{4}\right\}\right]\\ 
&\hspace{1.5cm}\leq \left(1-\min\left\{\frac{\varepsilon^{2}m}{16n},\frac{\varepsilon}{4}\right\}\right)^{m'}\leq \exp\left(-\min\left\{\frac{\varepsilon^{2}m\cdot m'}{16n},\frac{\varepsilon m'}{4}\right\}\right)\leq \frac{1}{500}.
\end{align*}
Combining these two inequalities, we have that $\Pr_{S,S'}\left[S\cap S'=\emptyset\right]\leq 0.01$.
\end{proof}

The ingredient we need in order to extend the classical birthday paradox to the bipartite graph version stated below in Lemma \ref{lem:bipbirthday} is the well-known relationship between fractional matchings and vertex covers in bipartite graphs. 

\birthdayparadoxone*

\begin{proof}
Let $P$ be the function over $U\cup V$ given by $P(u)=\mu( u )$ for all $u\in U$ and $P(v)=\nu( v )$ for all $v\in V$. Consider the following linear program on variables $\{x_{e}: e\in E\}$:
\begin{align}
    \text{maximize\ \  } & \sum_{e\in E}x_{e}\nonumber\\
    \text{subject to\ \ } & \sum_{e\in E:w\in e}x_{e}  \le P(w) \label{eq:LP_cond},\quad\text{for each $w\in U\cup V$}\\
    & x_{e}\geq 0,\,\quad\text{for all $e\in E$}.\nonumber
\end{align}
It's easy to see that the linear program has a maximum. Let $(\lambda_{e})_{e\in E}$ be an optimal solution to the linear program. By linear programming duality, the optimal value $\sum_{e\in E}\lambda_{e}$ is equal to that of the following dual program over variables $\{y_w:w\in U\cup V\}$:
\begin{align*}
\text{minimize\ \  } & \sum_{w\in U}P(w)y_{w} \\
\text{subject to\ \  } & y_{u}+y_{v}\geq 1,\,\quad\text{for each edge $\{u,v\}\in E$} \\
& y_{w} \geq 0,\,\quad \text{for all $w\in U\cup V$}. 
\end{align*}
Since all extreme points of the fractional vertex-cover polyhedron (of any bipartite graph) are $0/1$-vectors (refer to \cite[Theorem 7.1.3]{lovasz2009matching}), the optimal value of the dual program is the weight of a vertex cover and is thus at least $\varepsilon$ by the assumption of the lemma. Hence we have $\sum_{e\in E}\lambda_{e}\geq \varepsilon$. 

We then apply the birthday paradox argument. Let $U=\{u_{1},\dots,u_{n}\}$. Consider the distribution $\mathcal{D}$ on $[n+1]$ defined by
\begin{align*}
\mathcal{D}( n+1)&=1-\sum_{e\in E}\lambda_{e}\quad\text{and}\quad
\mathcal{D}( i ) =\sum_{e\in E:u_i\in e}\lambda_{e},\quad\text{for each }i\in [n], 
\end{align*}
We can then construct a Markov kernel $K$ from $[n+1]$ to $V\cup\{\#\}$ such that $\mathcal{D}(\{i\})\cdot K(v|i)=\lambda_{\{u_{i},v\}}$ for all $i\in \supp(\mathcal{D})$ and $v\in V$. Also let $\varphi:[n+1]\rightarrow U\cup\{\#\}$ be the map that sends $i$ to $u_{i}$ for $i\in [n]$ and sends $(n+1)$ to $\#$. By \eqref{eq:LP_cond}, the push-forward measure $\mathcal{D}\circ\varphi^{-1}$ is dominated by $\mu$ on $U$, and the push-forward measure $\mathcal{D}K$ is dominated by $\nu$ on $V$. Therefore, using $\sum_{e\in E}\lambda_{e}\geq \varepsilon$, the desired conclusion follows from \Cref{lem:birth_classical_bip} and stochastic domination.
\end{proof}

\subsection{Hypergraph Birthday Paradox}

We will prove a hypergraph version of the bipartite graph birthday paradox in this subsection. We begin with a couple of calculus facts.

\begin{proposition}\label{lem:birth_prelim_hyp1}
\emph{Let $k$ be a positive integer. The function $f:(0,+\infty)\rightarrow(0,1)$ defined by $$f(x)=\ln\left(1-(1-e^{-x})^{k}\right)$$ is concave in $x$.}
\end{proposition}
\begin{proof}
We have
$$f'(x)=-\frac{ke^{-x}(1-e^{-x})^{k-1}}{1-(1-e^{-x})^{k}}=-\frac{k}{\sum_{\ell=0}^{k-1}(1-e^{-x})^{-d}},$$
which is decreasing in $x$.
\end{proof}
\begin{proposition}\label{lem:birth_prelim_hyp2}
\emph{Suppose $k$ is a positive integer and $a_{1},\dots,a_{k}$ are positive numbers such that $a_{1}+\dots+a_{k}=1$. The function $g:(0,+\infty)\rightarrow(0,1)$ defined by
$$g(x)=\left(1-\prod_{j=1}^{k}(1-e^{-x^{-a_{j}}})\right)^{x}$$
is non-increasing in $x$.}
\end{proposition}
\begin{proof}
Let $\rho:(0,1)\rightarrow(0,1)$ be defined by $\rho(x)=\frac{1-x}{x}\ln\left(\frac{1}{1-x}\right)$. By taking the derivative, it is easy to see that $\rho$ is a decreasing function on $(0,1)$. Define the function $h_{j}(x)=1-e^{-x^{-a_{j}}}$ for each $j\in [k]$, and let $h(x)=\prod_{j=1}^{k}h_{j}(x)$.
We have
$$\frac{h'(x)}{h(x)}=-\sum_{j=1}^{k}a_{j}x^{-a_{j}-1}\frac{e^{-x^{-a_{j}}}}{1-e^{-x^{-a_{j}}}}=-\frac{1}{x}\sum_{j=1}^{k}a_{j}\cdot \rho(h_{j}(x)).$$
For each $j\in [k]$ and $x\in(0,1)$, since $h_{j}(x)\geq h(x)$ we have $\rho(h_{j}(x))\leq \rho(h(x))$. So
$$\frac{h'(x)}{h(x)}\geq -\frac{1}{x}\cdot\rho(h(x))=-\frac{\rho(h(x))}{x}=\frac{(1-h(x))\ln(1-h(x))}{x\cdot h(x)}.$$
Therefore 
$$\frac{\d}{\d x}\ln(g(x))=\ln(1-h(x))-\frac{x\cdot h'(x)}{1-h(x)}\leq 0,$$
and thus $g$ is non-increasing.
\end{proof}

We use the preceding propositions to prove a birthday paradox where we bound the probability of $k$ people sharing a birthday. To make the random variables arising from the balls-into-bins procedure mutually independent, we use the Poisson approximation as in~\cite[Theorem 2.11]{mitzenmacher1996power}.

\begin{lemma}[Classical birthday paradox]\label{lem:birth_classical_hyp}
Let $n$ and $k$ be two positive integers.
Let $p_1,\ldots,p_{n+1}$ be nonnegative numbers that sum to $1/k$ and   $\mathcal{D}$ be the probability distribution over $[n+1]\times [k]$
with $\mathcal{D}\big(\{(i,j)\}\big)=p_{i}$ for $i\in [n+1]$ and $j\in [k]$. Let $\eps=p_{1}+\dots+p_{n}$ and $m$ be a positive integer with  $$m\geq \frac{10k n^{(k-1)/k}}{\eps}.$$ Draw $m$ independent samples from $\mathcal{D}$, and let $Y_{i,j}$ be the number of  $(i,j)$ drawn.   Then 
$$\Pr\Big[\forall\hspace{0.04cm} i\in [n], \exists \hspace{0.04cm}j\in [k]\text{ such that } Y_{ij}=0\Big]\leq \frac{1}{100}.$$
\end{lemma}
\begin{proof}
Let $(Z_{ij})_{i\in [n+1],\,j\in [k]}$ be independent random variables such that each $Z_{ij}$ follows the Poisson distribution with mean $mp_{i}$. It is well known that the conditional joint distribution of $(Z_{ij})_{i\in [n+1],\,j\in[k]}$ on the event $\left\{\sum_{i=1}^{n+1}\sum_{j=1}^{n}Z_{ij}=m\right\}$ is exactly the joint distribution of $(Y_{ij})_{i\in [n+1],\,j\in [k]}$. Let function $f:\mathbb{Z}_{\geq 0}^{(n+1)k}\rightarrow \{0,1\}$ be defined as
$$f\left((x_{ij})_{i\in [n+1],\,j\in [k]}\right)=\mathbbm{1}\big\{\forall i\in [n], \exists j\in [k]\text{ such that } x_{ij}=0\big\}.$$
Then since $\E\left[f\left((Y_{ij})_{i\in [n+1],\,j\in[k]}\right)\right]$ is decreasing in $m$, we have
\begin{align}
\E\left[f\left((Z_{ij})_{i\in [n+1],\,j\in[k]}\right)\right]&\geq \sum_{t=0}^{m}\E\left[f\left((Z_{ij})_{i\in [n+1],\,j\in[k]}\right)\middle|\sum\nolimits_{i,j}Z_{ij}=t\right]\cdot\Pr\left[\sum\nolimits_{i,j}Z_{ij}=t\right]\nonumber \\ 
&\geq  \E\left[f\left((Z_{ij})_{i\in [n+1],\,j\in[k]}\right)\middle|\sum\nolimits_{i,j}Z_{ij}=m\right]\cdot\Pr\left[\sum\nolimits_{i,j}Z_{ij}\leq m\right]\nonumber \\
&= \E\left[f\left((Y_{ij})_{i\in [n+1],\,j\in[k]}\right)\right]\cdot\Pr\left[\sum\nolimits_{i,j}Z_{ij}\leq m\right]. \label{eq:Poisson_approx}
\end{align}
Since $\sum_{i,j}Z_{ij}$ follows a Poisson distribution with mean $m$, it is easy to see (for example by the Berry-Esseen theorem) that 
\begin{equation}\label{eq:median_Poisson}
\Pr\left[\sum\nolimits_{i,j}Z_{ij}\leq m\right]\geq \frac{1}{4}.
\end{equation}
By independence between the variables $Z_{ij}$, we have
\begin{align}
\E\left[f\left((Z_{ij})_{i\in [n+1],\,j\in[k]}\right)\right]&=\prod_{i=1}^{n}\Pr\left[\exists j\in [k]\text{ such that } Z_{ij}=0\right]\nonumber \\
&=\prod_{i=1}^{n}\left(1-(1-e^{-mp_{i}})^{k}\right)\nonumber \\
&\leq \left(1-(1-e^{-\varepsilon m/n})^{k}\right)^{n}& (\text{by \Cref{lem:birth_prelim_hyp1}})\nonumber \\
&\leq \left(1-(1-e^{-10kn^{-1/k}})^{k}\right)^{n} & (\text{since $m\geq 10k\varepsilon^{-1}n^{(k-1)/k}$})\nonumber \\
&\leq 1-(1-e^{-10k})^{k}\leq \frac{1}{400} &(\text{by \Cref{lem:birth_prelim_hyp2}}). \label{eq:convexity}
\end{align}
The conclusion follows by combining \eqref{eq:Poisson_approx}, \eqref{eq:median_Poisson} and \eqref{eq:convexity}.
\end{proof}

The hypergraph version of \Cref{lem:bipbirthday} follows by a similar linear programming argument.

\birthdayparadoxtwo*

\begin{proof}
Consider the following linear program over variables $\{x_{e}:e\in E\}$:
\begin{align*}
    \text{maximize\ \ } & \sum_{e\in E}x_{e}\\
    \text{subject to\ \ } & \sum_{e\in E:v\in e}x_{e}\le \mu(v),\quad\text{for each $v\in V$}\\
    & x_{e}\geq 0,\,\quad\text{for all $e\in E$}.
\end{align*}
 Let  $(\lambda_{e})_{e\in E}$ be an optimal solution. 
Then 
$$C=\left\{v\in V:\sum_{e\in E:v\in e}\lambda_{e}=\mu(v)\right\}$$ must be a vertex cover of $H$. So by assumption
$$\varepsilon\leq \sum_{v\in C}\mu(v)\leq 
k\sum_{e\in E}\lambda_{e}.$$
By Carath\'{e}odory's theorem there exists another optimal solution $(\lambda'_{e})_{e\in E}$ to the linear program such that $\big|\{e\in E:\lambda'_{e}\neq 0\}\big|\leq |V|$. Furthermore, we still have $\sum_{e\in E}\lambda'_{e}=\sum_{e\in E}\lambda_{e}\geq \varepsilon/k$. 

We are now ready for the birthday paradox argument. Let $n$ be the size of the set $\{e\in E:\lambda'_{e}\neq 0\}$ and denote its elements by $e_{1},\dots,e_{n}$. There obviously exists a map $\varphi:[n+1]\times [k]\rightarrow V\cup\{\#\}$ that maps the set $\{(i,j):j\in [k]\}$ one-to-one onto the vertices of $e_{i}$ for each $i\in [n]$, and maps the set $\{(n+1,j):j\in [k]\}$ to $\{\#\}$. Considering the distribution $\mathcal{D}$ on $[n+1]\times [k]$ defined by
\begin{align*}
\mathcal{D}\big(\{(i,j)\}\big)&=\lambda'_{e_{i}},\text{ for all } i\in [n]\text{ and } j\in [k], \\
\mathcal{D}\big(\{(n+1,j)\}\big)&=\left(\frac{1}{k}-\sum_{i= 1}^{n}\lambda'_{e_{i}}\right),\text{ for all } j\in [k].
\end{align*}
Using $\sum_{i}\lambda'_{e_{i}}\geq \varepsilon/k$ and $n\leq |V|$, the lemma follows by \Cref{lem:birth_classical_hyp} and stochastic domination. 
\end{proof}
\section{An $\tilde{\Omega}(\sqrt{n})$ Lower Bound}\label{sec:lowerbound}

Our proof of the lower bound in \Cref{thm:lower_bound_wrapup} is based on Yao's minimax principle. We first construct two distributions $\calD_{YES}$ and $\calD_{NO}$ over function-distribution pairs $(f, \calD)$. For simplicity, we will assume that $n$ is a multiple of 16. 


\begin{itemize}
	\item $\calD_{NO}$: For each permutation $\pi$ over $[n]$, we define  the distribution $\mathcal{D}^{\mathsf{no}}_{\pi}$ to be the uniform distribution over
    \[\bigcup_{k=\frac{n}{8}}^{\frac{n}{4}-1}\big \{e_{\pi(4k+1)} \lor e_{\pi(4k+2)},\; e_{\pi(4k+2)} \lor e_{\pi(4k+3)},\; e_{\pi(4k+3)} \lor e_{\pi(4k+4)},\; e_{\pi(4k+4)} \lor e_{\pi(4k+1)} \big\}. \]
    We then define $\nu^{\mathsf{no}}\in \{0,1\}^{n+1}$ by setting 
    \[\nu^{\mathsf{no}}_{4k+1}=\nu^{\mathsf{no}}_{4k+3}=0, \; \nu^{\mathsf{no}}_{4k+2}=\nu^{\mathsf{no}}_{4k+4}=1\]
    for each $k\in \{\frac{n}{8},\frac{n}{8}+1,\dots,\frac{n}{4}-1\}$, and $\nu^{\mathsf{no}}_{i}=1$ for each $i\in \{1,2,\dots,\frac{n}{4}-1\}\cup\{n+1\}$. Now suppose $x\in\{0,1\}^{n}$. Let $i$ be the smallest number in the set $\pi^{-1}(\supp(x))$. If $i=4k+1$ for some $k\in \{\frac{n}{8},\frac{n}{8}+1,\dots,\frac{n}{4}-1\}$ and $$\pi^{-1}(\supp(x))\cap\{4k+1,4k+2,4k+3,4k+4\}=\{4k+1,4k+4\},$$ then set $f^{\mathsf{no}}_{\pi}(x)=1$. Otherwise, set $f^{\mathsf{no}}_{\pi}(x)$ so that it agree with the monotone decision list represented by the pair $(\pi,\nu^{\mathsf{no}})$. Note that we have
    \begin{equation}\label{eq:f-no-1}
		f^{\mathsf{no}}_{\pi}(e_{\pi(4k+4)} \lor e_{\pi(4k+1)}) = f^{\mathsf{no}}_{\pi}(e_{\pi(4k+2)}\lor e_{\pi(4k+3)})=1
    \end{equation}
    and
    \begin{equation}\label{eq:f-no-2}
        f^{\mathsf{no}}_{\pi}(e_{\pi(4k+1)} \lor e_{\pi(4k+2)}) = f^{\mathsf{no}}_{\pi}(e_{\pi(4k+3)}\lor e_{\pi(4k+4)})=0
    \end{equation}
	for each $k\in \{\frac{n}{8},\frac{n}{8}+1,\dots,\frac{n}{4}-1\}$. The final distribution $\mathcal{D}_{NO}$ is taken to be the distribution of the random pair $(f^{\mathsf{no}}_{\pi},\calD^{\mathsf{no}}_{\pi})$, where $\pi$ is a uniformly random permutation over $[n]$.
	
	
	\item $\calD_{YES}$: For each permutation $\pi$ over $[n]$, we define  the distribution $\mathcal{D}^{\mathsf{yes}}_{\pi}$ to be the uniform distribution over
    \[\bigcup_{k=\frac{n}{8}}^{\frac{n}{4}-1}\big \{e_{\pi(4k+1)} \lor e_{\pi(4k+2)},\; e_{\pi(4k+1)} \lor e_{\pi(4k+3)},\; e_{\pi(4k+2)} \lor e_{\pi(4k+4)},\; e_{\pi(4k+3)} \lor e_{\pi(4k+4)} \big\}. \]
    We then define $\nu^{\mathsf{yes}}\in \{0,1\}^{n+1}$ by setting 
    \[\nu^{\mathsf{yes}}_{4k+1}=\nu^{\mathsf{yes}}_{4k+4}=0, \; \nu^{\mathsf{yes}}_{4k+2}=\nu^{\mathsf{yes}}_{4k+3}=1\]
    for each $k\in \{\frac{n}{8},\frac{n}{8}+1,\dots,\frac{n}{4}-1\}$, and $\nu^{\mathsf{no}}_{i}=1$ for each $i\in \{1,2,\dots,\frac{n}{4}-1\}\cup\{n+1\}$. We then let $f^{\mathsf{yes}}_{\pi}(x)$ be the monotone decision list represented by the pair $(\pi,\nu^{\mathsf{yes}})$. The final distribution $\calD_{YES}$ is taken to be the distribution of the random pair $(f^{\mathsf{yes}}_{\pi},\mathcal{D}^{\mathsf{yes}}_{\pi})$, where $\pi$ is a uniformly random permutation over $[n]$.
 
	


\end{itemize}

By definition, we have that any function drawn from $\calD_{YES}$ will be a decision list. Correspondingly, functions from $\calD_{NO}$ will be far from any decision list.

\begin{lemma}
For any permutation $\pi$ over $[n]$ and for any linear threshold function $g:\{0,1\}^{n}\rightarrow\{0,1\}$, we have that 
	\[\Pr_{x \sim \calD^{\no}_{\pi}}[g(x) \not = {f}^{\no}_{\pi}(x)] \geq \frac{1}{4}. \]
\end{lemma}

\begin{proof}

By \eqref{eq:f-no-1} and \eqref{eq:f-no-2}, it suffices to show that for each $k\in \{\frac{n}{8},\frac{n}{8}+1,\dots,\frac{n}{4}-1\}$, we cannot have
\begin{equation}\label{eq:LTF-1}
g(e_{\pi(4k+4)} \lor e_{\pi(4k+1)}) = g(e_{\pi(4k+2)}\lor e_{\pi(4k+3)})=1
\end{equation}
and
\begin{equation}\label{eq:LTF-2}
g(e_{\pi(4k+1)} \lor e_{\pi(4k+2)}) = g(e_{\pi(4k+3)}\lor e_{\pi(4k+4)})=0
\end{equation}
simultaneously. Assume on the contrary that both \eqref{eq:LTF-1} and \eqref{eq:LTF-2} hold for some $k\in \{\frac{n}{8},\frac{n}{8}+1,\dots,\frac{n}{4}-1\}$. Since $g$ is an linear threshold function, $g$ agrees with a halfspace function $h:\mathbb{R}^{n}\rightarrow\{0,1\}$. Let $x_{0}$ be the vector in $\mathbb{R}^{n}$ defined by 
$$x_{0}=\frac{1}{2}\left(e_{\pi(4k+1)}+e_{\pi(4k+2)}+e_{\pi(4k+3)}+e_{\pi(4k+4)}\right).$$
By convexity of halfspaces, \eqref{eq:LTF-1} implies $h(x_{0})=1$ and \eqref{eq:LTF-2} implies $h(x_{0})=0$, a contradiction.
\end{proof}

Note that the first $n/2$ rules of each decision list $f^{\yes}_{\pi}$ from $\calD_{YES}$ are all set to be 1. Intuitively, this ensures that any query of a random string of large Hamming weight would very likely get an answer of 1 and reveal little information. More formally, we make the following definition.

\begin{definition}
Suppose $\ALG$ is a deterministic algorithm running on a pair $(f,\calD)$. If it has taken samples $s^{(1)},\dots,s^{(m)}$ and made sequential queries $q^{(1)},\dots,q^{(j)}$, we call the $j$-th query a ``large'' query if
$$\left|\supp(q^{(j)})\setminus\supp\left(\bigvee_{i=1}^{m}s^{(i)}\vee\bigvee_{\ell=1}^{j-1}q^{(\ell)}\right)\right|\geq 10\log n.$$
\end{definition}
\begin{lemma}
\label{lem:unseen-vx-bound}
	Suppose that $\ALG$ is a deterministic algorithm that distinguishes with probability $5/6$ between $\calD_{YES}$ and $\calD_{NO}$ using at most $m$ samples and $M$ queries, where $m+M\leq\frac{n}{30\log n}$. It then follows that there is an algorithm $\ALG'$ that distinguishes between $\calD_{YES}$ and $\calD_{NO}$ with probability $2/3$ such that $\ALG'$ never makes ``large'' queries. Moreover, $\ALG'$ uses at at most $m$ samples and $M$ queries.
\end{lemma}

\begin{proof}
	Since $\ALG$ is a deterministic algorithm, conditioned on the samples $S=(s^{(1)},\dots,s^{(m)})$ it receives, all subsequent queries it will possibly make can be represented by a decision tree. We perform a sequence of transformations to remove all ``large'' queries in the decision trees. Let $\ALG_0 := \ALG$. For each $j\in\{1,\dots,M\}$, we get $\ALG_{j}$ by removing all ``large'' queries at depth $j$ of the decision trees of $\ALG_{j-1}$ (the root query, i.e. $q^{(1)}$, is considered to be at depth 1) and assuming the answers to these queries are all $1$. Since the decision trees of $\ALG$ have depths at most $M$, it follows that $\ALG_{M}$ does not make any bad queries.     
	
	Now we claim that for each $j\in\{1,\dots,M\}$,
    \begin{equation}\label{eq:alg_convert_YES}
		\Pr_{\substack{(f,\calD)\sim \calD_{YES}\\ S\sim \calD}}[\ALG_{j}(f,S) \text{ accepts}] \geq \Pr_{\substack{(f,\calD)\sim \calD_{YES}\\ S\sim \calD}} [\ALG_{j-1}(f,S) \text{ accepts}] - n^{-1}. 	
    \end{equation}
    
    For any underlying pair $(f,\calD)$ and sequence of samples $S=(s^{(1)},\dots,s^{(m)})$, by definition $\ALG_{j}(f,S)\neq \ALG_{j-1}(f,S)$ only if $\ALG_{j-1}$ makes a ``large'' query at depth $j$ and the gets answer 0. Let $P=\left(S, q^{(1)}, f(q^{(1)}), \dots, q^{(j-1)},f(q^{(j-1)}), q^{(j)}\right)$ be a decision-tree path of $\ALG_{j-1}$ such that $|\supp(q^{(j)})\setminus T|\geq 10\log n$, where $T=\supp\left(\bigvee_{i=1}^{m}s^{(i)}\vee\bigvee_{\ell=1}^{j-1}q^{(\ell)}\right)$ (note that this means $q^{(j)}$ is a ``large'' query).   Conditioned on $\ALG_{j-1}$ following the decision tree path $P$, the distribution of the random set $\{\pi(1),\pi(2),\dots,\pi(n/2)\}\setminus T$ is symmetric over all indices in $[n]\setminus T$ (recall that $\pi$ is a uniformly random permutation over $[n]$ that determines the pair $(f^{\yes}_{\pi},\calD^{\yes}_{\pi})$). Furthermore, since $\ALG_{j-1}$ does not make ``large'' queries at depth less than or equal to $j-1$, it follows that $|T|\leq 2m+(j-1)\cdot 10\log n\leq (m+M)\cdot 10\log n\leq n/3$. Combining the observations in this paragraph, we have
    \begin{align*}
    &\qquad\Pr_{\substack{(f,\calD)\sim \calD_{YES}\\ S\sim \calD}}\left[f(q^{(j)})=0\middle|\ALG_{j-1}(f,S)\text{ follows }P\right]\\
    &\leq\Pr_{\substack{\pi\sim \frak{S}_{n}\\S\sim \calD^{\yes}_{\pi}}}\left[\supp(q^{(j)})\cap \left\{\pi(1),\pi(2),\dots,\pi(n/2)\right\}=\emptyset\middle|\ALG_{j-1}(f^{\yes}_{\pi},S)\text{ follows }P\right]\\
    &\leq\left(\frac{n/2}{n-|T|}\right)^{\left|\supp(q^{(j)})\setminus T\right|}\leq \left(\frac{3}{4}\right)^{10\log n}\leq n^{-1},
    \end{align*}
    where $\frak{S}_{n}$ is the set of all permutations over $[n]$. By the observation at the beginning of the paragraph and the arbitrariness of $P$, we conclude that
    $$\Pr_{\substack{(f,\calD)\sim \calD_{YES}\\ S\sim \calD}}\left[\ALG_{j}(f,S)\neq \ALG_{j-1}(f,S)\right]\leq n^{-1},$$
    proving the claim \eqref{eq:alg_convert_YES}. Summing \eqref{eq:alg_convert_YES} over $j$ yields
    \[\Pr_{\substack{(f,\calD)\sim \calD_{YES}\\ S\sim \calD}}[\ALG_{k}(f,S) \text{ accepts}] \geq \Pr_{\substack{(f,\calD)\sim \calD_{YES}\\ S\sim \calD}} [\ALG_{k}(f,S) \text{ accepts}] - n^{-1}\cdot M\geq \frac{2}{3}.\]
    In an entirely similar way, we have
    \[\Pr_{\substack{(f,\calD)\sim \calD_{NO}\\ S\sim \calD}}[\ALG_{k}(f,S) \text{ rejects}] \geq \Pr_{\substack{(f,\calD)\sim \calD_{NO}\\ S\sim \calD}} [\ALG_{k}(f,S) \text{ rejects}]- n^{-1}\cdot M\geq \frac{2}{3}.\qedhere\]
\end{proof}

The reason we want to exclude ``large'' queries is that we want to make sure that in any set of the form $\{\pi(4k+1),\pi(4k+2),\pi(4k+3),\pi(4k+4)\}$, no more than two indices will appear in support of any query. To make it formal, we make the following definition.

\begin{definition}
For any fixed sequence $S=(s^{(1)},s^{(2)},\dots,s^{(m)})$ of samples and fixed sequence $Q=(q^{(1)},\dots,q^{(M)})$ of queries in $\{0,1\}^{n}$, let $\Pi(S,Q)$ denote the set of all permutations $\pi$ over $[n]$ such that
$$\left|\supp\left(\bigvee_{i=1}^{m}s^{(i)}\vee\bigvee_{j=1}^{M}q^{(j)}\right)\cap\{\pi(4k+1),\pi(4k+2),\pi(4k+3),\pi(4k+4)\}\right|\leq 2$$
for each $k\in \{\frac{n}{8},\frac{n}{8}+1,\dots,\frac{n}{4}-1\}$. We further define two subsets of $\Pi_{S,Q}$
$$\Pi^{\no}(S,Q):=\left\{\pi\in\Pi(S,Q):\calD^{\no}_{\pi}(s^{(i)})>0\text{ for all }i\in [m]\right\}$$
and
$$\Pi^{\yes}(S,Q):=\left\{\pi\in\Pi(S,Q):\calD^{\yes}_{\pi}(s^{(i)})>0\text{ for all }i\in [m]\right\}.$$
\end{definition}

\begin{lemma}\label{lem:no_more_than_2}
Suppose $\ALG$ is a deterministic algorithm using at most $m$ samples and $M$ queries, where $m+M\leq \frac{\sqrt{n}}{100\log n}$ and never makes ``large'' queries. Let $\calQ(f,S)$ be the sequence of queries made by $\ALG$ when running $\ALG(f,S)$. We then have (here $\frak{S}_{n}$ denotes the set of all permutations over $[n]$)
$$\Pr_{\substack{\pi\sim \frak{S}_{n}\\S\sim \calD^{\yes}_{\pi}}}\Big[\pi\in\Pi\big(S,\calQ(f^{\yes}_{\pi},S)\big)\Big]\geq \frac{4}{5}.$$
\end{lemma}

\begin{proof}
For each $k\in\{\frac{n}{8},\frac{n}{8}+1,\dots,\frac{n}{4}-1\}$, we call the set $\{\pi(4k+1),\pi(4k+2),\pi(4k+3),\pi(4k+4)\}$ the $k$-th group of indices. The probability that no two samples are supported on the same group is
\begin{equation}\label{eq:alg_norep_init}
\Pr_{\substack{\pi\sim \frak{S}_{n}\\S\sim \calD^{\yes}_{\pi}}}\Big[\pi\in\Pi(S,\emptyset)\Big]=\frac{\frac{n}{8}\cdot\left(\frac{n}{8}-1\right)\cdot\dots\cdot\left(\frac{n}{8}-m+1\right)}{\left(\frac{n}{8}\right)^{m}}\geq \left(1-\frac{8m}{n}\right)^{m}\geq 1-\frac{8m^{2}}{n}\geq \frac{9}{10}.
\end{equation}
Let $\calQ(f^{\yes}_{\pi},S)=(q^{(1)},\dots,q^{(M)})$. Fix any $j\in\{1,\dots,M\}$ and let $T=\supp\left(\bigvee_{i=1}^{m}s^{(i)}\vee\bigvee_{\ell=1}^{j-1}q^{(\ell)}\right)$. For any fixed $k\in \{\frac{n}{8},\frac{n}{8}+1,\dots,\frac{n}{4}-1\}$, conditioned on $T$, all indices in $[n]\setminus T$ has the same probability of being in the $k$-th group of indices. It is easy to see from this fact that
\begin{equation}\label{eq:alg_norep_iter}
\Pr_{\substack{\pi\sim \frak{S}_{n}\\S\sim \calD^{\yes}_{\pi}}}\Big[\pi\in\Pi\big(S,(q^{(1)},\dots,q^{(j)})\big)\Big|\pi\in\Pi\big(S,(q^{(1)},\dots,q^{(j-1)})\big)\Big]\geq \left(1-\frac{4|T|}{n-|T|}\right)^{\left|\supp(q^{(j)})\setminus T\right|}.
\end{equation}
By the assumption that every query is not large, we always have $|\supp(q^{(j)})|\leq 10\log n$ and  $|T|\leq (m+M)\cdot 10\log n\leq \frac{\sqrt{n}}{10}$. Plugging into \eqref{eq:alg_norep_iter} and taking product over $j$ yields
$$\Pr_{\substack{\pi\sim \frak{S}_{n}\\S\sim \calD^{\yes}_{\pi}}}\Big[\pi\in\Pi\big(S,(q^{(1)},\dots,q^{(M)})\Big|\pi\in\Pi(S,\emptyset)\Big]\geq \left(1-\frac{5\sqrt{n}/10}{n}\right)^{\sqrt{n}/10}\geq \frac{9}{10}.$$
Taking product with \eqref{eq:alg_norep_init} yields the conclusion.
\end{proof}

Finally, the following lemma establishes a connection between the YES case and the NO case, a central ingredient in our lower bound proof.

\begin{lemma}\label{lem:YES_NO_connection}
Fix any sequence of Boolean values $W=(w_{1},\dots,w_{M})$. The set
$$\Pi^{\no}(S,Q,W):=\left\{\pi\in \Pi^{\no}(S,Q):f^{\no}_{\pi}(q^{(j)})=w_{j}\text{ for all }j\in [M]\right\}$$
has the same cardinality as the set
$$\Pi^{\yes}(S,Q,W):=\left\{\pi\in \Pi^{\yes}(S,Q):f^{\yes}_{\pi}(q^{(j)})=w_{j}\text{ for all }j\in [M]\right\}.$$
\end{lemma}

\begin{proof}
Let $T=\supp\left(\bigvee_{i=1}^{m}s^{(i)}\vee\bigvee_{j=1}^{M}q^{(j)}\right)$. It suffices to construct a bijective map $\pi\mapsto\pi'$ from $\Pi^{\yes}(S,Q)$ to $\Pi^{\no}(S,Q)$ such that for all $\pi\in\Pi^{\yes}(S,Q)$ we have

\begin{equation}\label{eq:bijection_property}
f^{\yes}_{\pi}(x)=f^{\no}_{\pi'}(x)\text{ holds for all }x\in\{0,1\}^{n}\text{ such that }\supp(x)\subset T.
\end{equation}

We fix a $\pi\in\Pi^{\yes}(S,Q)$. For each $k\in \{\frac{n}{8},\frac{n}{8}+1,\dots,\frac{n}{4}-1\}$, we divide into the following cases based on the intersection of $T$ with $\{\pi(4k+1),\pi(4k+2),\pi(4k+3),\pi(4k+4)\}$. Note that since $\pi\in\Pi(S,Q)$, we know that this intersection has cardinality at most 2.

Case 1: The intersection is a subset of $\{\pi(4k+1),\pi(4k+2),\pi(4k+4)\}$, or is a subset of $\{\pi(4k+2),\pi(4k+3)\}$. In this case, we define a permutation $\tau_{k}$ over $[n]$ to be the transposition of $4k+3$ and $4k+4$.

Case 2: The intersection is $\{\pi(4k+1),\pi(4k+3)\}$, or is $\{\pi(4k+3),\pi(4k+4)\}$. In this case, we define the permutation $\tau_{k}$ over $[n]$ to be the 3-cycle that maps $4k+3$ to $4k+1$, maps $4k+1$ to $4k+4$ and maps $4k+4$ to $4k+3$.

We now define $\pi'$ to be $\pi\circ \tau_{\frac{n}{8}}\circ\tau_{\frac{n}{8}+1}\circ\dots\circ\tau_{\frac{n}{4}-1}$. It is easy to see that $\pi\mapsto\pi'$ is a well-defined bijection between $\Pi^{\yes}(S,Q)$ and $\Pi^{\no}(S,Q)$. To verify \eqref{eq:bijection_property}, first note that $f^{\yes}_{\pi}(x)=1=f^{\no}_{\pi'}(x)$ if $\supp(x)\cap\{\pi(1),\pi(2),\dots,\pi(n/2)\}\neq\emptyset$. In the nontrivial case $\supp(x)\cap\{\pi(1),\pi(2),\dots,\pi(n/2)\}=\emptyset$ and $x\neq 0^{n}$, let $k$ be the smallest integer in $\{\frac{n}{8},\frac{n}{8}+1,\dots,\frac{n}{4}-1\}$ such that $\supp(x)\cap\{\pi(4k+1),\pi(4k+2),\pi(4k+3),\pi(4k+4)\}\neq\emptyset$. Let $y=x\wedge(e_{\pi(4k+1)}\vee e_{\pi(4k+2)}\vee e_{\pi(4k+3)}\vee e_{\pi(4k+4)})$. Since $\left|T\cap\{\pi(4k+1),\pi(4k+2),\pi(4k+3),\pi(4k+4)\}\right|\leq 2$, we have $|\supp(y)|\leq 2$. It is then straightforward to verify that $f^{\yes}_{\pi}(y)=f^{\no}_{\pi'}(y)$, and hence we have $f^{\yes}_{\pi}(x)=f^{\no}_{\pi'}(x)$.
\end{proof}

We are finally ready to prove our lower bound theorem.

\begin{theorem}\label{thm:lower_bound}
Let $\varepsilon\leq \frac{1}{4}$ be a fixed positive parameter. Assume that $\ALG$ is a randomized algorithm that for any function $f:\{0,1\}^{n}\rightarrow\{0,1\}$ and any distribution $\calD$ over $\{0,1\}^{n}$, takes at most $m$ samples from $\calD$ and at most $M$ queries to $f$, and achieves the following:
\begin{itemize}[itemsep = 0pt]
\item If $f$ is a monotone decision list then $\Pr_{\ALG,\, S\sim \calD}\left[\ALG(f,S)=1\right]\geq 5/6$; and,
\item If $f$ is $\varepsilon$-far from linear threshold functions under $\calD$ then $\Pr_{\ALG,\, S\sim \calD}\left[\ALG(f,S)=0\right]\geq 5/6$.
\end{itemize}
Then $m+M=\Omega(\sqrt{n}/\log n)$.
\end{theorem}
\begin{proof}
Assume on the contrary that $m+M\leq \frac{\sqrt{n}}{100\log n}$. By Yao's minimax principle, there is a deterministic algorithm that distinguishes between $\calD_{YES}$ and $\calD_{NO}$ with probability $5/6$, using no more than $m$ samples and $M$ queries. By \Cref{lem:unseen-vx-bound}, we may further assume that $\ALG$ does not make ``large'' queries (though with a slightly lower success probability of $2/3$). We consider the collection of all full decision tree path 
$$P=\Big(s^{1},\dots,s^{(m)},q^{1},f(q^{(1)}),\dots,q^{(M)},f(q^{(M)})\Big)=(S^{*},Q,W),$$
where $W$ collects the data $\left(f(q^{(1)}),\dots,f(q^{(M)})\right)$. By \Cref{lem:no_more_than_2}, we have (note that for $\ALG(g,S)$ to ``follow the path $P=(S^{*},Q,W)$'' it is necessary that $S=S^*$)
\begin{equation}\label{eq:path_no_more_than_2}
\sum_{P=(S^{*},Q,V)}\Pr_{\substack{\pi\sim\frak{S}_{n}\\ S\sim\calD^{\yes}_{\pi}}}\Big[\ALG(f^{\yes}_{\pi},S)\text{ follows }P,\text{ and }\pi\in\Pi(S,Q)\Big]\geq \frac{4}{5}.
\end{equation}
By \Cref{lem:YES_NO_connection}, for each fixed path $P=(S^{*},Q,W)$ we have
\begin{align}
&\qquad\Pr_{\substack{\pi\sim\frak{S}_{n}\\ S\sim\calD^{\yes}_{\pi}}}\Big[\ALG(f^{\yes}_{\pi},S)\text{ follows }P,\text{ and }\pi\in\Pi(S,Q)\Big]
=\frac{\left|\Pi^{\yes}(S^{*},Q,W)\right|}{n!}\cdot\left(\frac{1}{n/2}\right)^{m}\nonumber\\
&=\frac{\left|\Pi^{\no}(S^{*},Q,W)\right|}{n!}\cdot\left(\frac{1}{n/2}\right)^{m}
=\Pr_{\substack{\pi\sim\frak{S}_{n}\\S\sim\calD^{\no}_{\pi}}}\Big[\ALG(f^{\no}_{\pi},S)\text{ follows }P,\text{ and }\pi\in\Pi(S,Q)\Big].\label{eq:path_connect}
\end{align}
Note that if $\ALG(g_{1},S)$ and $\ALG(g_{2},S)$ follow the same decision-tree path, then clearly we have $\ALG(g_{1},S)=\ALG(g_{2},S)$. Combining \eqref{eq:path_no_more_than_2} and \eqref{eq:path_connect} thus yields
$$\left|\Pr_{\substack{(f,\calD)\sim \calD_{YES}\\ S\sim \calD}}\Big[\ALG(f,S)\text{ accepts}\Big]-\Pr_{\substack{(f,\calD)\sim \calD_{NO}\\ S\sim \calD}}\Big[\ALG(f,S)\text{ accepts}\Big]\right|\leq \frac{1}{5},$$
a contradiction.
\end{proof}
\section{Sample-Based Lower Bounds}\label{sec:sample_based}

Our proof of the sample-based lower bounds (\Cref{thm:sample_based_wrapup}) relies on the support-size-distinction framework of \cite{blais2021vc}.
\begin{definition}[Support-size distinction, \cite{blais2021vc}]\label{def:SSD}

Fix positive integer $n$ and real numbers $\alpha,\beta$ with $0<\alpha<\beta<1$. Define $\mathsf{SSD}(n,\alpha,\beta)$ to be the minimum integer $m$ such that there is a randomized algorithm $\ALG'$ that for any distribution $\mathcal{D}$ over $[n]$, takes $m$ samples $x^{(1)},\dots,x^{(m)}$ from $\mathcal{D}$ and achieves the following:
\begin{itemize}[itemsep = 0pt]
\item If $|\supp(\mathcal{D})|\leq \alpha n$ and $\mathcal{D}(\{x\})\geq 1/n$ for all $x\in\supp(\mathcal{D})$, then $$\Pr_{\ALG',\,x^{(1)},\dots,x^{(m)}}\left[\ALG'(x^{(1)},\dots,x^{(m)})=1\right]\geq 2/3.$$
\item If $|\supp(\mathcal{D})|\geq \beta n$ and $\mathcal{D}(\{x\})\geq 1/n$ for all $x\in\supp(\mathcal{D})$, then $$\Pr_{\ALG',\,x^{(1)},\dots,x^{(m)}}\left[\ALG'(x^{(1)},\dots,x^{(m)})=0\right]\geq 2/3.$$
\end{itemize}
\end{definition}

The following lower bound on $\mathsf{SSD}(n,\alpha,\beta)$ essentially follows from \cite{wu2019chebyshev} but only explicitly appears in \cite{blais2021vc}. 

\begin{lemma}[\cite{wu2019chebyshev,blais2021vc}]\label{lem:SSD_lower_bound}

There exists a constant $C$ such that, for any $\delta\geq C(\log n)^{1/2}n^{-1/4}$ and $\alpha,1-\beta\geq \delta$, 
$$\mathsf{SSD}(n,\alpha,\beta)=\Omega\left(\frac{n\delta^{2}}{\log n}\right).$$
\end{lemma}

We also need the following corollary of the Sauer-Shelah lemma.

\begin{lemma}[Lemma 2.7 of \mbox{\cite{blais2021vc}}]\label{lem:Sauer-Shelah}

Let $\mathcal{X}$ be a finite set, and let $\mathcal{H}$ be a class of functions $\mathcal{X}\rightarrow\{0,1\}$ with $\mathsf{VC}(\mathcal{H})=d$. If $T\subset \mathcal{X}$ has size $|T|\geq 4d$, then a uniformly random function $f:T\rightarrow\{0,1\}$ satisfies 
$$\Pr_{f}\left[\Big|\{x\in T:f(x)\neq h(x)\}\Big|\geq \frac{|T|}{100}\text{ for all } h\in\mathcal{H}\right]\geq 1-e^{-d/10}.$$
\end{lemma}

We then state our main structural lemma.

\begin{lemma}\label{lem:shatter}
Let $m=\frac{n}{10\log n}$, and let $y^{(1)},y^{(2)},\dots,y^{(m)}$ be independent samples from the uniform distribution on $\{y\in\{0,1\}^{n}:\|y\|_{1}=n-\log n\}$. Then the set $\{y^{(1)},y^{(2)},\dots,y^{(m)}\}$ is shattered by the class of monotone conjunctions with probability at least $1-O(n^{-1})$.
\end{lemma}
\begin{proof}
For each $k\in [m]$, define the event
$$\Gamma_{k}:=\left\{\supp(y^{(k)})\supset \bigcap\nolimits_{j\in [m]\setminus\{k\}}\supp(y^{(j)})\right\}.$$

For each $k\in [m]$, the intersection $\bigcap_{j\in [m]\setminus\{k\}}\supp(y^{(j)})$ has cardinality at least $n-(m-1)\cdot \log n\geq \frac{9}{10}n$. Since $y^{(k)}$ is independent with the samples $\{y^{(j)}\}_{j\in [m]\setminus\{k\}}$, we know that

$$\Pr[\Gamma_{k}]=\frac{\Big|\big\{y\in\{0,1\}^{n}:\|y\|_{1}=n-\log n\text{, and }\supp(y)\supset \bigcap_{j\in [m]\setminus\{k\}}\supp(y^{(j)})\big\}\Big|}{\Big|\big\{y\in\{0,1\}^{n}:\|y\|_{1}=n-\log n\big\}\Big|}\leq \frac{\binom{n/10}{\log n}}{\binom{n}{\log n}}\leq \frac{1}{10^{\log n}}.$$


It then follows from union bound that
$$\Pr\left[\bigcup\nolimits_{k\in [m]}\Gamma_{k}\right]\leq \frac{m}{10^{\log n}}=O(n^{-1}).$$

Let $\Gamma=\bigcup\nolimits_{k\in [m]}\Gamma_{k}$. It now suffices to show that, on the complement event $\Gamma^{c}$, the set $\{y^{(1)},\dots,y^{(m)}\}$ is always shattered by the class of monotone conjunctions. Let $b^{(1)},\dots,b^{(m)}\in \{0,1\}$ be any sequence of bits, and whenver $\Gamma^{c}$ holds we will define a monotone conjunction $f$ such that $f(y^{(k)})=b^{(k)}$ for each $k\in [m]$. Let 
$$E=\bigcap_{j\in [m]:\; b^{(j)}=1}\supp(y^{(j)}),$$
and for $x\in \{0,1\}^{n}$, let $f(x)=\bigwedge\nolimits_{i\in E}x_{i}$. For $k\in [m]$, it follows directly from the definition of $f$ that if $b^{(k)}=1$ then $f(y^{(k)})=b^{(k)}$. If $b^{(k)}=0$, then due to the condition of $\Gamma^{c}$ we have $$\supp(y^{(k)})\not\supset \bigcap_{j\in [m]:\; b^{(j)}=1}\supp(y^{(j)}).$$
By the definition of $f$, this implies that $f(y^{(k)})=0=b^{(k)}$.
\end{proof}

Now we are ready to prove a reduction from sample-based testing of monotone conjunctions to the problem of support-size distinction.

\begin{theorem}\label{thm:sample_based}
Let $\varepsilon\leq 1/250$ be a fixed positive number. Assume that $\ALG$ is a randomized sample-based algorithm that for any function $f:\{0,1\}^{n}\rightarrow\{0,1\}$ and any distribution $\mathcal{D}$ over $\{0,1\}^{n}$, when given as input a sequence $S=\{(y^{(k)},f(y^{(k)}))\}_{k\in [m]}$ such that $y^{(k)}\overset{\text{i.i.d.}}{\sim}\mathcal{D}$, achieves the following:
\begin{itemize}[itemsep = 0pt]
\item If $f$ is a monotone conjunction, then $\Pr_{\ALG,S}\left[\ALG(S)=1\right]\geq 5/6$; and,
\item If $f$ is $\varepsilon$-far from linear threshold functions under $\calD$ then $\Pr_{\ALG,S}\left[\ALG(S)=0\right]\geq 5/6$.
\end{itemize}
Then we must have $m=\Omega(n/\log^{3}n)$.
\end{theorem}

\begin{proof}
We will construct a support-size-distinguisher out of $\ALG$. Let $\mathcal{D}$ be a distribution over $[10n]$ with $\mathcal{D}(\{x\})\geq (10n)^{-1}$ for each $x\in\supp(\mathcal{D})$. Consider the following algorithm:

\begin{algorithm}[t!]
\caption{Reduction to support-size distinction}\label{alg:reduce_SSD}
\begin{algorithmic}[1]
\Require{Sample access to $\mathcal{D}$ over $[10n]$}
\State Draw samples $x^{(1)},x^{(2)},\dots,x^{(m)}\in [10n]$. 
\State Let $\varphi$ be a uniformly random function from $[10n]$ to the set $\{y\in\{0,1\}^{n}:\|y\|_{1}=n-\log n\}$. 
\State Let $f$ be a uniformly random function from $\{y\in\{0,1\}^{n}:\|y\|_{1}=n-\log n\}$ to $\{0,1\}$.
\State Construct a sequence $$S=\Big\{\Big(\varphi(x^{(k)}), f\left(\varphi(x^{(k)})\right)\Big)\Big\}_{k\in [m]}.$$
\State \textbf{return} $\ALG(S)$.
\end{algorithmic}
\end{algorithm}

\vspace{6pt}

We then show that \Cref{alg:reduce_SSD} can distinguish between the case where $\mathcal{D}$ has a small support-size and the case where it has a large support-size.

\textbf{Small-support-size case.} If $\mathcal{D}$ has support size at most $\frac{n}{10\log n}$, then by \Cref{lem:shatter}, the image $\varphi(\supp(\mathcal{D}))$ is shattered by the class of monotone conjunctions with probability at least $1-O(n^{-1})$. So with probability at least $1-O(n^{-1})$, the sequence $S$ constructed in step 4 of \Cref{alg:reduce_SSD} is consistent with a monotone conjunction. Thus by assumption we have $$\Pr_{\ALG,\,\varphi,\,f,\,x^{(1)},\dots,x^{(m)}}[\ALG(S)=1]\geq \frac{5}{6}(1-O(n^{-1}))\geq \frac{2}{3}.$$

\textbf{Large-support-size case.} If $\mathcal{D}$ has support size at least $5n$, then by Chernoff bound we have 
\begin{equation}\label{eq:Chernoff_varphi}
\Pr_{\varphi}\left[\big|\varphi(\supp(\mathcal{D}))\big|\leq 4n\right]\leq e^{-n/10}.
\end{equation}
For any fixed $\varphi$ such that $\big|\varphi(\supp(\mathcal{D}))\big|\geq 4n$, using \Cref{lem:Sauer-Shelah} with $\mathcal{X}:=\{y\in\{0,1\}^{n}:\|y\|_{1}=n-\log n\}$,  $\mathcal{H}:=\mathscr{LTF}$ (the class of linear threshold functions) and $T:=\varphi(\supp(\mathcal{D}))$, we have
$$\Pr_{f}\left[\Big|\{x\in T:f(x)\neq h(x)\}\Big|\geq \frac{|T|}{100}\text{ for all } h\in\mathscr{LTF}\right]\geq 1-e^{-n/10}.$$
For each element of $T$, by assumption we have $\mathcal{D}\circ\varphi^{-1}(\{x\})\geq (10n)^{-1}$. So the previous display implies that for any fixed $\varphi$ such that $\big|\varphi(\supp(\mathcal{D}))\big|\geq 4n$,
$$
\Pr_{f}\left[\mathsf{dist}_{\mathcal{D}\circ\varphi^{-1}}(f,\,\mathscr{LTF})\geq \frac{1}{250}\right]\geq 1-e^{-n/10}.
$$
Combining with \eqref{eq:Chernoff_varphi}, we see that 
$$\Pr_{f,\varphi}\left[\mathsf{dist}_{\mathcal{D}\circ\varphi^{-1}}(f,\,\mathscr{LTF})\geq \frac{1}{250}\right]\geq 1-2e^{-n/10}.$$
Since $\varphi(x^{(1)}),\dots,\varphi(x^{(m)})$ are independent samples from the push-forward distribution $\mathcal{D}\circ \varphi^{-1}$, by the guarantee of $\ALG$, the previous display implies 
$$\Pr_{\ALG,\,\varphi,\,f,\,x^{(1)},\dots,x^{(m)}}[\ALG(S)=0]\geq \frac{5}{6}(1-2e^{-n/10})\geq \frac{2}{3}.$$

In conclusion, we have shown that $\ALG$ solves the problem (\Cref{def:SSD}) on $[10n]$ with $\alpha =(100\log n)^{-1}$ and $\beta =1/2$ with $m$ samples. Using \Cref{lem:SSD_lower_bound} with $\delta=(100\log n)^{-1}$, we conclude that
\[m\geq \Omega\left(\frac{10n\cdot\delta^{2}}{\log (10n)}\right)=\Omega\left(\frac{n}{\log^{3}n}\right). \qedhere\]
\end{proof}

\bibliographystyle{alpha}
\bibliography{reference}

\begin{thebibliography}{SSBD14}

\bibitem[AC06]{AilonChazelle}
N.~Ailon and B.~Chazelle.
\newblock Information theory in property testing and monotonicity testing in
  higher dimension.
\newblock {\em Information and Computation}, 204(11):1704--1717, 2006.

\bibitem[BFH21]{blais2021vc}
Eric Blais, Renato {Ferreira Pinto Jr}, and Nathaniel Harms.
\newblock {VC} dimension and distribution-free sample-based testing.
\newblock In {\em Proceedings of the 53rd Annual ACM SIGACT Symposium on Theory
  of Computing}, pages 504--517, 2021.

\bibitem[CP22]{ChenPatel}
Xi~Chen and Shyamal Patel.
\newblock {\em Distribution-free Testing for Halfspaces (Almost) Requires PAC
  Learning}, pages 1715--1743.
\newblock 2022.

\bibitem[CX16]{chen2016tight}
Xi~Chen and Jinyu Xie.
\newblock Tight bounds for the distribution-free testing of monotone
  conjunctions.
\newblock In {\em Proceedings of the 27th Annual ACM-SIAM Symposium on Discrete
  Algorithms}, pages 54--71. SIAM, 2016.

\bibitem[DR11]{DolevRon}
E.~Dolev and D.~Ron.
\newblock Distribution-free testing for monomials with a sublinear number of
  queries.
\newblock {\em Theory of Computing}, 7(1):155--176, 2011.

\bibitem[GGR98]{goldreich1998property}
Oded Goldreich, Shari Goldwasser, and Dana Ron.
\newblock Property testing and its connection to learning and approximation.
\newblock {\em Journal of the ACM (JACM)}, 45(4):653--750, 1998.

\bibitem[GLR01]{GLR01}
David Guijarro, Victor Lavin, and Vijay Raghavan.
\newblock Monotone term decision lists.
\newblock {\em Theoretical Computer Science}, 259(1-2):549--575, 2001.

\bibitem[GS09]{GlasnerServedio}
D.~Glasner and R.~Servedio.
\newblock Distribution-free testing lower bound for basic boolean functions.
\newblock {\em Theory of Computing}, 5(10):191--216, 2009.

\bibitem[HK07]{HalevyKushilevitz}
S.~Halevy and E.~Kushilevitz.
\newblock Distribution-free property-testing.
\newblock {\em SIAM Journal on Computing}, 37(4):1107--1138, 2007.

\bibitem[HK08a]{HalevyKushilevitz2}
S.~Halevy and E.~Kushilevitz.
\newblock Distribution-free connectivity testing for sparse graphs.
\newblock {\em Algorithmica}, 51(1):24--48, 2008.

\bibitem[HK08b]{HalevyKushilevitz3}
S.~Halevy and E.~Kushilevitz.
\newblock Testing monotonicity over graph products.
\newblock {\em Random Structures \& Algorithms}, 33(1):44--67, 2008.

\bibitem[LP09]{lovasz2009matching}
L{\'a}szl{\'o} Lov{\'a}sz and Michael~D Plummer.
\newblock {\em Matching theory}, volume 367.
\newblock American Mathematical Soc., 2009.

\bibitem[Mit96]{mitzenmacher1996power}
Michael Mitzenmacher.
\newblock The power of two choices in randomized load balancing.
\newblock {\em PhD thesis, University of California at Berkeley}, 1996.

\bibitem[Riv87]{Rivest87}
Ronald~L. Rivest.
\newblock Learning decision lists.
\newblock {\em Machine Learning}, 2(3):229--–246, 1987.

\bibitem[SSBD14]{MLbook}
Shai Shalev-Shwartz and Shai Ben-David.
\newblock {\em Understanding machine learning: From theory to algorithms}.
\newblock Cambridge university press, 2014.

\bibitem[Tur93]{turan1993lower}
Gy{\"o}rgy Tur{\'a}n.
\newblock Lower bounds for pac learning with queries.
\newblock In {\em Proceedings of the sixth annual conference on Computational
  learning theory}, pages 384--391, 1993.

\bibitem[Val84]{Valiant}
L.G. Valiant.
\newblock A theory of the learnable.
\newblock {\em Communications of the ACM}, 27(11):1134--1142, 1984.

\bibitem[WY19]{wu2019chebyshev}
Yihong Wu and Pengkun Yang.
\newblock Chebyshev polynomials, moment matching, and optimal estimation of the
  unseen.
\newblock {\em Annals of Statistics}, 47(2):857--883, 2019.

\end{thebibliography}
\end{document}